\documentclass[11pt]{article}
\usepackage{amsmath,amsthm,amssymb}
\usepackage{graphicx}
\usepackage{enumerate}
\usepackage{natbib}
\usepackage{booktabs}
\usepackage{xcolor}
\usepackage{bm}
\usepackage{bbm}
\usepackage{hyperref}
\usepackage[title]{appendix}
\usepackage{setspace}
\usepackage{multibib}
\newcites{App}{Extra References for the Online Appendix} 

\newcommand{\bs}[1]{\boldsymbol{#1}}
\DeclareMathOperator{\tr}{tr}
\DeclareMathOperator{\E}{E}
\DeclareMathOperator{\var}{var}
\let\Pr\relax
\DeclareMathOperator{\Pr}{P}
\DeclareMathOperator{\cov}{cov}
\DeclareMathOperator{\AR}{AR}
\DeclareMathOperator{\corr}{corr}
\DeclareMathOperator{\lambdamax}{\lambda_{\mathrm{max}}}
\DeclareMathOperator{\lambdamin}{\lambda_{\mathrm{min}}}
\DeclareMathOperator{\rank}{rank}

\newtheorem{nondecr}{Theorem}
\newtheorem{theorem}[nondecr]{Theorem}
\newtheorem{nondecrl}{Lemma}
\newtheorem{lemma}[nondecrl]{Lemma}

\newtheorem{nondecra}{Assumption}
\newtheorem{assumption}[nondecra]{Assumption}
\newtheorem{nondecrc}{Corollary}
\newtheorem{corollary}[nondecrc]{Corollary}

\newtheorem{nondecre}{Example}
\newtheorem{example}[nondecre]{Example}

\definecolor{lightblue}{rgb}{0, 0.4, 0.6}
\hypersetup{
    breaklinks=true,
	colorlinks = true,
	citecolor = {lightblue},
	linkcolor = {lightblue},
	urlcolor = {lightblue}
}

\newcounter{expectRademachers}
\allowdisplaybreaks

\makeatletter
\let\@fnsymbol\@alph
\makeatother

\begin{document}

\title{Identification- and Many Moment-Robust Inference via Invariant Moment Conditions}
  \author{Tom Boot\thanks{
    Corresponding author. University of Groningen, Nettelbosje 2, 9747 AE Groningen, The Netherlands. \texttt{t.boot@rug.nl}.}
    \and
    Johannes W. Ligtenberg\thanks{Erasmus University Rotterdam, Burgemeester Oudlaan 50, 3062 PA Rotterdam, The Netherlands.   \texttt{ligtenberg@ese.eur.nl}.}}
  \maketitle

  \begin{abstract}
Identification-robust hypothesis tests are commonly based on the continuous updating GMM objective function.
When the number of moment conditions grows proportionally with the sample size, the large-dimensional weighting matrix prohibits the use of conventional asymptotic approximations and the behavior of these tests remains unknown. We show that the structure of the weighting matrix opens up an alternative route to asymptotic results when, under the null hypothesis, the distribution of the moment conditions satisfies a symmetry condition known as reflection invariance. We provide several examples in which the invariance follows from standard assumptions. Our results show that existing tests will be asymptotically conservative, and we propose an adjustment to attain nominal size in large samples. We illustrate our findings through simulations for various linear and nonlinear models, and an empirical application on the effect of the concentration of financial activities in banks on systemic risk.
\\
\textit{Keywords}: identification-robust inference, many moment conditions, continuous updating.\\
    \textit{JEL codes}: C12, C26.
\end{abstract}

\newpage
\section{Introduction}
Identification-robust inference in models defined by moment conditions is commonly based on the continuous updating (CU) objective function of \citet*{hansen1996finite}, either directly through \citeauthor{anderson1949estimation}'s (\citeyear{anderson1949estimation}) AR statistic, via its score as in \citeauthor{kleibergen2005testing}'s (\citeyear{kleibergen2005testing}) K statistic, or via a combination thereof \citep{moreira2003conditional,kleibergen2005testing,iandrews2016conditional}. While conventional tools yield the limiting distribution of these statistics when the number of moment conditions is small, their asymptotic behavior is unknown when the number of moments increases proportionally with the sample size. 

The main obstacle in establishing the limiting distribution of statistics based on the CU objective function is that when the number of moment conditions increases proportionally with the sample size, the weighting matrix for the moment conditions is large-dimensional and does not converge to a well-defined object. As \citet[p.\ 689]{newey2009generalized} write for the instrumental variables (IV) model ``\textit{If the number of instruments grows as fast as the sample size, the number of elements of the weight matrix grows as fast as the square of the sample size. It seems difficult to simultaneously control the estimation error for all these elements.}" In the linear IV model, a recent stream of the literature circumvents this problem by changing the weighting matrix so that it depends only on the instruments, and not on the second stage errors as in the CU objective function \citep{hausman2012instrumental,crudu2021inference,mikusheva2021inference,matsushita2020jackknife,dovi2022ridge,lim2022conditional}. However, for nonlinear models, no such results are available, and existing theory requires that the number of moment conditions grows at a slower rate than the sample size, e.g.\ \citet{han2006gmm} and \citet{newey2009generalized}.

In this paper, we focus on the CU-GMM objective function when the number of moment conditions increases proportionally to the sample size, similar to the many instrument sequences in linear IV models studied by \citet{bekker1994alternative}. We show how the difficulty posed by the weighting matrix can be circumvented when the distribution of the moment conditions evaluated at the true parameters is symmetric around zero, referred to as orthant symmetry by \citet{efron1969student} and as reflection invariance by \citet{bekker2008symmetry}. 
We discuss various examples where reflection invariance is implied by existing assumptions: a censored panel data model, median IV regression, linear IV with reflection invariant instruments or regression errors, and GMM-based tests for symmetry. 

We show that under reflection invariance, the limiting distribution of the suitably recentered and rescaled CU objective function, and hence that of the AR statistic, follows from known results on the limiting behavior of bilinear forms by \citet*{chao2012asymptotic}. These results indicate that the usual AR test will be conservative under many moments. We propose adjusted critical values that make the procedure asymptotically exact regardless of whether the number of moments is fixed or increasing. The results naturally extend to settings with clustered observations that in a linear IV setting were considered by \citet{ligtenberg2023inference}. For the heteroskedastic linear IV model, we also study the asymptotic behavior of the score of the CU objective function \citep{kleibergen2002pivotal,iandrews2016conditional,matsushita2020jackknife}. We derive a new central limit theorem for bilinear forms with an additional third-order term. We provide a set of assumptions under which the additional variance terms that enter under many instruments are negative, indicating that \citeauthor{kleibergen2005testing}'s (\citeyear{kleibergen2005testing}) score test will yield conservative inference. We also provide a variance estimator that consistently estimates the additional variance terms.

We assess the finite sample performance of the tests through simulations in two nonlinear settings and one linear setting. Specifically, for the nonlinear settings we consider the moment conditions by \citet{honore1992trimmed} for censored panel data models and the quantile IV moments developed by \citet{chernozhukov2009finite}. The simulations highlight the effect of (i) the identification strength, (ii) the number of instruments, and (iii) the effect of deviations from the invariance assumption. The simulation shows that, unlike conventional asymptotic approximations, the many moment-robust tests have close to nominal size control regardless of the identification strength and regardless of the number of moment conditions. Consequently, we document an increase in power when using the many-moment robust tests over the classic AR statistic. In the quantile IV application, the power appears nearly identical to that of the finite sample inference procedure of \citet{chernozhukov2009finite}. The quantile IV setting also provides a natural setting to test the sensitivity to violations of reflection invariance, as the invariance only holds at the median. We find that the proposed test is reasonably robust, although the rejection rates exceed the nominal size for quantiles distant from the median. 

Finally, we consider the heteroskedastic linear IV model under independent and clustered errors. For the case with independent errors, the many instrument correction to the AR and score statistic yields higher power relative to the standard AR and score test and the jackknife AR and score test \citep{mikusheva2021inference,matsushita2020jackknife} in settings with high heteroskedasticity, while differences are small in settings with low heteroskedasticity. Under clustering, the difference between the standard AR statistic and the many instrument robust version is larger compared to the setting with independent errors. With low heteroskedasticity, a clustered version of the jackknife AR appears more powerful. When we increase the degree of heteroskedasticity, this difference disappears.

We consider an empirical illustration by revisiting the study by \citet{langfield2016bank} on the effect of concentration of financial activity in the banking sector on systemic risk. As the data are left-censored, the authors employ the reflection invariant moment conditions from \citet{honore1992trimmed}. These moment conditions are clustered at the bank level, and the number of moment conditions is not negligible relative to the number of banks. The confidence regions from the AR statistic and the many instrument-robust AR statistic are wide, but, similar to \citet{langfield2016bank}, provide evidence that countries with a relatively large banking sector experience lower financial stability during financial crises than countries with a relatively small banking sector.
As expected, the many moment-robust AR statistic offers smaller confidence regions compared to the standard AR statistic.

\paragraph*{Structure} In Section \ref{sec:group} we discuss the reflection invariance assumption, and several examples where the assumption holds, in Section \ref{subsec:invar}. Theoretical results under independent moment conditions are given in Section \ref{subsec:independent} and under clustered moment conditions in Section \ref{subsec:cluster}. Score-based tests in the linear IV model are discussed in Section \ref{sec:score}. Section \ref{sec:MC} contains the Monte Carlo results for a censored tobit model, quantile IV and linear IV. The empirical application is presented in Section \ref{sec:application}. Section \ref{sec:conclusion} concludes.

\paragraph*{Notation}\label{page:notation} For a vector $\bs v$, $\bs D_{v}$ denotes the diagonal matrix with $\bs v$ on its diagonal. For a square matrix $\bs A$, let $\bs D_{A}=\bs A\odot \bs I$, where $\odot$ is the Hadamard product. We use $\dot{\bs A}= \bs A-\bs D_{A}$ for a matrix with all diagonal elements equal to zero. $\bs\iota$ indicates a vector of ones and $\bs e_{k}$ a vector with its $k^{\text{th}}$ entry equal to one and the remaining entries equal to zero.  Observations are indexed by $i=1,\ldots,n$ and $j=1,\ldots,n$, parameters of interest by $l=1,\ldots,p$, clusters by $h=1,\ldots, H$, time by $t=1,\ldots, T$. If we require additional indices, we subscript the relevant index, for example $i_{1},i_{2}\ldots$. Let $\bs a_{i}=\bs A'\bs e_{i}$ and $\bs a_{(l)}=\bs A\bs e_l$ denote the $i^\text{th}$ row and $l^\text{th}$ column of a matrix $\bs A$. For random variables $A$ and $B$, $A\overset{(d)}{=}B$ means that $A$ is distributionally equivalent to $B$. $A\overset{(\E)}{=}B$ means that $\E[A] = \E[B]$. $\E_{A}[\cdot]$ is the expectation over the distribution of the random variable $A$. $\rightarrow_{d}$ denotes convergence in distribution, $\rightarrow_{p}$ convergence in probability and $\rightarrow_{a.s.}$ almost sure convergence. $a.s.n.$ is short for with probability 1 for all $n$ sufficiently large. For a symmetric $n\times n$ matrix $\bs A$, $\lambdamin(\bs A)=\lambda_1(\bs A)\leq\ldots\leq\lambda_n(\bs A)=\lambdamax(\bs A)$ denote its eigenvalues. $C$ denotes a finite positive constant that can differ between appearances. $\mathbbm{1}\{\cdot\}$ is the indicator function.

\section{Continuous updating and invariant moment conditions}\label{sec:group}
We start with a GMM set-up where we have a vector of independent observations $\bs W_{i}$ for $i=1,\ldots,n$. Define the $p\times 1$ vector of parameters $\bs\beta$ and the vector of functions $\bs g_{i}(\bs\beta) = (g_{1}(\bs W_i,\bs\beta), \ldots,g_{k}(\bs W_i,\bs\beta))'$ for $k\geq p$. The true parameter $\bs \beta_0$ satisfies the moment conditions $\E[\bs g(\bs W_{i},\bs\beta_{0})]=\bs 0$. We stack the moment conditions in the $n\times k$ matrix $\bs G(\bs\beta) = [\bs g_{1}(\bs\beta),\ldots,\bs g_{n}(\bs\beta)]'$, with $\rank(\bs G(\bs\beta_{0}))=k$. Define the projector $\bs P(\bs\beta) = \bs G(\bs\beta)(\bs G(\bs\beta)'\bs G(\bs\beta))^{-1}\bs G(\bs\beta)'$. 

The continuous updating (CU) objective function introduced by \citet{hansen1996finite} can be written as
\begin{equation}\label{eq:objgen}
    Q(\bs\beta) = \frac{1}{n}\bs\iota'\bs P(\bs\beta)\bs\iota= \frac{k}{n} + \frac{1}{n}\sum_{i \neq j}\bs P_{ij}(\bs\beta).
\end{equation}
This objective function is closely related to the identification robust Anderson--Rubin (AR) GMM statistic, defined as
\begin{equation}\label{eq:AR}
    \AR(\bs\beta) = nQ(\bs\beta).
\end{equation}
For a fixed number of moment conditions $k$, the AR statistic is asymptotically $\chi^2(k)$ distributed when evaluated at $\bs\beta_{0}$. Extending this result to the case where the number of moment conditions grows proportionally with the sample size is challenging as the weighting matrix $\bs G(\bs\beta)'\bs G(\bs\beta)/n$ does not converge to a non-stochastic object.

\subsection{Invariance assumptions}\label{subsec:invar}

To derive the asymptotic behavior of \eqref{eq:AR} when the number of moment conditions $k$ grows proportional with $n$, we make the following invariance assumption.
\begin{assumption}[Reflection invariance]\label{ass:main}
    Let $\{r_{i}\}_{i=1}^{n}$ be a sequence of independent Rademacher random variables, so that $\Pr(r_{i}=1)=\Pr(r_{i}=-1)=1/2$, and define $\bs r=(r_{1},\ldots,r_{n})'$. Then, $\bs G(\bs\beta_{0})\overset{(d)}{=}\bs D_{r}\bs G(\bs\beta_{0})$.
\end{assumption}
Assumption~\ref{ass:main} is equivalent to stating that the distribution of $\bs g_{i}(\bs\beta_{0})$ is symmetric around 0. 
Note that this does not preclude the distribution of the moment conditions to differ across observations.
We now discuss several examples under which Assumption \ref{ass:main} holds. 

\begin{example}[Panel tobit model]\label{ex:Honore}
\citet{honore1992trimmed} considers a (censored or truncated) two period panel data model
\begin{equation}
    \begin{split}
        y_{i1}&=\alpha_i+\bs x_{i1}'\bs\beta+\varepsilon_{i1},\\
        y_{i2}&=\alpha_i+\bs x_{i2}'\bs\beta+\varepsilon_{i2}.
    \end{split}
\end{equation}
If the $\varepsilon_{it}$ are 
exchangeable conditional on $\bs x_{i1}$, $\bs x_{i2}$ and $\alpha_i$, then $\varepsilon_{i1}-\varepsilon_{i2}$ is symmetrically distributed. Consequently, $y_{i1}-y_{i2}$ is symmetric around $(\bs x_{i1}-\bs x_{i2})'\bs\beta$. \citet{honore1992trimmed} exploits this symmetry to derive valid moment conditions for $\bs\beta$, which are by the same mechanism reflection invariant. Any two time periods can be used to obtain moment conditions. By restricting the set so that we use each time period only once, we obtain reflection invariant moment conditions. The number of moment conditions then increases linearly with the product of the number of time periods and the dimension of the regressors. See Section \ref{ssec:MC Honore} for details on the moment conditions.
\end{example}

\begin{example}[Median IV regression]\label{ex:IVQR}
\citet{chernozhukov2009finite} estimate a quantile regression function $q$ that relates a dependent variable $y$ to possibly endogenous variables $\bs x$ and a uniformly distributed rank variable $\varepsilon$ as $y=q(\bs x,\varepsilon)$. Their Assumption 1 implies that for the median $\mathbbm{1}\{y\leq q(\bs x,1/2)\}$ is Bernoulli(1/2) distributed and hence reflection invariant once shifted. If $q$ can be parameterized with parameters $\bs\theta$ and there are exogenous instrumental variables $\bs z$, then we can estimate $\bs\theta$ via GMM with the invariant moment conditions $g_i(\bs\theta)=(1/2-\mathbbm{1} \{y_i\leq q(\bs x_i,1/2;\bs\theta)\})\bs\Psi(\bs z_i)$, where $\bs\Psi(\bs z_i)$ is some function of the instruments. The number of instrumental variables can be large in itself, or after using power series or splines as approximating functions for $\bs\Psi(\cdot)$ \citep{newey1990efficient}.
\end{example}

\begin{example}[Heteroskedastic IV with reflection invariant instruments or second stage errors]\label{ex:linIV}  Consider the standard linear IV model \begin{align}
	y_{i} &= \bs x_{i}'\bs\beta_{0}+ \varepsilon_{i},\quad
	\bs x_{i} =\bs \Pi'\bs z_{i} + \bs\eta_{i},\label{eq:model1}
\end{align}
with $\bs x_{i}$ the endogenous regressors, $\bs\beta_{0}$ a $p\times 1$ vector, $\bs z_{i}$ the instruments, $\bs \Pi$ a $k\times p$ matrix. We assume an intercept has been projected out and write the model without exogenous control variables. Supplementary Appendix \ref{app:controls} shows that this does not affect the results if the number of the controls increases slower than the sample size. The moment conditions are $\bs g_{i}(\bs\beta) = \bs z_{i}(y_{i}-\bs x_{i}'\bs\beta)$ and are reflection invariant under $H_{0}\colon\bs\beta=\bs\beta_0$ if $\bs z_{i}$ and $\varepsilon_{i}$ are independent and at least one of them is reflection invariant. Examples of invariant instruments include sex-at-birth instruments \citep{angrist1998children} and randomized experiments with equal assignment probability to treatment and control groups. This leads to many instruments when these are interacted with control variables \citep{dupas2018banking,dube2020queens}.
\end{example}

\begin{example}[GMM based tests for symmetry]\label{ex:symmtest}
\citet{bontemps2005testing} use Stein's equation and GMM to test for normality of a variable. Let $H_{i}$ be the Hermite polynomial of order $i$ recursively defined as $H_0(x)=1$, $H_1(x)=x$ and $H_i(x)=[xH_{i-1}(x)-\sqrt{i-1}H_{i-2}(x)]/\sqrt{i}$, $i>1$. \citet{bontemps2005testing} show that for $X\sim N(0,1)$ the moment conditions $\E(H_i(X))=0$ hold. If $X$ is symmetrically distributed but not necessarily normal, then these moment conditions still hold for $i$ odd and are reflection invariant. The number of moment conditions is large if many orders are used.
\end{example}
There are also settings where invariance cannot be argued based on standard assumptions and one wishes to test this assumption. As the reflection invariance is part of the specification, an empty confidence set could point to a failure of this assumption. When the moment conditions are of a product form, such as in the linear IV $\boldsymbol{g}_{i}(\boldsymbol{\beta_{0}}) = \boldsymbol{z}_{i}(y_{i}-\boldsymbol{x}_{i}'\boldsymbol{\beta}_{0}) = \boldsymbol{z}_{i}\varepsilon_{i}$, reflection invariance is testable if it comes from the observable component, e.g.\ the instrument vector $\boldsymbol{z}_{i}$. Whether testing the unobserved component, e.g.\ $\varepsilon_{i}$ in the linear IV, is possible in combination with identification-robust inference is an important question for further research.

\subsection{Continuous updating under reflection invariance}\label{subsec:independent}
Under Assumption \ref{ass:main} we can relate the distribution of the CU objective function with a similar function written in terms of Rademacher random variables. Define, $Q_{r}(\bs\beta_{0})= \bs r'\bs P(\bs\beta_{0})\bs r/n$. Then,
\begin{equation*}
    Q(\bs\beta_{0}) \overset{(d)}{=} Q_{r}(\bs\beta_{0}).
\end{equation*}
As 
$Q(\bs\beta_{0})$ and $Q_{r}(\bs\beta_{0})$ are distributionally equivalent, it suffices to analyze the asymptotic distribution of $Q_{r}(\bs\beta_{0})$ to obtain the asymptotic distribution of $Q(\bs\beta_{0})$. Likewise, we analyze $\AR_{r}(\bs\beta_{0})=nQ_{r}(\bs\beta_{0})$ to establish the asymptotic distribution of the AR statistic defined in \eqref{eq:AR}. 

Conditioning on $\mathcal{J}=\{\bs g_{i}(\bs\beta_{0})\}_{i=1}^{n}$, the only randomness in $\AR_r(\bs\beta_{0})$ comes from the Rademacher random variables. Under the following assumptions, we can directly apply the CLT for bilinear forms by \citet{chao2012asymptotic} to obtain the asymptotic distribution of the AR statistic under many instrument sequences.
\begin{assumption}\label{ass:AR} 
With probability 1 for all sufficiently large $n$: (a)  $\rank[\bs P(\bs\beta_{0})]=k$, 
(b) $\sigma_{n}^2(\bs\beta_{0})>1/C$ where
\begin{equation}\label{eq:sigman}
    \sigma_{n}^2(\bs\beta_{0}) = \frac{2}{k}\sum_{i\neq j} P_{ij}(\bs\beta_{0})^2.
\end{equation}
\end{assumption}
Part (a) excludes any redundant moment conditions.  Part (b) bounds the variance of the scaled AR statistic away from zero and is required to apply the central limit theorem provided in Lemma A2 by \citet{chao2012asymptotic}. A stronger alternative assumption that can be found in the many instrument literature is that $P_{ii}(\bs\beta_{0})\leq C<1$ for $i=1,\ldots,n$, see e.g.\ \citet{hausman2012instrumental}, \citet{bekker2015jackknife} and \citet{anatolyev2019many}. The following result follows from Lemma A2 in \citet{chao2012asymptotic} and its proof.
\begin{theorem}\label{corr:AR}
	 Under Assumptions \ref{ass:main} and \ref{ass:AR}, when $k\rightarrow\infty$ as $n\rightarrow\infty$, $(k\sigma_{n}^2(\bs\beta_{0}))^{-1/2}(\AR_{r}(\bm \beta_{0})-k)\rightarrow_{d} N(0,1)$ $a.s.$ As a consequence $(k\sigma_{n}^2(\bs\beta_{0}))^{-1/2}(\AR(\bs \beta_{0})-k) \rightarrow_{d} N(0,1).$
\end{theorem}

Theorem \ref{corr:AR} shows that the AR statistic needs to be shifted and scaled to have a well-defined asymptotic distribution. A similar result is obtained by \citet{anatolyev2011specification} for the AR statistic in a homoskedastic IV model with many instruments. 
While Theorem \ref{corr:AR} requires $k\rightarrow\infty$, we can achieve uniform inference across $k$ by testing based on the quantiles of the distribution of $Z=(2k)^{-1/2}(Z_{1}-k)$ where $Z_{1}\sim \chi^2(k)$. When $k$ is fixed, $\sigma_{n}^2(\bs\beta_{0})\rightarrow_{p} 2$, and hence, we compare $\AR(\bs\beta_{0})$ against the quantiles of a $\chi^2(k)$ distribution. When $k$ increases, the quantiles of $Z$ approach that of the standard normal distribution and Theorem \ref{corr:AR} applies. 

Theorem~\ref{corr:AR} implies that the usual AR test will be conservative at conventional significance levels. The proof implies that we reject a value of $\bs\beta$ if \begin{equation*}
    \AR(\bs\beta)>\chi^2(k)_{1-\alpha}-[\chi^2(k)_{1-\alpha}-k]\cdot \bigg[1-\bigg(1-k^{-1}\sum_{i=1}^{n}P_{ii}(\bs\beta)^2\bigg)^{1/2}\bigg].
\end{equation*}
If $\alpha< 0.3$, we have that $\chi^2(k)_{1-\alpha}>k$ for all $k$. We then see that the usual AR-test will reject less often relative to the many moment robust test. 
This result does not contradict the fact that both procedures are asymptotically size correct under a fixed number of moments, as in that case  $k^{-1}\sum_{i=1}^{n}P_{ii}(\bs\beta)^2\rightarrow_{p}0$.

\subsection{Clustered moment conditions}\label{subsec:cluster}
Assumption \ref{ass:main} implies that under the null, randomly flipping the sign of the moment conditions of one observation, does not alter the distribution of the moment conditions of another observation. Many studies however, allow for clustered dependence between observations. 
If instead of Assumption~\ref{ass:main}, we assume that the moment conditions of all observations within a cluster are jointly reflection invariant, we can generalize the AR statistic from above to allow for clustered data by taking similar steps as in the linear IV model analyzed in
\citet{ligtenberg2023inference}. 

Assume that the $n$ observations can be grouped in $H$ clusters with $n_h$ observations in cluster $h=1,\dots,H$.  Let $[h]$ denote all observations from cluster $h$ and for any $n\times m$ matrix $\bs A$ write $\bs A_{[h]}$ as the $n_h\times m$ submatrix of $\bs A$ with only the rows in $[h]$ selected. Then we adapt Assumption \ref{ass:main} as follows.
\begin{assumption}\label{ass:main cluster}
    Let $\{r_{h}\}_{h=1}^{H}$ be a sequence of independent Rademacher random variables. For all $h=1,\dots,H$ it holds that $\bs G(\bs\beta_{0})_{[h]}\overset{(d)}{=}r_h\bs G(\bs\beta_{0})_{[h]}$.
\end{assumption}

The CU objective function \eqref{eq:objgen} can be easily adapted to clustered dependence by summing the moment conditions of all observations within a cluster. This allows to write Assumption \ref{ass:main cluster} in a form that resembles Assumption \ref{ass:main}. Let $\tilde{\bs G}(\bs\beta)_{h}=\bs\iota_{n_h}'\bs G(\bs\beta_0)_{[h]}$ be the summed moment conditions, which we stack in the $H\times k$ matrix $\tilde{\bs G}(\bs\beta_0)$. Then, under Assumption \ref{ass:main cluster}, $\tilde{\bs G}(\bs\beta_0)\overset{(d)}{=}\bs D_{\tilde{r}}\tilde{\bs G}(\bs\beta_0)$, where $\tilde{\bs r}=(r_1,\dots,r_H)'$. 
We then define $\tilde{\bs P}(\bs\beta_0)=\tilde{\bs G}(\bs\beta_0)(\tilde{\bs G}(\bs\beta_0)'\tilde{\bs G}(\bs\beta_0))^{-1}\tilde{\bs G}(\bs\beta_0)'$. The AR statistic is now given by $\tilde{\text{AR}}(
\bs\beta) = \bs\iota_{H}'\tilde{\bs P}(\bs\beta)\bs\iota_{H}$. To derive its distribution, make the following assumption similar to Assumption \ref{ass:AR}.

\begin{assumption}\label{ass:AR cluster} 
With probability 1 for all sufficiently large $n$: (a)  $\rank[\tilde{\bs P}(\bs\beta_{0})]=k$, 
(b) $\tilde{\sigma}_{n}^2(\bs\beta_{0})>1/C$ where
\begin{equation}\label{eq:sigman cluster}
    \tilde{\sigma}_{n}^2(\bs\beta_{0}) = \frac{2}{k}\sum_{h_1 \neq h_2} \tilde{P}_{h_1 h_2}(\bs\beta_{0})^2.
\end{equation}
\end{assumption}
Part (a) is ensured by full column rank of $\bs G(\bs\beta_0)$, which also implies that $\rank[\bs P(\bs\beta_0)]=k$. 
Assumptions \ref{ass:main cluster} and \ref{ass:AR cluster} then yield the following extension of Theorem \ref{corr:AR} to clustered data.

\begin{corollary}
    Under Assumptions \ref{ass:main cluster} and \ref{ass:AR cluster}, when $k\rightarrow\infty$ as $H\rightarrow\infty$, $(k\tilde{\sigma}_{n}^2(\bs\beta_{0}))^{-1/2}(\tilde{\AR}(\bs \beta_{0})-k) \rightarrow_{d} N(0,1)$.
\end{corollary}
We note that under clustering, many moment sequences are defined as a number of moments that grows proportional to the number of clusters. As such, the corrections can become relevant for a relatively small number of moments even in large data sets.

\subsection{Extension: score-based tests in linear IV}\label{sec:score}
The fact that the standard AR test is conservative under many moment sequences and reflection invariance as shown in Section \ref{subsec:independent}, raises the question whether the same is true for a statistic based on the score of the CU objective function given in \eqref{eq:objgen}. To gain insight into this question, we specialize to the linear IV model with heteroskedasticity from Example \ref{ex:linIV}. We then consider the application of Assumption \ref{ass:main} to analyze the score statistic and show that it is indeed conservative under many moment sequences and reflection invariance.

The following notation will be convenient below: for some $\bs\beta$, not necessarily equal to $\bs\beta_{0}$, $\varepsilon_{i}(\bs\beta)=y_{i}-\bs x_{i}'\bs\beta$ and $\bs\varepsilon(\bs\beta) = (\varepsilon_{1}(\bs\beta),\ldots,\varepsilon_{n}(\bs\beta))'$. The model \eqref{eq:model1} is accompanied by the following assumptions.
\begin{assumption}\label{ass:model} 
    (a) Conditional on $\bs Z$, $\{\varepsilon_{i},\bs{\eta}_{i}'\}_{i=1}^{n}$ is independent, with mean zero and $\E[(\varepsilon_{i},\bs\eta_{i}')'(\varepsilon_{i},\bs\eta_{i}')|\bs Z] = \bs \Sigma_{i}=\begin{array}{cccc} (\sigma_{i}^2 & \bs\sigma_{12i}'; & \bs\sigma_{12i}&\bs\Sigma_{22i})\end{array}$ (b) $0<C^{-1}\leq \lambda_{\min}(\bs\Sigma_{i})\leq\lambda_{\max}(\bs\Sigma_{i})\allowbreak\leq C<\infty $ a.s., (c) For all $i$, $\E[\varepsilon_{i}^4|\bs Z]\leq C<\infty$ a.s.\ and $\E[\Vert\bs\eta_{i}\Vert^4|\bs Z]\leq C<\infty$ a.s. (d) Conditional on $\bs Z$, $\bs\varepsilon\overset{(d)}{=}\bs D_{r}\bs\varepsilon$ with $\bs D_{r}$ as in Assumption \ref{ass:main}.
\end{assumption}
Part (a)--(c) are similar to those made by \citet{crudu2021inference}, \citet{mikusheva2021inference} and \citet{matsushita2020jackknife} for their jackknife tests. Part (d) ensures that the moment conditions are reflection invariant. 

To obtain the limiting distribution of the first order conditions of the CU objective function, we make the following assumption on the IV model in \eqref{eq:model1}.
\begin{assumption}\label{ass:decomposev}
    Consider $\bs\eta_{i}$ and $\varepsilon_{i}$ as in \eqref{eq:model1}. Then, $
	\bs\eta_{i} = \varepsilon_{i}\bs a_{i} + \bs u_{i},$
where $\bs a_{i}=\bs\sigma_{21i}/\sigma_{i}^2$, and, conditional on $\bs Z$, $\{\bs u_{i},\varepsilon_{i}\}$ are mutually independent. 
\end{assumption}
This assumption can also be found in \citet{bekker2003finite}, \citet{dokotchatoka2020exogeneity}, and \citet{fraizier2024weak}, but it is nevertheless a strong assumption. It is for example satisfied if $(\varepsilon_{i},\bs\eta_{i}')$ is multivariate normal. We use Assumption \ref{ass:decomposev} to write
\begin{equation}\label{def:barxi}
\bs x_i=\bar{\bs x}_i+\varepsilon_i\bs a_i, \quad \bar{\bs x}_i=\bar{\bs z}_i+\bs u_i,
\end{equation}
for $\bar{\bs z}_{i}=\bs z_{i}'\bs\Pi$ and stack these in the matrix $\bar{\bs Z}$. 
Define $\bs{x}_{(l)}$ as the column vector $(x_{1h},\ldots,x_{nh})'$. The score of a half times the CU objective function is
\begin{equation*}\label{eq:foc}
	S_{l}(\bs\beta) = \frac{1}{2}\frac{\partial Q(\bs\beta)}{\partial \beta_{l}}= -\frac{1}{n}\bs x_{(l)}'(\bs I-\bs D_{P(\beta)\iota})\bs Z(\bs Z'\bs D_{\varepsilon(\beta)^2}\bs Z)^{-1} \bs Z'\bs\varepsilon(\bs\beta).
\end{equation*}
Under Assumption \ref{ass:main}, and using the decomposition of $\bs\eta_i$ in Assumption \ref{ass:decomposev} we find that, conditional on $\bs Z$, $S_{l}(\bs\beta_{0})\overset{(d)}{=}S_{l,r}(\bs\beta_{0})$, with
\begin{equation}\label{eq:foc_r}
\begin{split}
        S_{l,r}(\bs\beta_{0})	&=-\frac{1}{n}\bar{\bs x}_{(l)}'\bs Z(\bs Z'\bs D_{\varepsilon^2}\bs Z)^{-1} \bs Z'\bs D_{\varepsilon}\bs r  -\frac{1}{n}\bs r'\bs D_{a_{(l)}}\bs P\bs r + \frac{1}{n}\bs r'\bs P\bs D_{a_{(l)}}\bs P\bs r\\
        &\quad + \frac{1}{n}\bs r'\bs P\bs D_{r}\bs D_{\bar{x}_{(l)}}\bs Z(\bs Z'\bs D_{\varepsilon^2}\bs Z)^{-1} \bs Z'\bs D_{\varepsilon}\bs r,
        \end{split}
\end{equation}
where $\bs{a}_{(l)} = (a_{1l},\ldots, a_{nl})'$ and $\bs P = \bs P(\bs\beta_{0})$.

Note that the limiting distribution of the score is not covered by the results in \citet{chao2012asymptotic} as it contains products of Rademacher variables up to order three. To describe the joint limiting distribution of the AR statistic and the score, we need the following assumptions.
\begin{assumption}\label{ass:eigbound} 
    (a) $\sum_{i=1}^{n}\|\bar{\bs z}_{i}\|^2/n\leq C<\infty$ a.s.n., (b) $\displaystyle{\max_{i=1,\ldots,n}}\|\bar{\bs z}_{i}\|^2/n\rightarrow_{a.s.}0$, 
    (c) $\displaystyle{\max_{i=1,\ldots,n}}\|\bar{\bs Z}'\bs V\bs D_{\varepsilon}\bs e_{i}\|^2/n\rightarrow_{a.s.}0,$ (d) $0<C^{-1}\leq\lambdamin(\frac{1}{n}\bs Z'\bs Z)\leq\lambdamax(\bs Z'\bs Z/n)\leq C<\infty$ a.s.n., $0<C^{-1}\leq\lambdamin(\bs Z'\bs D_{\varepsilon}^2\bs Z/n)\leq\lambdamax(\bs Z'\bs D_{\varepsilon}^2\bs Z/n)\leq C<\infty$ a.s.n., (e) $P_{ii}\leq C<1$ a.s.n.
\end{assumption}
Part (a) and (b) are standard assumptions under many instruments. Part (a) also appears in \citet{chao2012asymptotic} and \citet{hausman2012instrumental}, who instead of (b) require $n^{-2}\sum_{i=1}^{n}\|\bar{\bs z}_{i}\|^4\rightarrow_{a.s.}0$. We see that this condition is implied by Assumption \ref{ass:eigbound} parts (a) and (b). In particular, (b) is a Lyapunov condition needed for the CLT we employ. Part (c) is another Lyapunov condition needed for the CLT under heteroskedasticity. Part (d) ensures that $\bs Z(\bs Z'\bs D_{\varepsilon}^2\bs Z)^{-1}\bs Z'$ has bounded eigenvalues $a.s.n.$ Part (e) is a `balanced design' assumption, e.g. \citet{chao2012asymptotic}. Finally, note that we make no assumption on $\bs\Pi$, allowing for weak and even irrelevant instruments. 

The joint limiting distribution of the AR statistic and the score evaluated at the true parameter $\bs\beta_{0}$ is given in the following theorem.
\begin{theorem}\label{thm:main} Let $\bs T(\bs\beta_{0}) = (k^{-1/2}(\AR(\bs\beta_{0})-k),
		\sqrt{n} \bs S(\bs\beta_{0}))'$. Under Assumptions \ref{ass:model} to \ref{ass:eigbound}, when $n\rightarrow\infty$ and $k/n\rightarrow\lambda\in (0,1)$, 
		$\hat{\bs\Sigma}_{n}(\bs\beta_{0})^{-1/2}\bs T(\bs\beta_{0}) \rightarrow_{d} N(\bs 0,\bs I_{p+1})$,
     with 
    \begin{equation}
	\hat{\bs\Sigma}_n(\bs\beta) =
	\begin{pmatrix}
		\hat\sigma_{n}^2(\bs\beta) & \big[\hat{\bs\Sigma}_n(\bs\beta)\big]_{2:p+1,1}' \\
	    \big[\hat{\bs\Sigma}_n(\bs\beta)\big]_{2:p+1,1}&\hat{\bs\Omega}(\bs\beta)
	\end{pmatrix},
    \end{equation}
    an unbiased and consistent estimator of
    \begin{equation}
    \bs\Sigma_{n}(\bs\beta_0)=\var(\bs T(\bs\beta_{0})|\mathcal{J})=\begin{pmatrix}
		\sigma_{n}^2(\bs\beta_0) & \big[\bs\Sigma_n(\bs\beta_0)\big]_{2:p+1,1}' \\
	    \big[\bs\Sigma_n(\bs\beta_0)\big]_{2:p+1,1}&\bs\Omega(\bs\beta_0)
	\end{pmatrix}.
\end{equation}
    Both the variance and the estimator are detailed in Appendix~\ref{sapp:proof main}.

\end{theorem}
\begin{proof}
See Appendix \ref{sapp:proof main}.
\end{proof}

In Appendix~\ref{sapp:proof main} we show that the variance of the score, $\bs\Omega(\bs\beta_0)$, can be decomposed into terms that are present in the case with a limited number of moments, $\bs\Omega^{L}(\bs\beta_0)$, and additional terms labeled $\bs\Omega^{H}(\bs\beta_{0})$ that appear due to the presence of many moments. That is, $\bs\Omega(\bs\beta_0)=\bs\Omega^{L}(\bs\beta_0)+\bs\Omega^{H}(\bs\beta_0)$. The most important property of the additional variance terms in $\bs \Omega^{H}(\bs\beta_{0})$ is given by the following result.
\begin{lemma}
    \label{lem:negdef}
   Suppose that Assumptions~\ref{ass:model} to \ref{ass:decomposev} hold, and furthermore, $\max_{i=1,\ldots,n}P_{ii}\leq 0.9$ and $ n^{-1}\sum_{i=1}^{n}V_{ii}P_{ii}>0$. Then, $\bs\Omega^{H}(\bs\beta_{0})$ from Theorem \ref{thm:expandvar} is negative definite.
\end{lemma}
\begin{proof}
See Appendix \ref{app:proof negdef}.
\end{proof}
Lemma~\ref{lem:negdef}
implies that the use of the conventional inference procedures based on the score will be conservative. The condition that $n^{-1}\sum_{i=1}^{n}V_{ii}P_{ii}>0$ makes it clear that this is more likely to occur, and asymptotically only occurs, when the number of instruments is a nonnegligible fraction of the sample size. For instance, if $\{\varepsilon_{i}\}_{i=1}^{n}$ is itself a sequence of Rademacher variables and $\bs Z$ has independent elements with finite fourth moment. 
\citet*{bai2007asymptotics} show that $\frac{1}{n}\sum_{i=1}^{n}V_{ii}P_{ii} - k^2/n^2 \rightarrow_{a.s.} 0$, and hence  (c) requires that $k/n\rightarrow\lambda>0$.

\section{Simulation results}\label{sec:MC}
In this section we assess the finite sample properties of the proposed tests through a simulation study on the panel tobit model from \citet{honore1992trimmed} discussed in Example \ref{ex:Honore}, the quantile IV model from \citet{chernozhukov2006finite} discussed in Example \ref{ex:IVQR}, and the linear IV model discussed in Example \ref{ex:linIV}.  To contrast the tests developed under many moment sequences with those developed under fixed moment conditions, we label the former as many instrument (MI) and the latter as fixed-$k$ tests.

\subsection{Panel tobit}\label{ssec:MC Honore}
Consider the censored panel data model as in \citet{honore1992trimmed}
\begin{equation}\label{eq:honore}
    \begin{split}
        y_{it}&=\max\{0,y_{it}^*\}, \quad 
        y_{it}^*=\alpha_i+\bs x_{it}'\bs\beta+\varepsilon_{it},
    \end{split}
\end{equation}
where $y_{it}$ is the censored version of a latent variable $y_{it}^*$ that relates to regressors $\bs x_{it}\in\mathbb{R}^{p}$, unobserved individual fixed effects $\alpha_{i}$ and an error $\varepsilon_{it}$.

\citet{honore1992trimmed} assumes that for each individual $i$ the $\varepsilon_{it}$ are independent and identically distributed over time $t$ conditional on $\bs x_{i1}$, $\bs x_{i2}$ and $\alpha_i$, allowing for heteroskedasticity across individuals. Under this assumption, or conditional exchangeability more generally, the difference between $y_{it_1}^*$ and $y_{it_2}^*$ for some $t_1$ and $t_2$, as denoted by $\Delta_{t_1,t_2}y_{i}^*$, conditional on $\bs x_{it_1}$ and $\bs x_{it_2}$ is symmetrically distributed around $\Delta_{t_1,t_2}\bs x_{i}'\bs\beta=(\bs x_{it_1}-\bs x_{it_2})'\bs\beta$. Furthermore, if $y_{it_1}^*$ and $y_{it_2}^*$ are both greater than zero, neither is affected by the censoring and therefore also the observed values are unaffected by the censoring. 
\citet{honore1992trimmed} combines these two observations to define two sets of outcomes  for ($y_{it_1}$,$y_{it_2}$)
\begin{equation*}
    \begin{split}
    A_{1i,t_1,t_2}&=\begin{cases}
        \{(y_{it_1}^*,y_{it_2}^*): y_{it_1}^*>\Delta_{t_1,t_2}\bs x_{i}'\bs\beta, y_{it_2}^*>y_{it_1}^*-\Delta_{t_1,t_2}\bs x_{i}'\bs\beta\} &\text{if }\Delta_{t_1,t_2}\bs x_{i}'\bs\beta\geq0\\
        \{(y_{it_1}^*,y_{it_2}^*): y_{it_1}^*>0, y_{it_2}^*>y_{it_1}^*-\Delta_{t_1,t_2}\bs x_{i}'\bs\beta\}&\text{if }\Delta_{t_1,t_2}\bs x_{i}'\bs\beta<0
    \end{cases}\\
    B_{1i,t_1,t_2}&=\begin{cases}
        \{(y_{it_1}^*,y_{it_2}^*): y_{it_1}^*>\Delta_{t_1,t_2}\bs x_{i}'\bs\beta, 0<y_{it_2}^*<y_{it_1}^*-\Delta_{t_1,t_2}\bs x_{i}'\bs\beta\} &\text{if }\Delta_{t_1,t_2}\bs x_{i}'\bs\beta\geq0\\
        \{(y_{it_1}^*,y_{it_2}^*): y_{it_1}^*>0, -\Delta_{t_1,t_2}\bs x_{i}'\bs\beta<y_{it_2}^*<y_{it_1}^*-\Delta_{t_1,t_2}\bs x_{i}'\bs\beta\}&\text{if }\Delta_{t_1,t_2}\bs x_{i}'\bs\beta<0.
    \end{cases}
    \end{split}
\end{equation*}
Conditional on $(x_{it_1},x_{it_2})$, ($y_{it_1}$,$y_{it_2}$) falls with equal probability in $A_{1i,t_1,t_2}$ and $B_{1i,t_1,t_2}$, which leads to the following reflection invariant moment condition
\begin{equation}\label{eq:moment honore 1}
\E[(\mathbbm{1}\{(y_{it_1},y_{it_2})\in A_{1,t_1,t_2}\}-\mathbbm{1}\{(y_{it_1},y_{it_2})\in B_{1,t_1,t_2}\})\Delta_{t_1,t_2} \bs x_{i}]=\bs 0.
\end{equation}

Similarly, \citet{honore1992trimmed} notes that, conditional on $(\bs x_{it_1},\bs x_{it_2})$, the expected vertical distance from $(y_{it_1},y_{it_2})$ given that $(y_{it_1},y_{it_2})$ lies in $A_{1i,t_1,t_2}$, to the $45^\circ$ line through $(\Delta_{t_1,t_2}\bs x_{i}'\bs\beta,0)$ if $\Delta_{t_1,t_2}\bs x_{i}'\bs\beta\geq0$ or through $(0,-\Delta_{t_1,t_2}\bs x_{i}'\bs\beta)$ if $\Delta_{t_1,t_2}\bs x_{i}'\bs\beta<0$, equals the expected horizontal distance from $(y_{it_1},y_{it_2})$ given that $(y_{it_1},y_{it_2})$ lies in $B_{1i,t_1,t_2}$, to the same line. These distances are given by $-(y_{it_1}-y_{it_2}-\Delta_{t_1,t_2}\bs x_{i}'\bs\beta)$ and $(y_{it_1}-y_{it_2}-\Delta_{t_1,t_2}\bs x_{i}'\bs\beta)$, which leads to the second reflection invariant moment condition
\begin{equation}\label{eq:moment honore 2}
    \begin{split}
        &\E[(\E[\mathbbm{1}\{(y_{it_1},y_{it_2})\in A_{1i,t_1,t_2}\}(y_{it_1}-y_{it_2}-\Delta_{t_1,t_2}\bs x_{i}'\bs\beta)|\bs x_{it_1},\bs x_{it_2}]\\
        &\quad+\E[\mathbbm{1}\{(y_{it_1},y_{it_2})\in B_{1i,t_1,t_2}\}(y_{it_1}-y_{it_2}-\Delta_{t_1,t_2}\bs x_{i}'\bs\beta)|\bs x_{it_1},\bs x_{it_2}])\Delta_{t_1,t_2}\bs x_{i}]\\
        &=\E[\mathbbm{1}\{(y_{it_1},y_{it_2})\in A_{1i,t_1,t_2}\cup B_{1i,t_1,t_2}\}(y_{it_1}-y_{it_2}-\Delta_{t_1,t_2}\bs x_{i}'\bs\beta)\Delta_{t_1,t_2}\bs x_{i}]=\bs 0.
    \end{split}
\end{equation}

The sets in which the observations are unaffected by censoring can be slightly enlarged, to yield two additional moment conditions. These moment conditions are also the difference between either indicators of events with equal probability or of equal expected distances, and hence are also reflection invariant. To be precise, define
\begin{equation*}
    \begin{split}
    A_{2i,t_1,t_2}&=\begin{cases}
        \{(y_{it_1}^*,y_{it_2}^*): y_{it_1}^*\leq\Delta_{t_1,t_2}\bs x_{i}'\bs\beta, y_{it_2}^*>0\} &\text{if }\Delta_{t_1,t_2}\bs x_{i}'\bs\beta\geq0\\
        \{(y_{it_1}^*,y_{it_2}^*): y_{it_1}^*\leq0, y_{it_2}^*>-\Delta_{t_1,t_2}\bs x_{i}'\bs\beta\}&\text{if }\Delta_{t_1,t_2}\bs x_{i}'\bs\beta<0
    \end{cases}\\
    B_{2i,t_1,t_2}&=\begin{cases}
        \{(y_{it_1}^*,y_{it_2}^*): y_{it_1}^*>\Delta_{t_1,t_2}\bs x_{i}'\bs\beta, y_{it_2}^*\leq0\} &\text{if }\Delta_{t_1,t_2}\bs x_{i}'\bs\beta\geq0\\
        \{(y_{it_1}^*,y_{it_2}^*): y_{it_1}^*>0, y_{it_2}^*<-\Delta_{t_1,t_2}\bs x_{i}'\bs\beta\}&\text{if }\Delta_{t_1,t_2}\bs x_{i}'\bs\beta<0.
    \end{cases}
    \end{split}
\end{equation*}
Then in addition to \eqref{eq:moment honore 1} and \eqref{eq:moment honore 2}, the following reflection invariant moment conditions hold
\begin{equation*}
    \begin{split}
        &\E[(\mathbbm{1}\{(y_{it_1},y_{it_2})\in A_{1i,t_1,t_2}\cup A_{2i,t_1,t_2}\}-\mathbbm{1}\{(y_{it_1},y_{it_2})\in B_{1i,t_1,t_2}\cup B_{2i,t_1,t_2}\})\Delta_{t_1,t_2}\bs x_{i}']=\bs 0\\
        &\E[(\mathbbm{1}\{(y_{it_1},y_{it_2})\in A_{1i,t_1,t_2}\}(y_{it_1}-y_{it_2}-\Delta_{t_1,t_2}\bs x_{i}'\bs\beta)\\
        &\quad-\mathbbm{1}\{(y_{it_1},y_{i,t_2})\in A_{2i,t_1,t_2}\}(y_{it_2}-\max\{0,-\Delta_{t_1,t_2}\bs x_{i}'\bs\beta\})\\
        &\quad+\mathbbm{1}\{(y_{it_1},y_{it_2})\in B_{1i,t_1,t_2}\}(y_{it_1}-y_{it_2}-\Delta_{t_1,t_2}\bs x_{i}'\bs\beta)\\
        &\quad+\mathbbm{1}\{(y_{it_1},y_{i,t_2})\in B_{2i,t_1,t_2}\}(y_{it_1}-\max\{0,-\Delta_{t_1,t_2}\bs x_{i}'\bs\beta\}))\Delta_{t_1,t_2}\bs x_{i}']=\bs0.
    \end{split}
\end{equation*}

Any two periods $t_1\neq t_2$ can be used to obtain moment conditions, which brings the potential number of moment conditions per individual to $4\binom{T}{2}p$. In what follows, however, we will not use the same time period twice for the MI AR, since although the moment conditions are all reflection invariant on their own, they are not necessarily jointly invariant when the same period is used twice. 
To illustrate this, note that the invariance of the moment conditions stems from differences in exchangeable errors. Then take $\varepsilon_{it}$ discrete with $\Pr(\varepsilon_{it}=1)=2/3$ and $\Pr(\varepsilon_{it}=-2)=1/3$, and independent over $t=1,2,3$. In this case, $(\varepsilon_{it}-\varepsilon_{it'})$ is reflection invariant, but $(\varepsilon_{i1}-\varepsilon_{i2},\varepsilon_{i2}-\varepsilon_{i3})$ is not, as for example  $\Pr((\varepsilon_{i1}-\varepsilon_{i2},\varepsilon_{i2}-\varepsilon_{i3})=(0,-3))=4/27$, whereas $\Pr((\varepsilon_{i1}-\varepsilon_{i2},\varepsilon_{i2}-\varepsilon_{i3})=(0,3))=2/27$.
Hence, we will use $4\lfloor T/2\rfloor p$ moment conditions.

\begin{figure}[t]
    \centering
    \caption{Size in the panel tobit simulation.}
    \label{fig:size Honore heteroskedastic}
    \includegraphics{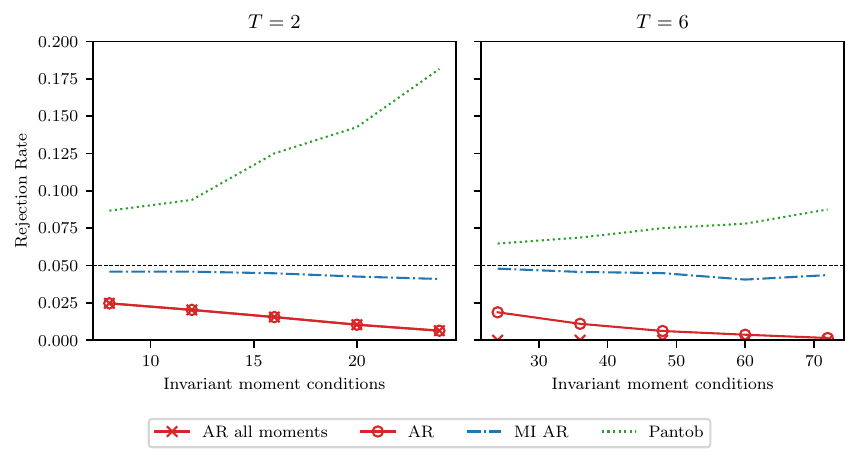}
    \begin{minipage}{\textwidth}\footnotesize\textit{Note:} Size when testing $H_{0}\colon \bs\beta=\bs\iota$ at $\alpha=0.05$ based on the fixed-$k$ Anderson--Rubin test with all moment conditions (AR all moments), the fixed-$k$ Anderson--Rubin test with the invariant moment conditions (AR), the many instrument Anderson--Rubin test (MI AR) and the test based on the panel tobit estimator by \citet{honore1992trimmed} (Pantob). $T$ denotes the number of time periods. Section \ref{ssec:MC Honore} describes the DGP.
    \end{minipage}
\end{figure}

\begin{figure}[t]
    \centering
    \caption{Power in the panel tobit simulation.}
    \label{fig:power Honore heteroskedastic}
    \includegraphics{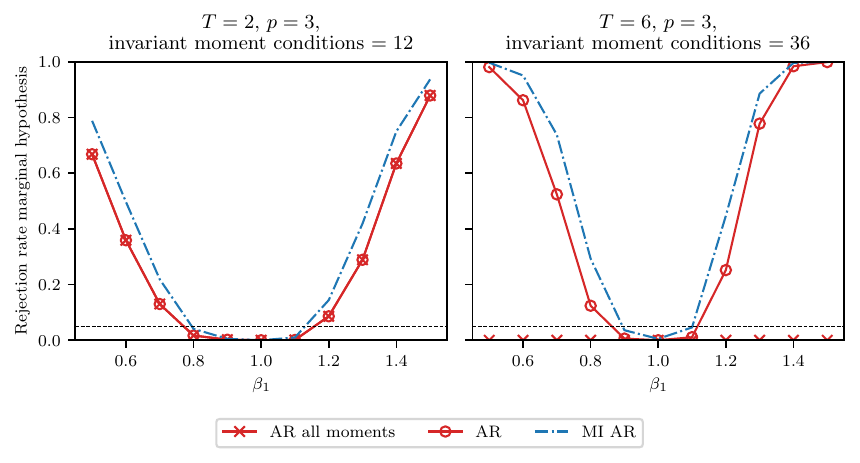}
    \begin{minipage}{\textwidth}\footnotesize\textit{Note:} Marginal power when testing $H_{0}\colon\beta_1=1$ at $\alpha=0.05$ when the true value of $\beta_1$ varies between 0.5 and 1.5, based on the fixed-$k$ Anderson--Rubin test with all moment conditions (AR all moments), the fixed-$k$ Anderson--Rubin test with the invariant moment conditions (AR) and the many instrument Anderson--Rubin test (MI AR). $T$ denotes the number of time periods, $p$ the number of regressors. Section \ref{ssec:MC Honore} describes the DGP.
    \end{minipage}
\end{figure}
We use the simulation design in \citet{honore1992trimmed} to study the size and power of the MI AR test compared to benchmarks that we explain below. We make two adaptations to the DGP, such that $T$ and $p$ can be flexible and to include heteroskedasticity. We generate $n=200$ observations from \eqref{eq:honore}, where $x_{1,it}=\alpha_i+\eta_{it}$, where $\alpha_{i},\eta_{it}\sim N(0,1)$. The other regressors $x_{l,it}$, $l=2,\dots, p$ are also standard normally distributed and $p$ varies over $\{2,\ldots,6\}$. Finally, $\varepsilon_{it}\sim N(0,\alpha_{i}^2)$. The results are averaged over $10,000$ draws from the DGP.

As benchmarks we include the fixed-$k$ AR test that uses the same moment conditions as the MI AR test, the AR test that uses all possible moment conditions, and the test based on \citeauthor{honore1992trimmed}'s (\citeyear{honore1992trimmed}) panel tobit estimator. For the latter we used the STATA package as on the author's website. To the best of our knowledge, there are few other tests for coefficients in panel tobit models that are as generally applicable as \citeauthor{honore1992trimmed}'s (\citeyear{honore1992trimmed}) panel tobit estimator. 

Figure \ref{fig:size Honore heteroskedastic} shows the rejection rates of the MI AR test and the benchmarks when testing the true null hypothesis $H_0\colon\bs\beta=\bs\iota$ at a $5\%$ significance level. The number of moment conditions increases with $T$ and $p$. For $T=2$ all moment conditions are jointly reflection invariant, hence the two fixed-$k$ AR tests have identical rejection rates. These rejection rates are well below the desired $5\%$ rate and decreasing with the number of moment conditions. For $T=6$ the fixed-$k$ AR test that uses all moment conditions never rejects. Also the rejection rates of the fixed-$k$ AR statistic that uses only the invariant moment conditions are too low and again decreasing with the number of moments. \citeauthor{honore1992trimmed}'s (\citeyear{honore1992trimmed}) panel tobit estimator on the other hand is oversized and more so with many moment conditions. Only the MI AR statistic is close to size correct uniformly over the number of moment conditions.

Next, we investigate the power of the MI AR test relative to the fixed-$k$ AR that uses all moment conditions and the fixed-$k$ AR test that uses only the invariant moment conditions. We do not include test based on \citeauthor{honore1992trimmed}'s (\citeyear{honore1992trimmed}) panel tobit estimator as it was not size correct in Figure \ref{fig:size Honore heteroskedastic}. In Figure \ref{fig:power Honore heteroskedastic} we show the rejection rates when testing the marginal hypothesis $H_{0}\colon\beta_{1}=1$ over $1,000$ data sets when we vary $\beta_{1}$ in the DGP from $-0.5$ to $1.5$. The other coefficients in $\bs\beta$ remain fixed at $1$. We reject the null hypothesis only if the statistic for $\beta_{1}=1$ exceeds the critical value for all values of $\beta_{2}$ to $\beta_{p}$. One can alternatively view this as, for a given value of $\beta_1$, minimizing the AR and MI AR statistics over $\beta_2$ to $\beta_p$ and checking whether the minimum exceeds the critical value \citep{chernozhukov2009finite}. In this case the minimization is by means of a grid search, but other algorithms can be used as well. From the figure we conclude that the identification robust tests have lower rejection rates than the nonrobust test based on the panel tobit estimator. More importantly, the MI AR test has higher power than the fixed-$k$ AR tests and the power difference increases with the number of moment conditions.

\begin{figure}[t!]
    \centering
    \caption{Size in the median IV regression simulation with skewed moment conditions.}
    \label{fig:IVQR size quantile}
    \includegraphics{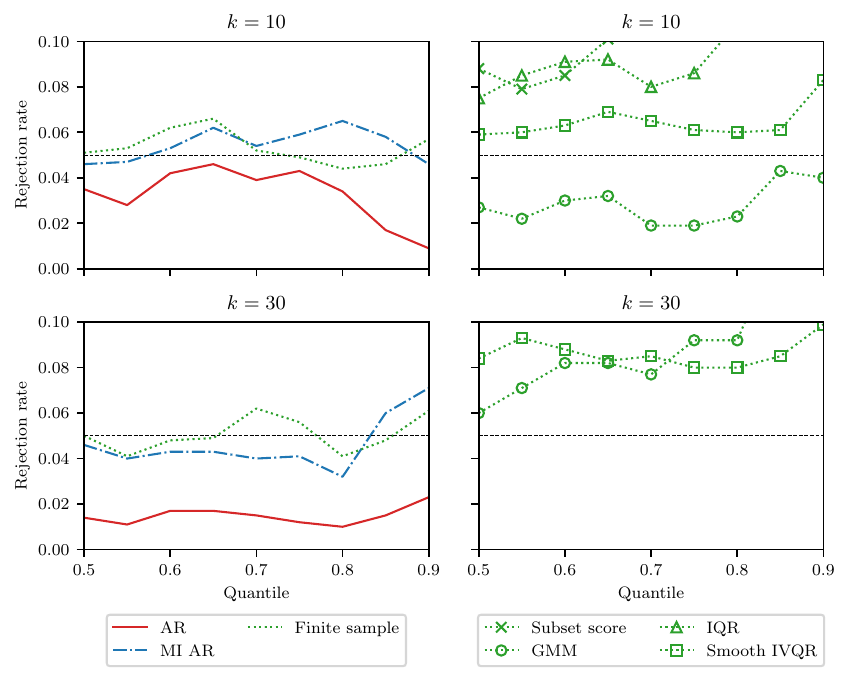}
    \begin{minipage}{\textwidth}\footnotesize\textit{Note:} Size when testing $H_0\colon\theta_{1}(\tau)=-1+\Phi^{-1}(\tau)$ and $\theta_{2}(\tau)=1$, where $\tau$ is the quantile at $\alpha=0.05$ based on the fixed-$k$ Anderson--Rubin test (AR), the many instrument Anderson--Rubin test (MI AR), the finite sample test by \citet{chernozhukov2009finite} (Finite sample), the subset score test by \citet{jun2008weak} (Subset score), and tests based on a GMM estimator (GMM), the inverse quantile regressor estimator by \citet{chernozhukov2006instrumental} (IQR) and the smoothed IVQR estimator by \citet{kaplan2017smoothed} (Smooth IVQR). $k$ denotes the number of instruments. Section \ref{ssec:MC IVQR} describes the DGP.
    \end{minipage}
\end{figure}

\subsection{Median IV regression}\label{ssec:MC IVQR}
We now consider the median IV regression from Example \ref{ex:IVQR}. This provides a natural setting to study the effect of violations of Assumption~\ref{ass:main} as the moment conditions in the quantile IV regression are symmetrically distributed for the median, but not for other quantiles.
We slightly adapt the simulation set-up from \citet{chernozhukov2006finite} to increase the number of instrumental variables and to have nonreflection invariant moment conditions. To be precise, there is a dependent variable $y_i$ that relates to a single endogenous regressor $x_i$ as $y_i=-1+x_i+\varepsilon_i$. The endogenous regressor relates to $k$ instruments $\bs z_i$ as $x_i=\bs z_i'\bs\iota_k\pi+\nu_i$ for $i=1,\dots,n$. Since the moment conditions consist of a product of $\tau-\mathbbm{1}\{y_i\leq\theta_{1}(\tau)+\theta_{2}(\tau)x_i\}$ with the instruments $\bs z_{i}$, we draw the instruments from a skewed distribution for the moment conditions to no longer be reflection invariant at $\tau\neq 1/2$. We draw from a Gamma distribution with shape parameter $1/\zeta$, rate parameter $\sqrt{1/(\sigma^2\zeta)}$ and we subtract $\sqrt{\sigma^2/\zeta}$, such that the distribution has mean zero, variance $\sigma^2$ and skewness governed by $\zeta$. Figure \ref{fig:histogram gamma} in Appendix \ref{sapp:skewed} shows how the skewness increases with $\zeta$. We set $\zeta=3/2$ and $\sigma^2=1$. The errors $\varepsilon_i$ and $\nu_i$ are standard normally distributed with correlation $\rho=0.8$. The first stage coefficient $\pi=0.5$. The number of observations is set to $n=100$. The results are averaged over $1,000$ draws from the DGP.

We fit the conditional quantile model $q(x_i,\tau;\bs\theta)=\theta_{1}(\tau)+\theta_{2}(\tau)x_i$, such that the true values for the intercept and slope parameters are $\theta_{1}(\tau)=-1+\Phi^{-1}(\tau)$ for $\Phi$ the CDF of the standard normal distribution, and  $\theta_{2}(\tau)=1$. We use the moment conditions as suggested in Example \ref{ex:IVQR} with $\bs\Psi(\bs z_{i})=\bs z_{i}$ to test values of $\bs\theta$ with a fixed-$k$ AR and the MI AR test. As benchmarks, we include the exact test by \citet{chernozhukov2009finite} with $1,000$ draws to approximate the critical values, the weak identification robust score test by \citet{jun2008weak}, a GMM based test with the limiting distribution derived as in \citet{andrews1994empirical} and \citet{arellano2009gmm}, the inverse quantile regressor (IQR) by \citet{chernozhukov2006instrumental} and the smoothed IVQR by \citet{kaplan2017smoothed}. See Appendix \ref{sapp:details MC IVQR} for details of these benchmarks.

The left panels of Figure \ref{fig:IVQR size quantile} show the size of the fixed-$k$ AR, MI AR and the test by \citet{chernozhukov2009finite} for $k=10$ and $k=30$ against the chosen quantile. The MI AR test appears relatively robust to small deviations from the invariance assumption. Only for $k=30$ and quantiles higher than $0.8$ there is a sudden increase in the rejection rate, but size distortions stay limited. The fixed-$k$ AR is conservative due to the many instruments, while  the test by \citet{chernozhukov2009finite} attains the nominal size for all quantiles. 

The right panels of the same figure show the size of the other benchmarks. It strikes that, except for the test based on the GMM estimator when $k=10$, all tests are oversized.

Further results on size and power at $\tau=1/2$ can be found in Appendix~\ref{sapp:medianIV} showing that the MI AR test offers size and power close to the exact test.

\subsection{Linear IV model}\label{ssec:MC linear IV}
For the linear IV, we base our simulations  on \citet{hausman2012instrumental}. We adapt the DGP slightly to allow for clustering in the data and vary the degree of heteroskedasticity.
We generate $n=800$ observations from $y_i=\alpha+\beta x_{i}+\varepsilon_i$, $i=1,\dots,n$, where $\alpha=0$. We test $H_{0}\colon\beta = 0$. The regressor is generated as $x_{i}=\pi \tilde{z}_{i}+\eta_i$, whose distribution we detail below. In addition to the instrumental variable $\tilde{z}_i$ with coefficient $\pi$, there are $k$ other instrumental variables stacked in $\bs z_i=(1,\tilde{z}_i, \tilde{z}_i^2, \tilde{z}_i^3, \tilde{z}_i^4, \tilde{z}_i D_{i1},\dots, \tilde{z}_iD_{ik-4})'$, where the $D_{ij}$, $j=1,\dots,k-4$, are independent Bernoulli(1/2) random variables. 

To allow the data to be clustered, we divide the observations over $H=100$ clusters containing between 4 and 12 observations. The random variables $\tilde{z}_i$, $\eta_i$, $v_{1i}$ and $v_{2i}$ consist of a cluster common component and an idiosyncratic component weighted by $\lambda$. That is, $\tilde{\bs z}_{[h]}=\sqrt{\lambda}\tilde{\bs z}_{[h]}^\text{ind}+\sqrt{1-\lambda}\tilde{z}_{h}^\text{cl}$, $\bs\eta_{[h]}=\sqrt{\lambda}\bs\eta_{[h]}^\text{ind}+\sqrt{1-\lambda}\eta_{h}^\text{cl}$, $\bs v_{1,[h]}=\sqrt{\lambda}\bs v_{1,[h]}^\text{ind}+\sqrt{1-\lambda}v_{1,h}^\text{cl}$ and $\bs v_{2,[h]}=\sqrt{\lambda}\bs v_{2,[h]}^\text{ind}+\sqrt{1-\lambda}v_{2,h}^\text{cl}$ for $h=1,\dots, H$. For the independent setting, we set $\lambda=1$, for the clustered setting $\lambda=0.5$. We draw the cluster common and idiosyncratic components of $z_{i}$ and $\eta_i$ from standard normal distributions. We generate $v_{1i}^\text{ind}\sim N(0,(\tilde{z}_{i}^\text{ind})^\kappa)$ and $v_{1h}^\text{cl}\sim N(0,(\tilde{z}_{h}^\text{cl})^\kappa)$, where $\kappa$ controls the degree of heteroskedasticity. We start with $\kappa=2$ as in \citet{hausman2012instrumental} and increase it to $\kappa = 6$. The cluster common and idiosyncratic components of $v_{2i}$ come from a normal distribution with mean zero and variance $0.86^2$. 
Finally, the second stage error generated as $\varepsilon_i=\rho\eta_i+\sqrt{(1-\rho^2)/(\phi^2+0.86^4)}(\phi v_{1i}+0.86v_{2i})$, where $\rho=0.3$ is the degree of endogeneity and $\phi=1.38072$ as in the implementation by \citet{bekker2015jackknife}. Appendix \ref{sapp:skewed} considers a DGP with skewed moment conditions where the AR test rejects around 12\% in the most extreme scenario, while the score test rejects 8\% of the cases.

\begin{figure}[t]
    \centering
    \caption{Size in the linear IV simulation.}
    \label{fig:size}
    \includegraphics{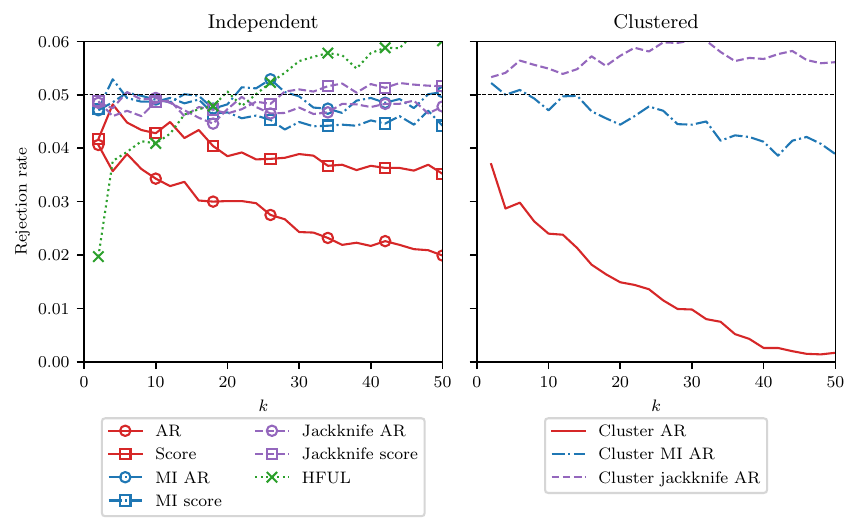}
    \begin{minipage}{\textwidth}\footnotesize\textit{Note:} Size when testing $H_0\colon\beta=0$ at $\alpha=0.05$ based on the (cluster) fixed-$k$ Anderson--Rubin test (AR), the fixed-$k$ score test (Score), the (cluster) many instrument Anderson--Rubin test (MI AR), the many instrument score test (MI score), the (cluster) jackknife Anderson--Rubin test (Jackknife AR, without cross-fit variance estimator), the jackknife score test (Jackknife score, without cross-fit variance estimator) and HFUL (HFUL). In the left panel the observations are independent. In the right panel the 800 observations are clustered in 100 unbalanced clusters. $k$ denotes the number of instruments. Section \ref{ssec:MC linear IV} describes the DGP and Appendix \ref{sapp:details MC} the implementation of the tests.
    \end{minipage}
\end{figure}

We investigate the size of the MI AR and MI score, and compare this to the fixed-$k$ AR, the fixed-$k$ score, the jackknife AR, the jackknife score test and the HFUL when the data are independent. See Appendix \ref{sapp:details MC} for details on the implementation of the different tests. 
The rejection rates are calculated over $10,000$ data sets.

The left panel of Figure \ref{fig:size} shows the rejection rate of the different tests when testing $H_0\colon\beta=0$ at a $5\%$ level for $k$ ranging from $2$ to $50$ and $\pi=\sqrt{8/n}$ as in \citet{hausman2012instrumental}. For weak instruments all identification robust tests are close to size correct as long as the number of instruments is not too high. In line with the theory from Sections \ref{subsec:independent} and \ref{sec:score},  the rejection rates of the fixed-$k$ AR and score tests drop when the number of instruments increases. The MI and jackknife AR and score tests remain size correct. HFUL is conservative for few weak instruments, but becomes oversized for many instruments.

\begin{figure}[t]
    \centering
    \caption{Power in the linear IV simulation with independent observations.}
    \label{fig:power}
    \includegraphics{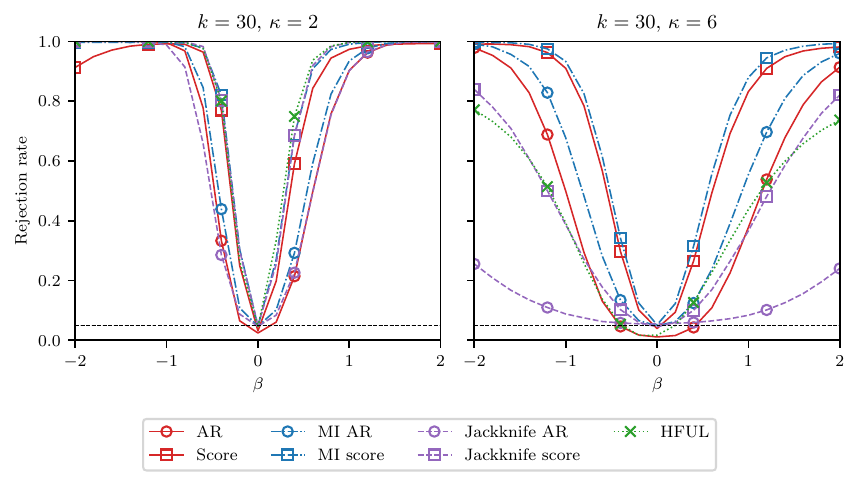}
    \begin{minipage}{\textwidth}\footnotesize\textit{Note:} Power when testing $H_{0}\colon\beta=0$ at $\alpha=0.05$ when the true value of $\beta$ varies between -2 and 2 based on the fixed-$k$ Anderson--Rubin test (AR), the fixed-$k$ score test (Score), the many instrument Anderson--Rubin test (MI AR), the many instrument score test (MI AR), the jackknife Anderson--Rubin test (Jackknife AR, without cross-fit variance estimator), the jackknife score test (Jackknife score, without cross-fit variance estimator) and HFUL (HFUL). The observations are independent. $k$ denotes the number of instruments, $\kappa$ the heteroskedasticity parameter and $\pi$ the instrument strength parameter. Section \ref{ssec:MC linear IV} describes the DGP and Appendix \ref{sapp:details MC} the implementation of the tests.
    \end{minipage}
\end{figure}

In the right panel of Figure \ref{fig:size}, we consider a DGP in which the data are clustered.
The figure shows the rejection rates of the cluster fixed-$k$ AR, the cluster MI AR and the cluster jackknife AR for the same hypothesis and 10,000 data sets. Appendix \ref{sapp:details MC} details the exact implementation. This figure is similar to the figure for independent data in the sense that the cluster fixed-$k$ AR test is conservative, whereas the cluster MI AR and cluster jackknife AR tests are size correct.

\begin{figure}[t]
    \centering
    \caption{Power in the linear IV simulation with clustered observations.}
    \label{fig:power clustered}
    \includegraphics{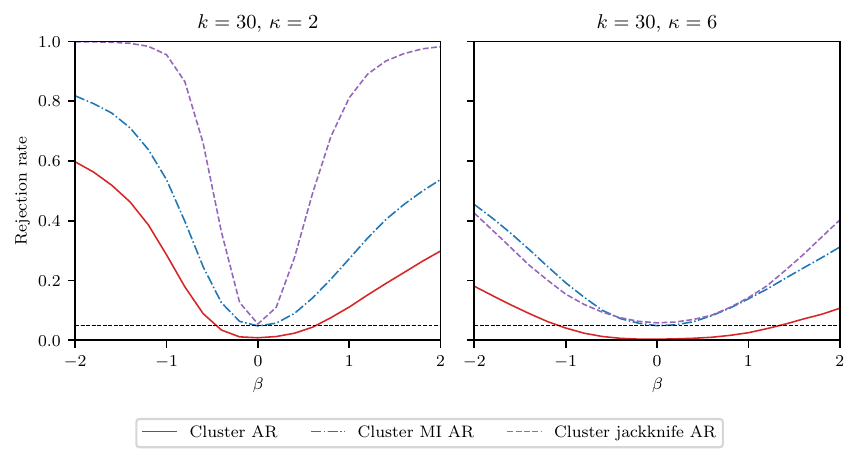}
    \begin{minipage}{\textwidth}\footnotesize\textit{Note:} Power when testing $H_{0}\colon\beta=0$ at $\alpha=0.05$ when the true value of $\beta$ varies between -2 and 2 based on the cluster fixed-$k$ Anderson--Rubin test (Cluster AR), the cluster many instrument Anderson--Rubin test (Cluster MI AR) and the cluster jackknife Anderson--Rubin test (Cluster jackknife AR, without cross-fit variance estimator). The 800 observations are clustered in 100 unbalanced clusters. $k$ denotes the number of instruments, $\kappa$ the heteroskedasticity parameter and $\pi$ the instrument strength parameter. Section \ref{ssec:MC linear IV} describes the DGP and Appendix \ref{sapp:details MC} the implementation of the tests.
    \end{minipage}
\end{figure}

We next study the power against $H_0\colon\beta=0$ of the tests from the previous subsection when varying the true $\beta$ between $-2$ and $2$. We fix $\pi=\sqrt{128/n}$ to attain reasonable power, set the number of instruments  $k=30$ and vary the degree of heteroskedasticity by choosing $\kappa\in\{2,6\}$. The values of the other parameters in the DGP are unaltered.

Figure \ref{fig:power} shows the power for independent data. We focus first on the AR statistics, which we have marked with a circle. In the left panel, with mild heteroskedasticity, we see that although the power of the three different types of AR tests are comparable, the MI AR has a slight edge over the other two. In the right panel, with stronger heteroskedasticity, the power differences are bigger. Especially the jackknife AR has low rejection rates. 

The score tests are marked with a square. In both panels, they outperform the AR tests, which might be due to the relatively strong instruments. In the left panel we furthermore see that the score tests have comparable power to HFUL. In the right panel we again have higher power for the fixed-$k$ and MI tests have higher power than for the jackknife test. In this case also HFUL does not perform as well as the fixed-$k$ and MI score.

For clustered data, as shown in Figure \ref{fig:power clustered}, the relative performance of the MI AR test is more nuanced. It still outperforms the fixed-$k$ test and by a larger margin than for independent data. The jackknife variant is clearly more powerful when the degree of heteroskedasticity is low, although the difference disappears if we increase the degree of heteroskedasticity in the right panel.

\section{Empirical application: the effect of bank bias on stability}\label{sec:application}
We revisit the study by \citet{langfield2016bank} 
on the effect of the relative concentration of finance in the banking sector in Europe on financial stability. The effects of this concentration, also called ``bank bias", are ambiguous: some argue that a larger banking sector increases systemic risk, whereas others argue that larger banking mitigates systemic risk \citep[see the discussion in][]{langfield2016bank}. 
To study the question empirically, \citet{langfield2016bank} construct a panel with 3981 observations on 467 banks in 20 countries observed over maximally 12 years. 

The outcome variable is the bank-specific systemic risk intensity defined as the bank's SRIKS, a measure of systemic risk provided by the New York University's \href{https://vlab.stern.nyu.edu/}{Volatility Laboratory}, over the bank's total assets. The variable of interest is the country-specific  bank-market ratio, defined as the amount of bank assets over the stock and private bond market capitalization.
\citet{langfield2016bank} estimate the effect of the bank-market ratio on the bank systemic risk intensity controlling for bank and time fixed effects, a crisis indicator and an interaction term between this indicator and the bank-market ratio.

Since a negative systemic risk intensity does not contribute to systemic risk, the outcome variable is censored at zero. \citet{langfield2016bank} therefore use \citeauthor{honore1992trimmed}'s (\citeyear{honore1992trimmed}) panel tobit specification in their Table 2. 
As shown in Section \ref{ssec:MC Honore}, the panel tobit uses many symmetric moment conditions to estimate the effect of interest. In this section we use the same moment conditions in identification robust tests. We the model,
\begin{equation}
    \begin{split}
        y_{it}=\max\{0,\alpha_{i}+\nu_{t}+x_{it}\beta_1+d_{it}\beta_2+x_{it}d_{it}\beta_3+\varepsilon_{it}\},
    \end{split}
\end{equation}
where $y_{it}$ is the censored systemic risk intensity of bank $i$ in year $t$, $\alpha_i$ and $\nu_{t}$ are bank and year fixed effects, $x_{it}$ is the bank-market ratio of the country that bank $i$ is in at time $t$, $d_{it}$ is a crisis indicator, and $\varepsilon_{it}$ is an error term that is conditionally exchangeable over time.

We take differences of the observations per bank to get rid of the bank fixed effects. We then calculate the symmetrically distributed events as in Section \ref{ssec:MC Honore}, which we multiply with the corresponding difference of the bank-market ratio, $x_{it}$, the crisis dummy, $d_{it}$, and their interaction. Since there may be dependence between the moment conditions per bank, we cluster the moment conditions on the bank level. This leaves us with a total of $12$ moment conditions, which is large relative to the 467 banks in the sample.

The AR tests jointly test hypothesized values for all parameters in the model. 
As we are mainly interested in $\beta_{1}$ and $\beta_{3}$, the coefficients on the bank market ratio and its interaction with the crisis dummy, 
we conduct a grid search only over $\beta_{1}$ and $\beta_{3}$, and minimize the AR statistics over all other parameters, which is computationally more efficient \citep{chernozhukov2009finite}.

\begin{figure}[t]
    \centering
    \caption{Confidence regions for bank-market ratio on systemic risk.}
    \label{fig:LP16}
    \includegraphics{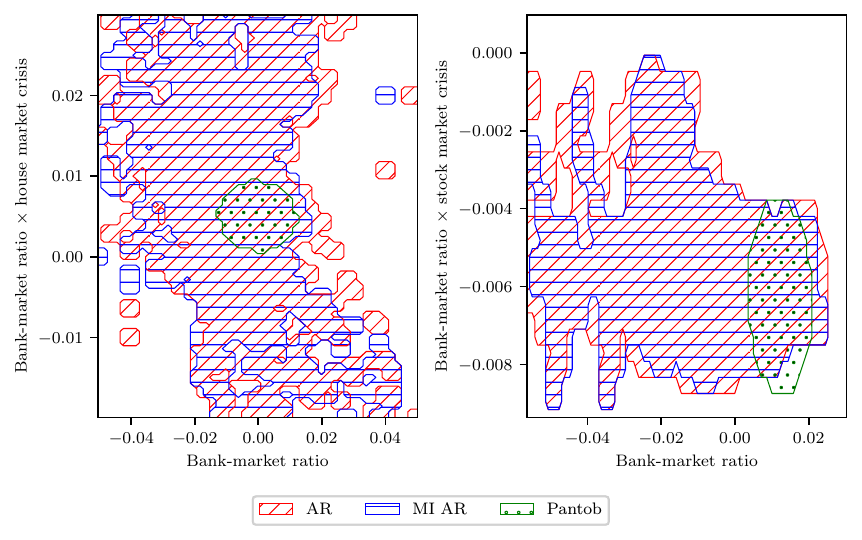}
    \begin{minipage}{\textwidth}\footnotesize\textit{Note:}
   $95\%$ confidence regions for the effect of bank-market ratio and its interaction with a crisis indicator on systemic risk, based on the fixed-$k$ Anderson--Rubin test (AR), the many instrument Anderson--Rubin test (MI AR) and the panel tobit estimator by \citet{honore1992trimmed} (Pantob). The confidence regions of the Anderson--Rubin tests are smoothed over a $3\times 3$ grid to mitigate the effect of numerical optimization. The data and model come from \citet{langfield2016bank}. The 3981 and 3909 observations are clustered on the bank level in 467 clusters. There are 12 moment conditions.
    \end{minipage}
\end{figure}

Figure~\ref{fig:LP16} shows the $95\%$ confidence regions for the effect of bank-market ratio and its interaction with a crisis dummy on systemic risk intensity, based on inverting the AR and MI AR statistics and \citeauthor{honore1992trimmed}'s (\citeyear{honore1992trimmed}) panel tobit estimator. In the left panel a crisis is defined as a year in which a country's real house prices drop by at least $10\%$, wheres in the right panel a crisis is defined as a year in which a country's real stock prices drop by at least $20\%$. Given that in the simulations in Section~\ref{ssec:MC Honore}, the AR statistic that uses all moment conditions has lower power compared to the AR statistic that only uses the reflection invariant moment conditions, we limit ourselves to the latter. To mitigate the effect of possible inaccuracies in the numerical optimization, we exclude a value $(\beta_1,\beta_3)$ from the confidence region only when the minimum of all statistics in a $3\times 3$ grid around the hypothesized point is above the critical value. This enlarges the confidence regions, but makes them more reliable. The confidence region for the panel tobit estimator is not smoothed and matches the results of columns I and III in Table 2 of \citet{langfield2016bank}.

The figure shows that the confidence regions from the AR and MI AR statistics are considerably larger than those obtained with the panel tobit estimator. When the crisis dummy captures a decrease in house prices, the robust confidence regions both contain $(0,0)$, whereas the confidence region for the panel tobit estimator excludes a zero effect for the interaction term. In contrast, when the crisis dummy captures a stock crisis, the robust confidence regions only contain zero for the bank-market ratio. The interaction term is significantly different from zero, indicating that countries with a high bank-market ratio have lower financial stability during crises, than countries with a low bank-market ratio.

Importantly, for the grids considered here, the confidence regions of the AR are larger than those of the MI AR. For the housing and stock market crises, the AR confidence regions are 10\% and 15\% larger than the MI AR confidence regions.

\section{Conclusion}\label{sec:conclusion}
We develop a new approach for identification-robust inference under many moment conditions based on the continuous updating AR GMM statistic. Since the presence of a large-dimensional weighting matrix prohibits the use of conventional asymptotic approximations, we show how reflection invariance in the moment conditions can be used to establish asymptotic results. The many moment condition correction to the variance of the AR statistic negative, indicating that conventional approximations lead to overly wide confidence intervals. We also provide an analysis of the linear IV model that provides conditions under which the same is true for the score statistic.
Monte Carlo simulations further show close to nominal size of the developed procedures regardless of the identification strength and the number of moment conditions, as well as an increase in power under many moment conditions. Empirically, we use the corrected AR test to construct confidence regions for effect of concentration of financial activities in banks on systemic risk intensity and find evidence for a negative effect of bank concentration on systemic risk during stock market crises.

\section{Acknowledgements}
We thank the editor Michael Jansson, the associate editor and two referees for their detailed comments. We also thank Federico Crudu, Paul Elhorst, Jinyong Hahn, Bruce Hansen, Art\={u}ras Juodis, Lingwei Kong, Nick Koning, Zhipeng Liao, Xinwei Ma, Sophocles Mavroeidis, Jack Porter, Andres Santos, Xiaoxia Shi, Yixao Sun, Aico van Vuuren, and Tom Wansbeek for many insightful comments. Tom Boot's work was supported by the Dutch Research Council (NWO) under grant VI\allowbreak.Veni\allowbreak.201E\allowbreak.11. 

\bibliographystyle{apalike}

\bibliography{literature}

\newpage

\begin{appendices}
\numberwithin{equation}{section}
\numberwithin{figure}{section}
\numberwithin{table}{section}

\section{Proofs}
\subsection{Proof of Theorem \texorpdfstring{\ref{thm:main}}{2}}\label{sapp:proof main}
For the proof of Theorem \ref{thm:main} we use the following results.

\begin{lemma}\label{thm:expandvar} For $\mathcal{J}=\{\varepsilon_{i},\bs Z_{i}'\}_{i=1}^{n}$ and under Assumption \ref{ass:model},  $\E[S_{l,r}(\bs\beta_{0})|\mathcal{J}]=0.$
	The $(l_{1},l_{2})^{\text{th}}$ element of the conditional variance matrix is 
		\begin{equation}\label{eq:var}
	\Omega_{l_{1} l_{2}}(\bs\beta_{0})=\E\big[n\cdot S_{l_{1},r}(\bs\beta_{0})S_{l_{2},r}(\bs\beta_{0})\big|\mathcal{J}\big] =\Omega_{l_{1} l_{2}}^{L}(\bs\beta_{0}) + \Omega_{l_{1} l_{2}}^{H}(\bs\beta_{0}),
	\end{equation}
	where $\Omega^{L}_{l_{1} l_{2}}(\bs\beta_{0}) =\Omega^{L,z}_{l_{1} l_{2}}(\bs\beta_{0}) +\Omega^{L,a}_{l_{1} l_{2}}(\bs\beta_{0}) +\Omega^{L,u}_{l_{1} l_{2}}(\bs\beta_{0})$ with
	\begin{align*}
        \Omega^{L,z}_{l_{1} l_{2}}(\bs\beta_{0})&=\frac{1}{n}\bar{\bs z}_{(l_{1})}'[(\bs I-\bs D_{P})\bs V(\bs I-\bs D_{P})+\bs D_{P}\bs D_{V}(\bs I-2\bs D_{P})+\bs V\odot\bs P\odot\bs P]\bar{\bs z}_{(l_{2})},\\
        \Omega^{L,a}_{l_{1} l_{2}}(\bs\beta_{0})&= \frac{1}{n}\bs a_{(l_{1})}'(\bs D_{P}-\bs P\odot\bs P)\bs a_{(l_{2})},\\
        \Omega^{L,u}_{l_{1} l_{2}}(\bs\beta_{0})&=\frac{1}{n}\tr(\bs D_{\Sigma^{U}(l_{1},l_{2})}\bs D_{V}(\bs I-\bs D_{P})),
        \intertext{and $\Omega_{l_{1} l_{2}}^{H}(\bs\beta_{0})= \Omega^{H,z}_{l_{1} l_{2}}(\bs\beta_{0}) +\Omega^{H,a}_{l_{1} l_{2}}(\bs\beta_{0}) +\Omega^{H,u}_{l_{1} l_{2}}(\bs\beta_{0})$ with}
        \Omega^{H,z}_{l_{1} l_{2}}(\bs\beta_{0})
        &=\frac{1}{n}\bar{\bs z}_{(l_{1})}'(\bs V\odot \bs W)\bar{\bs z}_{(l_{2})},\\
        \Omega^{H,a}_{l_{1} l_{2}}(\bs\beta_{0})&=
        -\frac{2}{n}\bs a_{(l_{1})}'(\bs D_{P}-\bs P\odot\bs P)^2\bs a_{(l_{2})},\\
          \Omega^{H,u}_{l_{1} l_{2}}(\bs\beta_{0})&=
        -\frac{2}{n}\tr(\bs D_{\Sigma^{U}(l_{1},l_{2})}(\bs V\bs D_{\varepsilon}\odot\bs V\bs D_{\varepsilon})(\bs I-2\bs D_{P}+\bs P\odot\bs P)),
	\end{align*}
	with $\bs D_{\Sigma^{U}(l_{1},l_{2})}$ an $n\times n$ diagonal matrix with the $i^{\text{th}}$ diagonal element equal to $\cov( u_{i l_{1}}, u_{i l_{1}}|\mathcal{J})$ and the elements of $(\bs W\odot \bs V)$ are defined as
    \begin{align}
        [\bs V\odot \bs W]_{i_{1} i_{2}} &= V_{i_{1} i_{2}}\big[(P_{i_{1} i_{1}}P_{i_{2} i_{2}}+P_{i_{1} i_{2}}^2)(3-4P_{i_{1} i_{1}} -4P_{i_{2} i_{2}})\nonumber\\
        &\qquad\quad -2P_{i_{1} i_{1}} -2P_{i_{2} i_{2}}+2(P_{i_{1} i_{1}} +P_{i_{2} i_{2}})^2\big],\nonumber\\
        [\bs V\odot\bs W]_{i_{1} i_{1}}&=-2V_{i_{1} i_{1}}P_{i_{1} i_{1}}(1-2P_{i_{1} i_{1}})-2\sum_{i_{2}=1}^{n}V_{i_{1} i_{2}}^2\varepsilon_{i_{2}}^2P_{i_{1} i_{2}}^2.\label{eq:WdotV}
    \end{align}
\end{lemma}

\begin{proof}
    See Supplementary Appendix \ref{sapp:proof expandvar}.
\end{proof}

For the next Lemma define the conditional variance  $\bs\Sigma_{n}(\bs\beta_0)$ as in Theorem~\ref{thm:main}
with $\sigma_n^2$ from \eqref{eq:sigman}, $\bs\Omega(\bs\beta_0)$ from Lemma \ref{thm:expandvar} and \begin{equation}\label{eq:covariance AR score}
    \big[\bs\Sigma_n(\bs\beta)\big]_{l+1,1}=\big[\bs\Sigma_n(\bs\beta)\big]_{l+1,1}=2/\sqrt{nk}\tr(\bs\Psi^{(l)}\odot\bs P),
\end{equation}
with $\bs\Psi^{(l)}=\bs M\bs D_{a_{(l)}}\bs P$ and $\bs M=\bs I-\bs P$.

Estimate $\bs\Sigma_n(\bs\beta_0)$ with 
\begin{equation}\label{eq:varestimator}
	\hat{\bs\Sigma}_n(\bs\beta) =
	\begin{pmatrix}
		\hat\sigma_{n}^2(\bs\beta) & \big[\hat{\bs\Sigma}_n(\bs\beta)\big]_{2:p+1,1}' \\
	    \big[\hat{\bs\Sigma}_n(\bs\beta)\big]_{2:p+1,1}&\hat{\bs\Omega}(\bs\beta)
	\end{pmatrix},
\end{equation}
with for the variance of the AR statistic $\hat{\sigma}_n^2(\bs\beta)=2k^{-1}(k-\bs\iota'\bs D_{P}^2\bs\iota)$ and to lighten the notation $\bs D_{P}$ is short for $\bs D_{P(\beta)}$. Furthermore,
\begin{align}
	    \hat{\Omega}_{l_{1} l_{2}}^{L}(\bs\beta)& = \frac{1}{n}\bs x_{(l_{1})}'(\bs I-\bs D_{P\iota})\bs V(\bs I-\bs D_{P\iota})\bs x_{(l_{2})},\nonumber\\
	    \hat{\Omega}_{l_{1} l_{2}}^{H}(\bs\beta)&=\frac{1}{n}\bs x_{(l_{1})}'[7\bs D_{P}\dot{\bs V}\bs D_{P}-4\bs D_{P}^2\dot{\bs V}\bs D_{P}-4\bs D_{P}\dot{\bs V}\bs D_{P}^2+3\dot{\bs V}\odot\bs P\odot\bs P\-4\bs D_{P}(\dot{\bs V}\odot\bs P\odot\bs P)\nonumber\\
	    &\qquad  -4(\dot{\bs V}\odot\bs P\odot\bs P)\bs D_{P}-2\bs D_{P}\dot{\bs V}-2\dot{\bs V}\bs D_{P}+2\bs D_{P}^2\dot{\bs V}+2\dot{\bs V}\bs D_{P}^2]\bs x_{(l_{2})}\nonumber\\
	    &\quad -\frac{2}{n}\bs x_{(l_{1})}'\big[\big((\bs V\bs D_{\varepsilon}\odot\bs V\bs D_{\varepsilon})(\bs P\odot\bs P)\big)\odot\bs I\big]\bs x_{(l_{2})}\label{eq:omegahat}.
\end{align}
Finally,
\begin{equation}\label{eq:omegacov}
	[\hat{\bs \Sigma}_n(\bs\beta)]_{l+1,1}
    = \frac{2}{\sqrt{n k}}\bs x_{(l)}'(\bs D_{V}-(\bs V\odot\bs P))\bs D_{P}\bs \varepsilon.
\end{equation}

The conditional unbiasedness and consistency of \eqref{eq:varestimator} for $\bs\Sigma(\bs\beta_{0})$ is stated in the following result. 
\begin{lemma}\label{thm:var unbias}
    Under Assumptions \ref{ass:model} and \ref{ass:decomposev}, $\E[\hat{\bs\Sigma}_n(\bs\beta_{0})|\mathcal{J}]=\bs\Sigma_n(\bs\beta_{0})$. Also, under Assumptions \ref{ass:model} to \ref{ass:eigbound}, $\hat{\bs\Sigma}_n(\bs\beta_0)\rightarrow_p\bs\Sigma_n(\bs\beta_{0})$.
\end{lemma}

\begin{proof}
    See Supplementary Appendix \ref{sapp:proof var unbias}.
\end{proof}

\begin{lemma}\label{lem:AR score}
    Define $\bs T(\bs\beta_{0})$ as in Theorem~\ref{thm:main}. Under Assumptions \ref{ass:model} to \ref{ass:eigbound}, when $n\rightarrow\infty$ and $k/n\rightarrow\lambda\in (0,1)$, $
		\bs\Sigma_{n}(\bs\beta_{0})^{-1/2}\bs T(\bs\beta_{0})\rightarrow_{d} N(\bs 0,\bs I_{p+1})$, 
\end{lemma}

\begin{proof}
    See Supplementary Appendix \ref{sapp:proof AR score}.
\end{proof}

The proof of Theorem \ref{thm:main} then follows from Lemma \ref{thm:var unbias} and Lemma \ref{lem:AR score} and Slutsky's theorem.

\subsection{Proof of Lemma \texorpdfstring{\ref{lem:negdef}}{1}}\label{app:proof negdef}
\begin{proof} It is clear that in Theorem \ref{thm:expandvar}, $\bs\Omega^{H,a}(\bs\beta_{0})$ is negative semidefinite. To show that $\bs\Omega^{H,u}(\bs\beta_{0})\prec \bs O$, observe that $\bs\Omega^{H,u}(\bs\beta_{0})=-\frac{2}{n}\sum_{i=1}^{n}c_{i}\bs\Sigma^{U}_{i}$ where
\begin{align*}
    c_{i}&=\bs e_{i}'(\bs D_{V}\bs D_{P}(1-2\bs D_{P})+(\bs V\bs D_{\varepsilon}\odot \bs V\bs D_{\varepsilon})(\bs P\odot \bs P))\bs e_{i}\\
    &=V_{ii}P_{ii}(1-2P_{ii}+P_{ii}^2)+\sum_{j\neq i}V_{ij}^2\varepsilon_{j}^2P_{j i}^2\geq V_{ii}P_{ii}(1-P_{ii})^2.
\end{align*}
Since by assumption $P_{ii}\leq C<1$, $\lambda_{\min}(\bs\Sigma^{U}_{i})\geq C$, and $\sum_{i=1}^{n}V_{ii}P_{ii}>0$, we have that $\bs\Omega^{H,u}(\bs\beta_{0})\leq -(C\sum_{i=1}^{n}V_{ii}P_{ii})\cdot \bs I_{p}\prec\bs O$. 

It now suffices to prove that $\bar{\bs Z}'(\bs V\odot\bs W)\bar{\bs Z}\preceq \bs O.$
We first consider the diagonal terms of $\bs V\odot \bs W$. If $V_{ii}=0$, the result is trivial. We therefore assume $V_{ii}>0$. First, we make the observation that $\bs\Omega_{jj}(\bs\beta_{0})\geq 0$, since it is the expectation of the squared score. This also holds in a model with $\bs\Pi=\bs O$, and $\bs a_{(l)}=0$ for all $l$ and $\bs\Sigma^{U}(l_{1}, l_{2})=1$ for some $l_{1}=l_{2}$ and $0$ for all $l_{1}\neq l_{2}$. We then deduce that for all $i$,
\begin{equation}\label{eq:usefuldiagonal}
   V_{ii}-3P_{ii}V_{ii}+4V_{ii}P_{ii}^2-2\sum_{j}V_{ij}^2\varepsilon_{j}^2P_{jj}^2\geq 0.
\end{equation}
Now, $w_{ii}V_{ii} = -2(V_{ii}P_{ii}-2V_{ii}P_{ii}^2+\sum_{j}V_{ij}^2\varepsilon_{j}^2P_{ij}^2)$. To show that $w_{ii}\leq 0$, we need to show that $-\frac{1}{2}w_{ii}V_{ii}=V_{ii}P_{ii}-2V_{ii}P_{ii}^2+\sum_{j}V_{ij}^2\varepsilon_{j}^2P_{ij}^2\geq 0$. Using \eqref{eq:usefuldiagonal}, we have that $
   -\frac{1}{2}w_{ii}V_{ii}\geq V_{ii}\left(P_{ii}-2P_{ii}^2+\frac{1}{2}-\frac{3}{2}P_{ii}+2P_{ii}^2\right)= \frac{1}{2}V_{ii}(1-P_{ii}).$ 
Since $V_{ii}> 0$, we can now conclude that $w_{ii}< 0$ if $P_{ii}\leq C< 1$, which holds by Assumption \ref{ass:AR}.
   
Consider now the off-diagonal terms of $\bs V\odot \bs W$. For $i\neq j$, we have $
        w_{i j}=(P_{i i}P_{j j}+P_{i j}^2)(3-4(P_{i i}+P_{j j}))-2(P_{i i}+P_{j j})+2(P_{i i}+P_{j j})^2.$ 
Suppose that $P_{i i}+P_{j j}< 3/4$. Then, $P_{i i}P_{j j}+P_{i j}^2\leq (P_{i i}+P_{j j})^2$ by using that $P_{i j}^2\leq P_{i i }P_{j j}$ and $2P_{i i }P_{j j}\leq (P_{i i }+P_{j j})^2$. Defining $x_{i j}=P_{i i }+P_{j j }$, we then have that $w_{i j}\leq -(4x_{i j}^2-5x_{i j}+2)x_{i j}< 0$. Now suppose that $P_{i i }+P_{j j}\geq 3/4$, then $
    w_{i j}\leq P_{i i }P_{j j }(3-4(P_{i i }+P_{j j }))-2(P_{i i }+P_{j j})+2(P_{i i }+P_{j j })^2.$
We can verify that $w_{i j }\leq 0$ if $\max_{i=1,\ldots,n}P_{ii}\leq 1/8[3+2\sqrt{2}+\sqrt{3}(4\sqrt{2} - 5)^{1/2}]\approx 0.904$.

Define $\tilde{\bs\Pi} = (\bs Z'\bs D_{\varepsilon}^2\bs Z)^{1/2}\bs\Pi$ and $\tilde{\bs z}_{i} =(\bs Z'\bs D_{\varepsilon}^2\bs Z)^{-1/2}\bs z_{i}$. Then,
\begin{align*}
   \bs \Pi'\sum_{i,j}\bs z_{i }V_{i j}\bs z_{j}'\bs \Pi\cdot w_{i j }&=\bs\Pi'\sum_{i ,j}\bs z_{i }\bs z_{i }'(\bs Z'\bs D_{\varepsilon}^2\bs Z)^{-1}\bs z_{j}\bs z_{j}'\bs\Pi\cdot w_{i j}\\
   &=\tilde{\bs\Pi}'\left(\sum_{i=1}^{n}\tilde{\bs z}_{i}\tilde{\bs z}_{i}'\tilde{\bs z}_{i}\tilde{\bs z}_{i}'w_{ii}+\sum_{i > j}(\tilde{\bs z}_{i }\tilde{\bs z}_{i }'\tilde{\bs z}_{j}\tilde{\bs z}_{j}'+\tilde{\bs z}_{j}\tilde{\bs z}_{j}'\tilde{\bs z}_{i }\tilde{\bs z}_{i }')\cdot w_{i j}\right)\tilde{\bs\Pi}.
\end{align*}
The first term within the brackets is a negatively weighted sum of symmetric positive semidefinite matrices and hence, negative semidefinite. The second term is a nonpositively weighted sum of symmetric positive semidefinite matrices, so that the resulting matrix is symmetric and negative semidefinite. We conclude that $\bar{\bs Z}'(\bs W\odot \bs V)\bar{\bs Z}\preceq \bs O$. This completes the proof of Lemma \ref{lem:negdef}.\end{proof}

\section{Additional numerical results}

\subsection{Implementation of tests in the IVQR in Section~\ref{ssec:MC IVQR}}\label{sapp:details MC IVQR}
\paragraph{Subset score test}
Due to the indicator function the moment conditions for the IVQR model in Example \ref{ex:IVQR} are nonsmooth, which hinders the direct application of the \citeauthor{kleibergen2005testing}'s (\citeyear{kleibergen2005testing}) score test as the expectation of the Jacobian of the moment conditions is not always available. \citet{jun2008weak} notes however, that instead of the expectation of the Jacobian one can use the Jacobian of the expectation of the moment conditions, which for the IVQR model is available. In particular, if the quantile function is linear in the regressors $q(\bs x,\tau;\bs\theta)=\bs x_{i}'\bs\theta(\tau)$, and $\bs g_i(\bs\theta(\tau))=(\tau-\mathbbm{1}\{y_i\leq\bs x_{i}'\bs\theta(\tau)\})\bs\Psi(\bs z_{i})$ then $\frac{\partial\E[\bs g_{i}(\bs\theta(\tau))]}{\partial\bs\theta'}\big\rvert_{\theta=\theta_0}=\E[\bs x_if_{\varepsilon_{\tau,i}}(0|\bs x_i,\bs z_i)\bs\Psi(\bs z_i)]$ where $f_{\varepsilon_{\tau,i}}(\cdot|\bs X_i,\bs Z_i)$ is the conditional density of $\varepsilon_i(\tau)=Y_{i}-\bs X_i'\bs\beta(\tau)$. This Jacobian of the expectation can then be used to construct a score test.

As a second contribution \citet{jun2008weak} shows how the effect of exogenous regressors that are not of interest can be partialled out to obtain a more powerful test. We conjecture that his results, which assumed a single endogenous regressor, can be generalized to the case with more endogenous regressors.

Staying close to \citet{jun2008weak}, we assume that for each observation $i=1,\dots,n$ the following moment conditions holds
\begin{equation}
    \begin{split}
        \E(\mathbbm{1}\{Y_{i}\leq\bs D_{i}'\bs\alpha+\bs X_{i}'\bs\beta\}-\tau\}|\bs Z_{i},\bs X_{i})=\bs 0,
    \end{split}
\end{equation}
where $Y_{i}$ is the outcome variable, $\tau$ is the quantile of interest, $\bs D_{i}\in\mathbb{R}^{p_{d}}$ are the endogenous regressors, $\bs X_{i}\in\mathbb{R}^{p_{x}}$ are the exogenous regressors and $\bs Z_{i}\in\mathbb{R}^k$ are the instruments.

Let $\hat{\bs\beta}(\tau,\bs\alpha)$ be an estimate of the quantile regression of $Y_{i}-\bs\alpha\bs D_{i}$ on $\bs X_{i}$. Next, define the following objects
\begin{equation}
    \begin{split}
        \hat{\bs m}_c(\bs\alpha(\tau),\bs\beta)&=\frac{1}{n}\sum_{i=1}^n(\mathbbm{1}\{Y_{i}\leq\bs D_{i}'\bs\alpha+\bs X_{i}\bs\beta\}-\tau)\bs Z_{i}\\
        \hat{\bs H}_{1}(\bs\alpha(\tau))&=\frac{1}{nh_n}\sum_{i=1}^n\bs Z_{i}\bs X_{i}'k(\frac{\bs D_{i}'\bs\alpha+\bs X_{i}'\hat{\bs\beta}(\tau,\bs\alpha)-Y_{i}}{h_{n}})\\
        \hat{\bs H}_{x}(\bs\alpha(\tau))&=\frac{1}{nh_n}\sum_{i=1}^n\bs X_{i}\bs X_{i}'k(\frac{\bs D_{i}'\bs\alpha+\bs X_{i}'\hat{\bs\beta}(\tau,\bs\alpha)-Y_{i}}{h_{n}})\\
        \hat{\hat{\bs\Gamma}}_{n,h_n}(\bs\alpha(\tau),\hat{\bs\beta}(\tau,\bs\alpha(\tau)))&=\frac{1}{nh_n}\sum_{i=1}^n\hat{\bs Q}_{i}\bs D_{i}'k(\frac{\bs D_{i}'\bs\alpha+\bs X_{i}'\hat{\bs\beta}(\tau,\bs\alpha)-Y_{i}}{h_{n}}),
    \end{split}
\end{equation}
where $h_{n}$ is a bandwidth parameter, $k$ a kernel and $\hat{\bs Q}_{i}=\bs Z_i-\hat{\bs H}_{1}(\bs\alpha(\tau))\hat{\bs H}_{x}(\bs\alpha(\tau))^{-1}\bs X_i$.

Also, let $\hat{\bs V}(\bs\alpha(\tau))$ be the sample variance of $\hat{\bs Q}_{i}(\mathbbm{1}\{Y_{i}-\bs D_{i}'\bs\alpha(\tau)-\bs X_{i}'\hat{\bs\beta}(\tau,\bs\alpha)\leq0\}-\tau)$ and $\bs P$ the projection matrix onto the column space of $\hat{\bs V}^{-1/2}(\bs\alpha(\tau))\hat{\hat{\bs\Gamma}}_{n,h_n}(\bs\alpha(\tau),\hat{\bs\beta}(\tau,\bs\alpha(\tau)))$. Finally, we conjecture that, under assumptions similar to those in \citet{jun2008weak}, 
$\hat{K}_{n}(\bs\alpha(\tau))=(\sqrt{n}\hat{\bs m}_{c}(\bs\alpha(\tau),\hat{\bs\beta}(\tau,\bs\alpha(\tau))))'\hat{\bs V}^{-1/2}\bs P\hat{\bs V}^{-1/2}(\sqrt{n}\hat{\bs m}_{c}(\bs\alpha(\tau),\hat{\bs\beta}(\tau,\bs\alpha(\tau))))\rightarrow_{d}\chi^2(p_{d})$.

As \citet{jun2008weak} we use $h_{n}=n^{-1/5}$ and a Gaussian kernel for $k$.

\paragraph{GMM}
The nonsmoothness of the moment conditions also makes that one cannot expand the moment conditions to derive a limiting distribution of the GMM estimator. \citet{andrews1994empirical} shows that, under certain conditions, a limiting distribution still obtains when the expectation of the moment conditions is differentiable. \citet{arellano2009gmm} uses this theory to derive the limiting distribution of the GMM estimator in an unconditional IVQR. Combining his derivations with those in \citet{jun2008weak}, we obtain the limiting distribution of the GMM estimator in a conditional IVQR.

As above, for a linear quantile function, $\frac{\partial\E[\bs g_{i}(\bs\theta(\tau))]}{\partial\bs\theta'}\big\rvert_{\theta=\theta_0}=\E[\bs x_if_{\varepsilon_{\tau,i}}(0|\bs x_i,\bs z_i)\bs\Psi(\bs z_i)]$. Let $\bs J=\sum_{i=1}^n\frac{\partial\E(\bs g_{i}(\bs\theta))}{\partial\bs\theta'}\big\rvert_{\theta=\theta_0}$ and $\bs\Omega$ the variance of the moment conditions. The two or higher step optimal GMM estimator then has in the limit a normal distribution with variance $\bs J'\bs\Omega\bs J$. 

Since we have independent observations in our simulation, we estimate $\bs\Omega$ with $\hat{\bs\Omega}=\sum_{i=1}^n\bs g_{i}(\bs\theta(\tau))\bs g_{i}(\bs\theta(\tau))'$. Like \citet{jun2008weak} we estimate $\bs J$ with $\hat{\bs J}=\frac{1}{h_n}\sum_{i=1}^n\bs X_i\bs Z_i'k(\frac{\bs Y-\bs X_{i}'\bs\beta(\tau)}{h_{n}})$ where $h=n^{-1/5}$ and $k$ a Gaussian kernel. Then the estimated limiting variance of the GMM estimator is $\hat{\bs J}'\hat{\bs\Omega}\hat{\bs J}$.

\paragraph{Inverse quantile regressor}
We implement the inverse quantile regressor (IQR) by \citet{chernozhukov2006instrumental} with the STATA package \texttt{ivqregress}. We use the default options, except when evaluating the power. In that case we set the grid over which the nuisance parameters are evaluated, equal to the nuisance parameter grid used for the AR and MI AR tests.

\paragraph{Smoothed IVQR}
We also implement the smoothed IVQR by \citet{kaplan2017smoothed} with \texttt{ivqregress} package and all default options.

\subsection{Size and power for median IV regression}\label{sapp:medianIV}
We start by testing the joint hypothesis $H_0\colon\theta_{1}(1/2)=-1+\Phi^{-1}(1/2)=-1$ and $\theta_{2}(1/2)=1$ at a $5\%$ significance level, for $k$ ranging from 3 to 30. To more closely match the DGP from \citet{chernozhukov2006finite}, we generate the instruments in this section as $\bs Z_i\sim N(0,\bs I)$. Figure \ref{fig:IVQR size k} shows the results over $1,000$ data sets from this DGP. The left panel shows the rejection rates for the fixed-$k$ AR test, the MI AR test and the finite sample correct test from \citet{chernozhukov2006finite}. For $k=3$ we see that all tests are size correct. When the number of instruments increases, the rejection rates of the fixed-$k$ AR drop below the nominal size.

The right panel of Figure \ref{fig:IVQR size k} shows the rejection rates for the subset score test by \citet{jun2008weak} and tests based on a GMM estimator, the IQR estimator by \citet{chernozhukov2006instrumental} and the smoothed IVQR estimator by \citet{kaplan2017smoothed}. The rejection rates of all these tests increase with the number of moment conditions.

\begin{figure}
    \centering
    \caption{Size in the median IV regression simulation.}
    \label{fig:IVQR size k}    \includegraphics{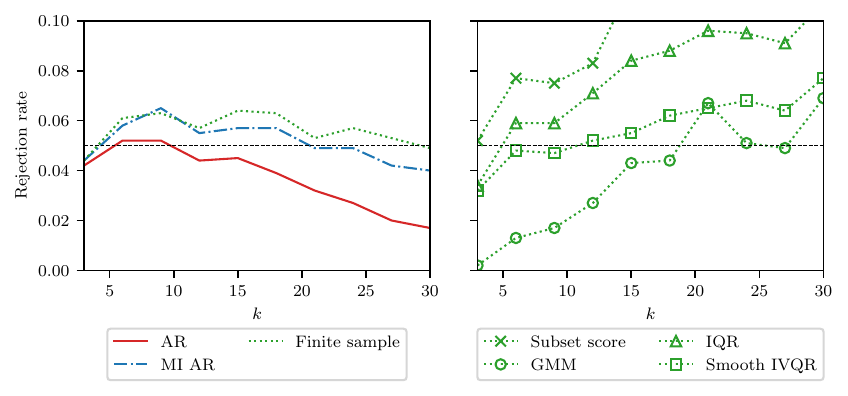}
     \begin{minipage}{\textwidth}\footnotesize\textit{Note:} Size when testing $H_0\colon\theta_{1}(1/2)=-1+\Phi^{-1}(1/2)=-1$ and $\theta_{2}(1/2)=1$ at $\alpha=0.05$ based on the fixed-$k$ Anderson--Rubin test (AR), the many instrument Anderson--Rubin test (MI AR), the finite sample test by \citet{chernozhukov2009finite} (Finite sample), the subset score test by \citet{jun2008weak} (Subset score), and tests based on a GMM estimator (GMM), the inverse quantile regressor estimator by \citet{chernozhukov2006instrumental} (IQR) and the smoothed IVQR estimator by \citet{kaplan2017smoothed} (Smooth IVQR). $k$ denotes the number of instruments. Sections \ref{ssec:MC IVQR} and \ref{sapp:medianIV} describe the DGP.
    \end{minipage}
\end{figure}

Next, we investigate the power of the fixed-$k$ AR test, the MI AR test and the finite sample correct test by \citet{chernozhukov2009finite}. We follow \citet{chernozhukov2006finite} and consider rejection rates of the marginal hypothesis $H_0\colon\theta_{2}(1/2)=\theta_{2}^*(1/2)$, where $\theta_{2}^*(1/2)$ varies between $-4$ and $6$. We reject the marginal hypothesis only if the joint test rejects for all values of $\theta_{1}(1/2)$. As a consequence the test of the marginal hypothesis will be conservative.

Figure \ref{fig:IVQR power} shows the power curves when testing at a $5\%$ significance level for $k=10$ and $k=30$ over $10,000$ data sets. Clearly, the MI AR test has higher power than the fixed-$k$ AR test and especially so when the number of instruments is large. Note furthermore that the MI AR has comparable power to the test by \citet{chernozhukov2009finite}. Unlike the latter, the former does not require any simulations to determine the critical values, which makes it a computationally friendly alternative.

\begin{figure}
    \centering
    \caption{Power in the median IV regression simulation.}
    \label{fig:IVQR power}
    \includegraphics{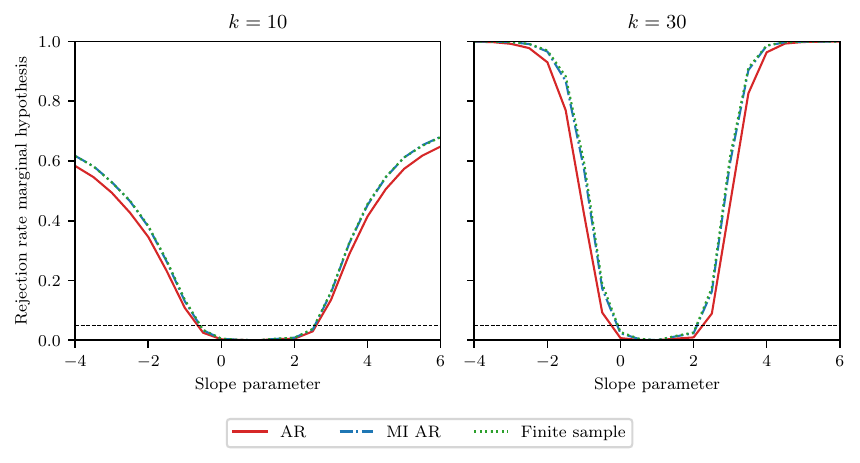}
    \begin{minipage}{\textwidth}\footnotesize\textit{Note:}     Marginal power when testing the slope parameter $H_{0}\colon\theta_{2}(1/2)=\theta^*_2(1/2)$ at $\alpha=0.05$ where $\theta_{2}^*(1/2)$ varies between -4 and 6, based on the fixed-$k$ Anderson--Rubin test (AR), the many instrument Anderson--Rubin test (MI AR) and the finite sample test by \citet{chernozhukov2009finite} (Finite sample). $k$ denotes the number of instruments. Section \ref{ssec:MC IVQR} describes the DGP.
    \end{minipage}
\end{figure}

\subsection{Skewed moment conditions}\label{sapp:skewed}
To generate skewed moment conditions for the quantile IV model in Section~\ref{ssec:MC IVQR} and the linear IV model in Section~\ref{ssec:MC linear IV}, we draw from a centered and scaled Gamma distribution with shape parameter $1/\zeta$, rate parameter $\sqrt{1/(\sigma^2\zeta)}$ and we subtract $\sqrt{\sigma^2/\zeta}$, such that the distribution has mean zero, variance $\sigma^2$ and skewness governed by $\zeta$. 
We set $\sigma^2$ for the different random variables to the same values as in the symmetric case. Note that for $\zeta$ close to zero, the asymmetry is small, whereas for $\zeta$ large the asymmetry is large. In Figure \ref{fig:histogram gamma} we plot empirical densities for different values of $\zeta$ based on $1,000,000$ draws of this distribution.
\begin{figure}[t]
    \centering
    \caption{Empirical densities of the centered and scaled Gamma distribution.}
    \label{fig:histogram gamma}
    \includegraphics{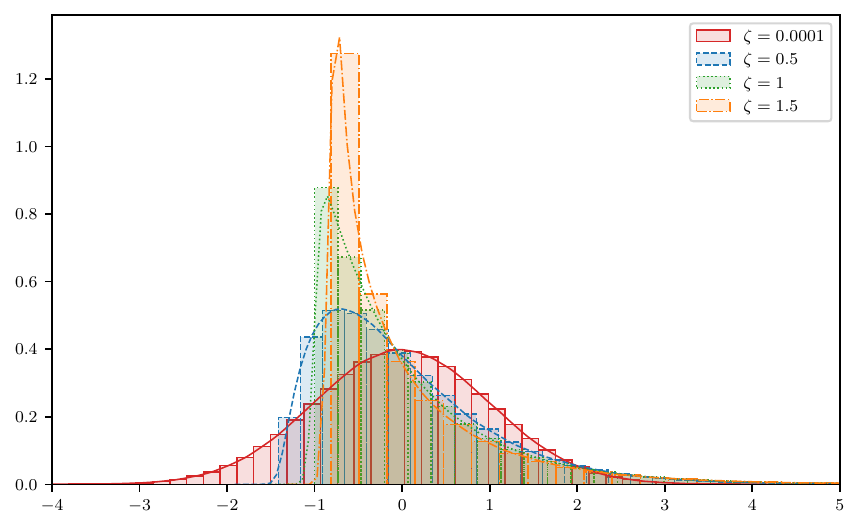}
    \begin{minipage}{\textwidth}\footnotesize\textit{Note:} Empirical densities of the centered and scaled Gamma distribution for different values of skewness parameter $\zeta$.
    \end{minipage}
\end{figure}
\begin{figure}[t]
    \centering
    \caption{Size in the linear IV simulation with skewed moment conditions.}
    \label{fig:size skewed}
    \includegraphics{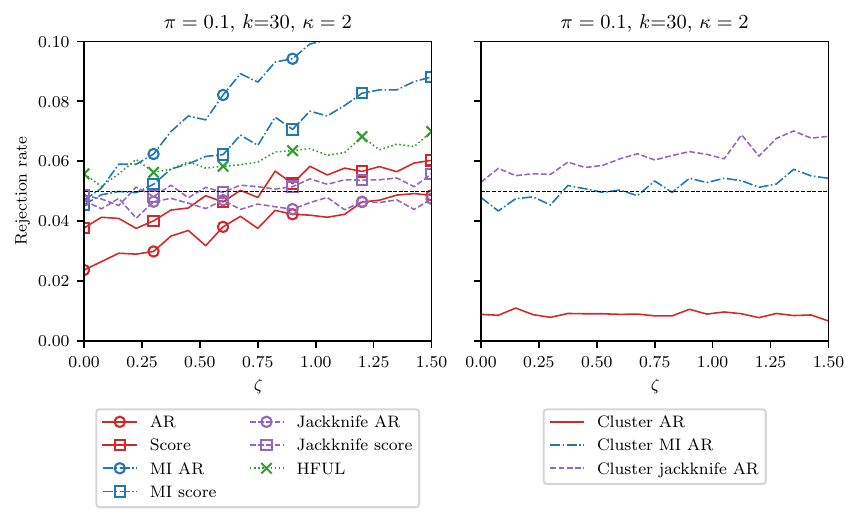}
    \begin{minipage}{\textwidth}\footnotesize\textit{Note:} Size when testing $H_0\colon\beta=0$ at $\alpha=0.05$ based on the fixed-$k$ (cluster) Anderson--Rubin test (AR), the fixed-$k$ score test (Score), the (cluster) many instrument Anderson--Rubin test (MI AR), the many instrument score test (MI score), the (cluster) jackknife Anderson--Rubin test (Jackknife AR, without cross-fit variance estimator), the jackknife score test (Jackknife score, without cross-fit variance estimator) and HFUL (HFUL). In the left panel the observations are independent. In the right panel the 800 observations are clustered in 100 unbalanced clusters. $k$ denotes the number of instruments, $\kappa$ the heteroskedasticity parameter, $\pi$ the instrument strength parameter and the $\zeta$ the skewness parameter. Section \ref{ssec:MC linear IV} and Appendix \ref{sapp:skewed} describe the DGP and Appendix \ref{sapp:details MC} the implementation of the tests.
    \end{minipage}
\end{figure}
\subsubsection{Size in the linear IV under skewed moment conditions}
The left panel of Figure \ref{fig:size skewed} shows the rejection rate of the tests for independent data over 10,000 draws from this DGP for different values of $\zeta$. None of the tests is insensitive to skewness of the moment conditions. However, the rejection rates of the tests developed in this paper rise the fastest when $\zeta$ increases and become oversized. The right panel of the same figure, in which the data are clustered, shows that although the rejection rates of the cluster many instrument AR and cluster jackknife AR seem to rise with the skewness, the cluster many instrument AR is out of the three AR tests closest to size correct.

\subsection{Implementation of tests in the linear IV in Section~\ref{ssec:MC linear IV} }\label{sapp:details MC}
\paragraph{Fixed-\texorpdfstring{$k$}{k} AR test}
The general GMM fixed-$k$ AR test is explained in Section \ref{sec:group}. In the linear IV model as considered in Section \ref{ssec:MC linear IV}, we simplify the weighting matrix to $\bs Z'\bs D_{\varepsilon}^2\bs Z$ for independent data and to $\bs Z'\bs B_{\varepsilon}\bs B_{\varepsilon}'\bs Z$ for clustered data, where $\bs B_{\varepsilon}$ is a $H\times n$ matrix with in column $h$ in the rows indexed by $[h]$ the elements of $\bs\varepsilon_{[h]}$ and zeroes elsewhere.

\paragraph{Fixed-\texorpdfstring{$k$}{k} score test}
In Section \ref{sec:score} we argue that in the linear IV model the fixed-$k$ and many instrument score statistics are identical up to additional terms in the variance estimator. We therefore implement the fixed-$k$ score statistic as $\bs S(\bs\beta)'[\hat{\bs\Omega}^{L}(\bs\beta)
]^{-1}\bs S(\bs\beta)\rightarrow_{d}\chi^2(p)$ where $p=\dim(\bs\beta)$ and $\hat{\bs\Omega}^L(\bs\beta)$ defined in \eqref{eq:omegahat}. 

\paragraph{Jackknife AR test}
Let $\bs P_{Z}=\bs Z'(\bs Z'\bs Z)^{-1}\bs Z$. We implement the jackknife AR test for the linear IV model and independent data by \citet{crudu2021inference} and \citet{mikusheva2021inference} as $
    {\bs\varepsilon'\dot{\bs P}_{Z}\bs\varepsilon}/{\sqrt{2(\bs\varepsilon\odot\bs\varepsilon)'(\dot{\bs P}_{Z}\odot\dot{\bs P}_{Z})(\bs\varepsilon\odot\bs\varepsilon)}}$. That is, we neither use the symmetric jackknife nor a cross-fit variance estimator in its implementation. We follow \citet{mikusheva2021inference} and use critical values from $2k^{-1/2}(Z_1-k)$ with $Z_1\sim\chi^2(k)$ to test.

We implement the cluster jackknife AR from \citet{ligtenberg2023inference} with the same expression as above, but with $\ddot{\bs P}_{Z}$ instead of $\dot{\bs P}_{Z}$. Here $\ddot{\bs P}_{Z}$ is $\bs P_{Z}$, but with its block diagonal elements, rather than only its diagonal elements, set to zero.

\paragraph{Jackknife score test}
We implement the jackknife AR test for the linear IV model and independent data by \citet{matsushita2020jackknife} as $\bs S_{J}(\bs\beta)'\hat{\bs\Omega}_{J}(\bs\beta)\bs S_{J}(\bs\beta)\rightarrow_{d} N(0,1),$ where $\bs S_{J}(\bs\beta)=\bs X'\dot{\bs P}_{Z}\bs\varepsilon$ and $\hat{\bs\Omega}_{J}(\bs\beta)=\bs X'\dot{\bs P}_{Z}\bs D_{\varepsilon}^2\dot{\bs P}_{Z}\bs X+\sum_{i\neq j}\bs x_{i} P_{Z,i j}\varepsilon_{i}\varepsilon_{j}P_{Z,j i}\bs x_{j}'$.

\paragraph{HFUL}
We implement the HFUL estimator from \citet{hausman2012instrumental} as $\hat{\bs \beta}_{\text{HFUL}}=(\bs X'\dot{\bs P}_{Z}\bs X-\hat{\alpha}\bs X'\bs X)^{-1}(\bs X'\dot{\bs P}_{Z}\bs y-\hat{\alpha}\bs X'\bs y)$, where $\hat{\alpha}=[\tilde{\alpha}-(1-\tilde{\alpha})C_{\text{HFUL}}/n]/[1-(1-\tilde{\alpha})C_{\text{HFUL}}/n]$ for $\tilde{\alpha}$ the minimum eigenvalue of $(\bar{\bs X}'\bar{\bs X})^{-1}(\bs X'\dot{\bs P}_{Z}\bar{\bs X})$ with $\bar{\bs X}=(\bs y, \bs X)$ and we follow \citet{hausman2012instrumental} and set $C_{\text{HFUL}}=1$. For its variance estimator we use $\hat{\bs V}=\hat{\bs H}^{-1}\hat{\bs \Sigma}\hat{\bs H}^{-1}$, with $\hat{\bs H}=\bs X'\dot{\bs P_{Z}}\bs X-\hat{\alpha}\bs X'\bs X$ and $\hat{\bs\Sigma}=\hat{\bs X}'\dot{\bs P}_Z\bs D_{\hat{\varepsilon}}^2\dot{\bs P}_Z\hat{\bs X}+\hat{\bs X}'\bs D_{\hat{\varepsilon}}(\dot{\bs P}_Z\odot \dot{\bs P}_Z)\bs D_{\hat{\varepsilon}}\hat{\bs X}$. Here $\hat{\bs\varepsilon}=\bs y-\bs X'\hat{\bs\beta}_{\text{HFUL}}$ and $\hat{\bs X}=\bs X-\hat{\bs\varepsilon}\hat{\bs\gamma}$ for $\hat{\gamma}=\bs X'\hat{\bs\varepsilon}/\hat{\bs\varepsilon}'\hat{\bs\varepsilon}$. As test we use a $t$-test with normal critical values.

\end{appendices}

\clearpage
\begin{appendices}
\renewcommand{\appendixname}{Supplementary Appendix}
\setcounter{section}{0}
\renewcommand{\thesection}{S\Alph{section}}
\renewcommand{\theequation}{S\Alph{section}.\arabic{equation}}

\section{Additional proofs}
\subsection{Proof of Lemma \texorpdfstring{\ref{thm:expandvar}}{2}}\label{sapp:proof expandvar}
\begin{proof}
    The score function from the continuous updating objective function is given by $
    \frac{\partial Q(\bs\beta)}{\partial \beta_{l}} = -\frac{1}{n}\bs x_{(l)}'(\bs I-\bs D_{P\iota})\bs V\bs\varepsilon.$ Under Assumption \ref{ass:main}, the score satisfies
    \begin{align*}
        \frac{\partial Q(\bs\beta)}{\partial \beta_{l}}& \overset{(d)}{=}-\frac{1}{n}(\bar{\bs x}_{(l)} + \bs D_{r}\bs D_{\varepsilon}\bs a_{(l)})'(\bs I - \bs D_{r}\bs D_{Pr})\bs V\bs D_{r}\bs\varepsilon\\
        &=-\frac{1}{n}\bar{\bs x}_{(l)}'\bs V\bs D_{\varepsilon}\bs r + \frac{1}{n}\bs r'\bs P\bs D_{r}\bs D_{\bar{x}_{(l)}}\bs V\bs D_{\varepsilon}\bs r-\frac{1}{n}\bs r'\bs D_{a_{(l)}}\bs P\bs r + \frac{1}{n}\bs r'\bs P\bs D_{a_{(l)}}\bs P\bs r\overset{(\E_{r})}{=}0.
	\end{align*}
	This proves the first statement of Theorem \ref{thm:expandvar}.

    The $(l_{1},l_{2})^{\text{th}}$ element of the conditional variance is given by
    \begin{align*}
    \E\left[n\frac{\partial Q(\bs\beta)}{\partial \beta_{l_{1}}}\frac{\partial Q (\bs\beta)}{\partial \beta_{l_{2}}}|\mathcal{J}\right]& =\E\bigg[ \underbrace{\frac{1}{n}\bs x_{(l_{1})}'\bs V\bs\varepsilon\bs\varepsilon'\bs V\bs x_{(l_{2})}}_{(I)} + \underbrace{\frac{1}{n}\bs x_{(l_{1})}'\bs D_{P\iota}\bs V\bs\varepsilon\bs\varepsilon'\bs V\bs D_{P\iota}\bs x_{(l_{2})}}_{(II)}\\
    &\quad \underbrace{-\frac{1}{n}\bs x_{(l_{1})}'\bs V\bs\varepsilon\bs\varepsilon'\bs V\bs D_{P\iota}\bs x_{(l_{2})}}_{(III)}\underbrace{-\frac{1}{n}\bs x_{(l_{1})}'\bs D_{P\iota}\bs V\bs\varepsilon\bs\varepsilon'\bs V\bs x_{(l_{2})}}_{(IV)}\big|\mathcal{J}\bigg].
    \end{align*}
    We write $(I)$--$(IV)$ in terms of the Rademacher random variables $\bs r$ from Assumption \ref{ass:main}, take the expectation over $\bs r$ using the results from Lemma \ref{thm:Rademacher} in Supplementary Appendix \ref{sapp:rademacher}, and then take the expectation over the first stage errors $\bs U$. The marker ``fixed-$k$ approximation" indicates terms that appear when we take the estimator for the variance of the score as in \citet{kleibergen2005testing}. Using that $\bs x_{(l)}  = \bar{\bs x}_{(l)}+\bs D_{\varepsilon}\bs a_{(l)}$,
    \begin{align*}
        (I)& \overset{(d)}{=} \frac{1}{n}\bar{\bs x}_{(l_{1})}'\bs V\bs D_{\varepsilon}\bs r\bs r'\bs D_{\varepsilon}\bs V\bar{\bs x}_{(l_{2})} +\frac{1}{n}\bs r'\bs D_{a_{(l_{1})}}\bs P\bs r \bs r'\bs P\bs D_{a_{(l_{2})}}\bs r + \frac{1}{n}\bar{\bs x}_{(l_{1})}'\bs V\bs D_{\varepsilon}\bs r\bs r'\bs P\bs D_{a_{(l_{2})}}\bs r\\
        &\quad +  \frac{1}{n}\bs r'\bs D_{a_{(l_{1})}}\bs P\bs r\bs r'\bs D_{\varepsilon}\bs V\bar{\bs x}_{(l_{2})} \\
        &\overset{(\E_{r})}{=}\frac{1}{n}[\bar{\bs x}_{(l_{1})}'\bs V\bar{\bs x}_{(l_{2})} -2\tr(\bs D_{P}\bs D_{a_{(l_{1})}}\bs D_{P}\bs D_{a_{(l_{2})}}) + \tr(\bs D_{a_{(l_{1})}}\bs P)\tr(\bs D_{a_{(l_{2})}}\bs P)\\
        &\quad +\tr(\bs D_{a_{(l_{1})}}\bs D_{a_{(l_{2})}}\bs P)+\tr(\bs D_{a_{(l_{1})}}\bs P\bs D_{a_{(l_{2})}}\bs P)]\\
        &\overset{(\E_{U})}{=}\frac{1}{n}[\underbrace{\bar{\bs z}_{(l_{1})}'\bs V\bar{\bs z}_{(l_{2})}+\tr(\bs D_{\Sigma^U(l_{1},l_{2})}\bs V)+\tr(\bs D_{a_{(l_{1})}}\bs D_{a_{(l_{2})}}\bs P)}_{\text{fixed-$k$ approximation}}\\
        &\quad -2\tr(\bs D_{P}\bs D_{a_{(l_{1})}}\bs D_{P}\bs D_{a_{(l_{2})}}) + \tr(\bs D_{a_{(l_{1})}}\bs P)\tr(\bs D_{a_{(l_{2})}}\bs P) +\tr(\bs D_{a_{(l_{1})}}\bs P\bs D_{a_{(l_{2})}}\bs P)].\\
        (II) & \overset{(d)}{=}\frac{1}{n}\bs r'\bs P\bs D_{r}\bs D_{\bar{x}_{(l_{1})}}\bs V\bs D_{\varepsilon}\bs r\bs r'\bs D_{\varepsilon}\bs V\bs D_{\bar{x}_{(l_{2})}}\bs D_{r}\bs P\bs r + \frac{1}{n}\bs r'\bs P\bs D_{a_{(l_{1})}}\bs P\bs r\bs r'\bs P\bs D_{a_{(l_{2})}}\bs P\bs r \\
        &\quad + \frac{1}{n}\bs r'\bs P\bs D_{a_{(l_{1})}}\bs P\bs r\bs r'\bs D_{\varepsilon}\bs V\bs D_{\bar{x}_{(l_{2})}}\bs D_{r}\bs P\bs r + \frac{1}{n}\bs r'\bs P\bs D_{a_{(l_{2})}}\bs P\bs r\bs r'\bs D_{\varepsilon}\bs V\bs D_{\bar{x}_{(l_{1})}}\bs D_{r}\bs P\bs r\\
        &\overset{(\E_{r,U})}{=}
         \underbrace{\frac{1}{n}\bar{\bs z}_{(l_{1})}'[\bs D_{P}\bs D_{V} + \bs D_{P}\bs V\bs D_{P}+\bs P\odot\bs P\odot\bs V-2\bs D_{P}^2\bs D_{V}]\bar{\bs z}_{(l_{2})}}_{\text{fixed-$k$ approximation (I)}}\\
        &\quad +\underbrace{\frac{1}{n}\tr(\bs P\bs D_{a_{(l_{1})}}\bs P\bs D_{a_{(l_{2})}}) + \frac{1}{n}\tr(\bs D_{\Sigma^U(l_{1},l_{2})}(\bs D_{P}\bs D_{V})}_{\text{fixed-$k$ approximation (II)}}\\
        &\quad
        +\frac{1}{n}\bar{\bs z}_{(l_{1})}'\left\{2\bs D_{P}\bs D_{V} + 7\bs D_{P}\bs V\bs D_{P} -10 \bs D_{P}\bs D_{V}\bs D_{P}+3(\bs V\odot\bs P\odot \bs P)\right.\\
        &\quad -2(\bs V\bs D_{\varepsilon}\odot \bs V\bs D_{\varepsilon})(\bs P\odot\bs P)\odot \bs I-4\bs D_{P}^2\bs V\bs D_{P}-4\bs D_{P}\bs V\bs D_{P}^2+16\bs D_{P}^3\bs D_{V}\\
        &\quad \left.-4(\bs V\odot \bs P\odot \bs P)\bs D_{P}-4\bs D_{P}(\bs V\odot \bs P\odot \bs P)\right\}\bar{\bs z}_{(l_{2})} \\
        &\quad +\frac{2}{n}\tr(\bs D_{\Sigma^U(l_{1},l_{2})}(\bs D_P\bs D_V-(\bs V\bs D_\varepsilon\odot\bs V\bs D_\varepsilon)(\bs P\odot\bs P)))\\
        &\quad-\frac{2}{n}\tr((\bs I\odot \bs P\bs D_{a_{(l_{1})}}\bs P)\bs P\bs D_{a_{(l_{2})}}\bs P) + \frac{1}{n}\tr(\bs P\bs D_{a_{(l_{1})}})\tr(\bs P\bs D_{a_{(l_{2})}})\\
        &\quad+ \frac{1}{n}\tr(\bs D_{a_{(l_{1})}}\bs P\bs D_{a_{(l_{2})}}\bs P).\\
        (III)&\overset{(d)}{=}  -\frac{1}{n}\bar{\bs x}_{(l_{1})}'\bs V\bs D_{\varepsilon}\bs r\bs r'\bs D_{\varepsilon}\bs V\bs D_{\bar{x}_{(l_{2})}}\bs D_{r}\bs P\bs r -\frac{1}{n}\bs r'\bs D_{a_{(l_{1})}}\bs P\bs r\bs r'\bs D_{\varepsilon}\bs V\bs D_{\bar{x}_{(l_{2})}}\bs D_{r}\bs P\bs r \\
        &\quad -\frac{1}{n}\bs r'\bs P \bs D_{r}\bs D_{\bar{x}_{(l_{1})}}\bs V\bs D_{\varepsilon}\bs r\bs r'\bs P \bs D_{a_{(l_{2})}}\bs r-\frac{1}{n}\bs r'\bs D_{a_{(l_{1})}}\bs P\bs r\bs r'\bs P\bs D_{a_{(l_{2})}}\bs P\bs r\\
         &\overset{(\E_{r,U})}{=}
       \underbrace{-\frac{1}{n}\bar{\bs z}_{(l_{2})}'\bs V\bs D_{P}\bar{\bs z}_{(l_{1})}-\frac{1}{n}\tr(\bs D_{a_{(l_{2})}}\bs P\bs D_{a_{(l_{1})}}\bs P)-\frac{1}{n}\tr(\bs D_{\Sigma^U(l_{1},l_{2})}\bs D_{P}\bs D_{V} )}_{\text{fixed-$k$ approximation}}\\
        &\quad -\frac{2}{n}\bar{\bs z}_{(l_{2})}'\bs D_{P}(\bs I-\bs D_{P})\bs V\bar{\bs z}_{(l_{1})} -\frac{2}{n}\tr(\bs D_{\Sigma^U(l_{1},l_{2})}\bs D_{P}(\bs I-\bs D_{P})\bs D_{V} )\\
        &\quad + \frac{2}{n}\tr(\bs D_{P}\bs P\bs D_{a_{(l_{1})}}\bs P\bs D_{a_{(l_{2})}}) -\frac{1}{n}\tr(\bs P\bs D_{a_{(l_{2})}})\tr(\bs P\bs D_{a_{(l_{1})}}) -\frac{1}{n}\tr(\bs D_{a_{(l_{2})}}\bs P\bs D_{a_{(l_{1})}}\bs P).\\
        (IV)
        &\overset{(\E_{r,U})}{=}\underbrace{-\frac{1}{n}\bar{\bs z}_{(l_{1})}'\bs V\bs D_{P}\bar{\bs z}_{(l_{2})}-\frac{1}{n}\tr(\bs D_{a_{(l_{1})}}\bs P\bs D_{a_{(l_{2})}}\bs P)-\frac{1}{n}\tr(\bs D_{\Sigma^U(l_{2},l_{1})}\bs D_{P}\bs D_{V} )}_{\text{fixed-$k$ approximation}}\\
        &\quad -\frac{2}{n}\bar{\bs z}_{(l_{1})}'\bs D_{P}(\bs I-\bs D_{P})\bs V\bar{\bs z}_{(l_{1})} -\frac{2}{n}\tr(\bs D_{\Sigma^U(l_{2},l_{1})}\bs D_{P}(\bs I-\bs D_{P})\bs D_{V} )\\
        &\quad + \frac{2}{n}\tr(\bs D_{P}\bs P\bs D_{a_{(l_{2})}}\bs P\bs D_{a_{(l_{1})}}) -\frac{1}{n}\tr(\bs P\bs D_{a_{(l_{2})}})\tr(\bs P\bs D_{a_{(l_{1})}})-\frac{1}{n}\tr(\bs D_{a_{(l_{1})}}\bs P\bs D_{a_{(l_{2})}}\bs P).
    \end{align*}
Note that in $(II)$, $\tr((\bs I\odot \bs P\bs D_{a_{(l_{1})}}\bs P)\bs P\bs D_{a_{(l_{2})}}\bs P)=\bs a_{(l_{1})}'(\bs P\odot\bs P)^2\bs a_{(l_{2})}$ and that in $(III)$, $\tr(\bs D_{P}\bs P\bs D_{a_{(l_{1})}}\bs P\bs D_{a_{(l_{2})}})=\bs a_{(l_{1})}'(\bs P\odot\bs P)\bs D_{P}\bs a_{(l_{2})}$.
    
The products $n^{-1}\tr(\bs P\bs D_{a_{(l_{2})}})\tr(\bs P\bs D_{a_{(l_{1})}})$ cancel when adding $(I)$--$(IV)$. Then,
\begin{equation}\label{eq:conditionalvariance}
\Omega_{l_{1}l_{2}}(\bs\beta_{0})=\E_{r}\big[n\cdot S_{l_{1},r}(\bs\beta_{0})S_{l_{2},r}(\bs\beta_{0})\big|\mathcal{J}\big] =\Omega_{l_{1}l_{2}}^{L}(\bs\beta_{0}) + \Omega_{l_{1}l_{2}}^{H}(\bs\beta_{0}),
\end{equation}
where $\Omega_{l_{1}l_{2}}^{L}(\bs\beta_{0})$ includes all terms labeled ``fixed-$k$ approximation", and $\Omega_{l_{1}l_{2}}^{H}(\bs\beta_{0})$ includes the remaining terms. Some further algebraic manipulations give the result in Lemma \ref{thm:expandvar}.
\end{proof}

\subsection{Proof of Lemma \texorpdfstring{\ref{thm:var unbias}}{3}}\label{sapp:proof var unbias}
\begin{proof} We first show unbiasedness and then proceed with consistency of the estimator.
\subsubsection{Unbiasedness}
The estimator for the variance of the score given in \eqref{eq:omegahat}, evaluated at the true parameter vector $\bs\beta_{0}$, consists of the following components.
\begin{align}\label{eq:variance estimator appendix}
	    \hat{\Omega}_{l_{1}l_{2}}^{L}(\bs\beta_{0})& = \frac{1}{n}\bs x_{(l_{1})}'(\bs I-\bs D_{P\iota})\bs V(\bs I-\bs D_{P\iota})\bs x_{(l_{2})},\nonumber\\
        \hat{\Omega}_{l_{1}l_{2}}^{H}(\bs\beta_{0})&=\frac{1}{n}\bs x_{(l_{1})}'[7\bs D_{P}\dot{\bs V}\bs D_{P}-4\bs D_{P}^2\dot{\bs V}\bs D_{P}-4\bs D_{P}\dot{\bs V}\bs D_{P}^2+3\dot{\bs V}\odot\bs P\odot\bs P\nonumber\nonumber\\
	    &\qquad  -4\bs D_{P}(\dot{\bs V}\odot\bs P\odot\bs P)-4(\dot{\bs V}\odot\bs P\odot\bs P)\bs D_{P}-2\bs D_{P}\dot{\bs V}-2\dot{\bs V}\bs D_{P}+2\bs D_{P}^2\dot{\bs V}\nonumber\\
	    &\qquad+2\dot{\bs V}\bs D_{P}^2]\bs x_{(l_{2})}-\frac{2}{n}\bs x_{(l_{1})}'\big[\big((\bs V\bs D_{\varepsilon}\odot\bs V\bs D_{\varepsilon})(\bs P\odot\bs P)\big)\odot\bs I\big]\bs x_{(l_{2})}.
\end{align}
For $\hat{\Omega}_{l_{1}l_{2}}^{L}(\bs\beta_{0})$, we use Assumption \ref{ass:decomposev} and then Assumption \ref{ass:model} to obtain
\begin{align}
	\bs x_{(l_{2})}'\bs V\bs x_{(l_{1})} 
	&\overset{(d)}{=}\bar{\bs x}_{(l_{2})}'\bs V\bar{\bs x}_{(l_{1})} +\bs r'\bs D_{a_{(l_{2})}}\bs P\bs D_{a_{(l_{1})}}\bs r + \bs r'\bs D_{a_{(l_{2})}} \bs D_{\varepsilon}\bs V\bar{\bs x}_{(l_{1})}+ \bar{\bs x}_{(l_{2})}'\bs V \bs D_{\varepsilon}\bs D_{a_{(l_{1})}}\bs r\nonumber\\
	&\overset{(\E_{r,U})}{=}
 \bar{\bs z}_{(l_{2})}'\bs V\bar{\bs z}_{(l_{1})} +\tr(\bs D_{\Sigma^U(l_{2},l_{1})}\bs V)+ \tr(\bs D_{a_{(l_{2})}}\bs P\bs D_{a_{(l_{1})}}).\label{eq:unbiasedvariancepartsV}
\end{align}
Similarly, we obtain
\begin{align}
    \bs x_{(l_{2})}'\bs D_{P\iota}\bs V\bs x_{(l_{1})}&\overset{(\E_{r,U})}{=}\bar{\bs z}_{(l_{2})}'\bs D_{P}\bs V\bar{\bs z}_{(l_{1})} + \tr(\bs D_{\Sigma^U(l_{2},l_{1})}\bs D_{P}\bs D_{V})+\tr(\bs D_{a_{(l_{1})}}\bs P\bs D_{a_{(l_{2})}}\bs P),\nonumber\\
    \bs x_{(l_{2})}'\bs D_{P\iota}\bs V\bs D_{P\iota}\bs x_{(l_{1})}&\overset{(\E_{r,U})}{=}\bar{\bs z}_{(l_{2})}[\bs D_{P}\bs D_{V}+\bs D_{P}\bs V\bs D_{P}-2\bs D_{P}^2\bs D_{V}+(\bs P\odot\bs P\odot \bs V)]\bar{\bs z}_{(l_{1})} \nonumber\\
    &\qquad +\tr(\bs D_{\Sigma^U(l_{2},l_{1})}\bs D_{V}\bs D_{P}) + \tr(\bs P\bs D_{a_{(l_{1})}}\bs P\bs D_{a_{(l_{2})}}).\label{eq:fixedkvarcons}
\end{align}
Aggregating these results, we see that $\E[\hat{\Omega}_{l_{1}l_{2}}^{L}(\bs\beta_{0})|\mathcal{J}] = \Omega_{l_{1}l_{2}}^{L}(\bs\beta_{0})$.

For $\hat{\Omega}_{l_{1}l_{2}}^{H}(\bs\beta_{0})$, we use the following results
\begin{align}
    &\bs x_{(l_{2})}'\bs D_{P}^{k}\dot{\bs V}\bs D_{P}^{l}\bs x_{(l_{1})} \overset{(\E_{r,U})}{=}\bar{\bs z}_{(l_{2})}'\bs D_{P}^{k}\dot{\bs V}\bs D_{P}^{l}\bar{\bs z}_{(l_{1})},\quad l,k=0,1,2,\nonumber\\
    &\bs x'_{(l_{2})}(\dot{\bs V}\odot \bs P\odot\bs P)\bs D_{P}^{k}\bs x_{(l_{1})} \overset{(\E_{r,U})}{=}\bar{\bs z}_{(l_{2})}'(\dot{\bs V}\odot \bs P\odot\bs P)\bs D_{P}^{k}\bar{\bs z}_{(l_{1})},\quad k=0,1,\nonumber\\
    &\bs x_{(l_{2})}'\bs D_{P}\bs D_{V}\bs x_{(l_{1})}\overset{(\E_{r,U})}{=} \bar{\bs z}_{(l_{2})}'\bs D_{P}\bs D_{V}\bar{\bs z}_{(l_{1})} +\tr(\bs D_{\Sigma^U(l_{2},l_{1})}\bs D_{P}\bs D_{V})+ \tr(\bs D_{P}^2\bs D_{a_{(l_{1})}}\bs D_{a_{(l_{2})}}),\nonumber\\
    &\bs x'_{(l_{2})}(\bs V\odot\bs P)\bs D_{P}\bs x_{(l_{1})}\overset{(\E_{r,U})}{=}\bar{\bs z}_{(l_{2})}'\bs D_{V}\bs D_{P}^{2}\bar{\bs z}_{(l_{1})} + \tr(\bs D_{\Sigma^U(l_{2},l_{1})}\bs D_{V}\bs D_{P}^{2}) +\bs a_{(l_{2})}'(\bs P\odot\bs P)\bs D_{P}\bs a_{(l_{1})}\nonumber,\\
	&\bs x_{(l_{2})}'\big[\big((\bs V\bs D_{\varepsilon}\odot\bs V\bs D_{\varepsilon})(\bs P\odot\bs P)\big)\odot\bs I\big]\bs x_{(l_{1})}\nonumber
    \overset{(\E_{r,U})}{=}\nonumber\\
    &\quad \bar{\bs z}_{(l_{2})}'[((\bs V\bs D_{\varepsilon}\odot\bs V\bs D_{\varepsilon})(\bs P\odot\bs P))\odot\bs I]\bar{\bs z}_{(l_{1})}+ \tr(\bs D_{\Sigma^U(l_{2},l_{1})}(\bs V\bs D_{\varepsilon}\odot\bs V\bs D_{\varepsilon})(\bs P\odot\bs P))\nonumber\\
    &\quad + \bs a_{(l_{2})}'(\bs P\odot\bs P)^2\bs a_{(l_{1})}.\label{eq:unbiasedvarianceparts}
\end{align}
Aggregating these results and using symmetry shows that $\hat{\bs\Omega}(\bs\beta_{0})$ is a conditionally unbiased estimator for $\bs\Omega(\bs\beta_{0})$.

Similarly, we have $\E[\hat{\sigma}^2_n(\bs\beta_{0})]=\E[\frac{2}{k}(k-\bs\iota'\bs D_P^2\bs\iota)]=\frac{2}{k}(\sum_{i=1}^nP_{ii}-\sum_{i=1}^nP_{ii}^2)=\frac{2}{k}(\sum_{i\neq j }P_{ij}^2)=\sigma_n^2(\bs\beta_{0})$ and
\begin{align*}
       \hat{\bs\Sigma}_{l+1,1}(\bs\beta_{0})
        &\overset{(d)}{=}\frac{2}{\sqrt{n\cdot k}}[\bar{\bs x}_{(l)}'(\bs D_{V}-(\bs V\odot\bs D_r\bs P\bs D_r))\bs D_r^2\bs D_{P}\bs D_\varepsilon\bs r\\
        &\quad+\bs r'\bs D_\varepsilon\bs D_{a_{(l)}}(\bs D_{V}-(\bs V\odot\bs D_r\bs P\bs D_r))\bs D_r^2\bs D_{P}\bs D_\varepsilon\bs r]\\
        &\overset{(\E_{r,U})}{=}
        \frac{2}{\sqrt{n\cdot k}}\tr(\bs\Psi^{(l)}\odot\bs P).
\end{align*}
We conclude that $\hat{\bs\Sigma}(\bs\beta_{0})$ is a conditionally unbiased estimator for $\bs\Sigma(\bs\beta_{0})$.

\subsubsection{Consistency}
We first show consistency of the variance estimator of the AR statistic. Under $H_0:\bs\beta=\bs\beta_0$, $\sigma^2_n(\bs\beta_{0})$ and $\hat{\sigma}^2_n(\bs\beta_{0})$ are identical, hence under $H_0$ the estimator is consistent.

Next, we consider the variance estimator of the score statistic. Define $\bs x_{(l_{1}),r}=\bar{\bs z}_{(l_{1})}+\bs u_{(l_{1})}+\bs D_{r}\bs D_{\varepsilon}\bs a_{(l_{1})}$. Under many instrument sequences the variance of the score is bounded away from zero as we establish in Supplementary Appendix \ref{subsubsec:varianceboundedawayfromzero}. Then, to show consistency of the variance estimator, we need to show for some matrix $\bs A_{r}$ that possibly depends on the vector of Rademacher random variables $\bs r$, that
\begin{equation}\label{eq:gen}
    n^{-2}\E[(\bs x_{(l_{1}),r}'\bs A_{r}\bs x_{(l_{2}),r}-\E[\bs x_{(l_{1}),r}'\bs A_{r}\bs x_{(l_{2}),r}|\mathcal{J}])^2|\mathcal{J}]\rightarrow_p 0.
\end{equation}
For $\bs A_{r}$ we consider the general cases (a) $\bs A_{r} = \bs D_{r}\bs A\bs D_{r}$ and (b) $\bs A_{r}=\bs A$, and the specific cases (c) $\bs A_{r} = \bs D_{r}\bs D_{Pr}\bs V$, and (d) $\bs A_{r} = \bs D_{r}\bs D_{Pr}\bs V\bs D_{Pr}\bs D_{r}$.
Cases (a) and (b) cover the consistency of the terms listed in \eqref{eq:unbiasedvariancepartsV} and \eqref{eq:unbiasedvarianceparts} that are all of the form $\bs x_{(l_{1})}'\bs A_{r}\bs x_{(l_{2})}$. For all these terms $\lambda_{\max}(\bs A\odot\bs A)\leq C$ $a.s.n.$ and $\lambda_{\max}(\bs A\bs A')\leq C$ $a.s.n.$, which we will use repeatedly below. We frequently invoke the bound that for a random vector $\bs w$ with independent elements that have bounded fourth moment, we have $\E[(\bs w'\bs A\bs w-\E[\bs w'\bs A\bs w])^2|\bs A]\leq C\tr(\bs A\bs A')$, see for instance \citet{whittle1960bounds}. Cases (c) and (d) will cover the consistency of the terms in \eqref{eq:fixedkvarcons}.

For (a)--(d), we decompose \eqref{eq:gen} into three parts that will be treated separately,
\begin{align*}
   & n^{-2}\E[(\bs x_{(l_{1}),r}'\bs A_{r}\bs x_{(l_{2}),r}-\E[\bs x_{(l_{1}),r}'\bs A_{r}\bs x_{(l_{2}),r}|\mathcal{J}])^2|\mathcal{J}]\\
   &\quad\leq 4\underbrace{n^{-2}\E[(\bar{\bs z}_{(l_{1})}'\bs A_{r}\bar{\bs z}_{(l_{2})}-\E[\bar{\bs z}_{(l_{1})}'\bs A_{r}\bar{\bs z}_{(l_{2})}|\mathcal{J}])^2|\mathcal{J}]}_{(I)}\\
    &\qquad+4\underbrace{n^{-2}\E[(\bs u_{(l_{1})}'\bs A_{r}\bs u_{(l_{2}),r}-\E[\bs u_{(l_{1})}'\bs A_{r}\bs u_{(l_{2})}|\mathcal{J}])^2|\mathcal{J}]}_{(II)}\\
    &\qquad+4\underbrace{n^{-2}\E[(\bs a_{(l_{1})}'\bs D_{\varepsilon}\bs D_{r}\bs A_{r}\bs D_{r}\bs D_{\varepsilon}\bs a_{(l_{2})}-\E[\bs a_{(l_{1})}'\bs D_{\varepsilon}\bs D_{r}\bs A_{r}\bs D_{r}\bs D_{\varepsilon}\bs a_{(l_{2})}|\mathcal{J}])^2|\mathcal{J}]}_{(III)}.
\end{align*}
We start with $(a.I)$--$(a.III)$. 
\begin{align*}
   (a.I) &=n^{-2}\E[(\bar{\bs z}_{(l_{1})}'\bs A_{r}\bar{\bs z}_{(l_{2})}-\E[\bar{\bs z}_{(l_{1})}'\bs A_{r}\bar{\bs z}_{(l_{2})}|\mathcal{J}])^2|\mathcal{J}]\\
     &=n^{-2}\E[(\bs r'\bs D_{\bar{z}_{(l_{1})}}\dot{\bs A}\bs D_{\bar{z}_{(l_{2})}}r)^2|\mathcal{J}]\\
    &=n^{-2}\tr(\bs D_{\bar{z}_{(l_{1})}}\dot{\bs A}\bs D_{\bar{z}_{(l_{2})}}\bs D_{\bar{z}_{(l_{1})}}\dot{\bs A}\bs D_{\bar{z}_{(l_{2})}}) + n^{-2}\tr(\bs D_{\bar{z}_{(l_{1})}}\dot{\bs A}\bs D_{\bar{z}_{(l_{2})}}\bs D_{\bar{z}_{(l_{2})}}\dot{\bs A}\bs D_{\bar{z}_{(l_{1})}})\\
    &\leq 2\lambda_{\max}(\dot{\bs A}\odot\dot{\bs A})\left(\frac{1}{n^2}\sum_{i=1}^{n}\bar{z}_{(l_{1}),i}^4\frac{1}{n^2}\sum_{i=1}^{n}\bar{z}_{(l_{2}),i}^4\right)^{1/2}\rightarrow_{a.s.}0,
\end{align*}
by Assumption \ref{ass:eigbound}. Similarly, for $(a.II)$
\begin{align*}
    (a.II)&=n^{-2}\E[(\bs u_{(l_{1})}'\bs A_{r}\bs u_{(l_{2})}-\E[\bs u_{(l_{1})}'\bs A_{r}\bs u_{(l_{2})}|\mathcal{J}])^2|\mathcal{J}]\\
    &\quad \leq 2n^{-1}\lambda_{\max}(\dot{\bs A}\odot\dot{\bs A})\frac{1}{n}\sum_{i=1}^{n}\E[u_{(l_{1}),ki}^4]^{1/2}\E[u_{(l_{2}),i}^4]^{1/2}\rightarrow_{a.s.}0,
\end{align*}
since Assumptions \ref{ass:model} and \ref{ass:decomposev} imply that $u_{(l_{1}),}$ has bounded fourth moment.
Finally, $(a.III)=\E[(\bs a_{(l_{1})}'\bs D_{\varepsilon}\bs D_{r}\bs A_{r}\bs D_{r}\bs D_{\varepsilon}\bs a_{(l_{2})}-\E[\bs a_{(l_{1})}'\bs D_{\varepsilon}\bs D_{r}\bs A_{r}\bs D_{r}\bs D_{\varepsilon}\bs a_{(l_{2})}|\mathcal{J}])^2|\mathcal{J}]=0$.

For $(b.I)$ we get $\E[(\bar{\bs z}_{(l_{1})}'\bs A_{r}\bar{\bs z}_{(l_{2})}-\E[\bar{\bs z}_{(l_{1})}'\bs A_{r}\bar{\bs z}_{(l_{2})}|\mathcal{J}])^2|\mathcal{J}]=0$, because conditional on $\mathcal{J}$ there is no randomness. For $(b.II)$ we have
\begin{align*}
(b.II)&=n^{-2}\E[(\bs u_{(l_{1})}'\bs A\bs u_{(l_{2})}-\E[\bs u_{(l_{1})}'\bs A\bs u_{(l_{2})}|\mathcal{J}])^2|\mathcal{J}]\\
&=n^{-2}\E[\sum_{i_{1},i_{2},i_{3},i_{4}}u_{(l_{1}),i_{1}}u_{(l_{1}),i_{2}}u_{(l_{2}),i_{3}}u_{(l_{2}),i_{4}}A_{i_{1}i_{3}}A_{i_{2}i_{4}}]-n^{-2}\tr(\bs D_{\Sigma^U(l_{1},l_{2})}\bs D_{A})^2\\
&\leq n^{-2}\E[\bs u_{(l_{1})}'\bs D_{u_{(l_{1})}}(\bs A\odot\bs A)\bs D_{u_{(l_{2})}}\bs u_{(l_{2})}+n^{-2}\bs u_{(l_{2})}'\bs D_{u_{(l_{1})}}\bs A\bs A'\bs D_{u_{(l_{1})}}\bs u_{(l_{2})}]\\
&\leq n^{-1}(\lambda_{\max}(\dot{\bs A}\odot\dot{\bs A})+\lambda_{\max}(\bs A\bs A'))\frac{1}{n}\sum_{i=1}^{n}\E[u_{(l_{1}),i}^4]^{1/2}\E[u_{(l_{2}),i}^4]^{1/2}\rightarrow_{a.s.}0.
\end{align*}
Finally, $(b.III)$ satisfies
\begin{align*}
    (b.III)&=n^{-2}\E[(\bs a_{(l_{1})}'\bs D_{\varepsilon}\bs D_{r}\bs A\bs D_{r}\bs D_{\varepsilon}\bs a_{(l_{2})}-\E[\bs a_{(l_{1})}'\bs D_{\varepsilon}\bs D_{r}\bs A\bs D_{r}\bs D_{\varepsilon}\bs a_{(l_{2})}|\mathcal{J}])^2|\mathcal{J}]\\
    &=n^{-2}\E[(\bs r'\bs D_{a_{(l_{1})}}\bs D_{\varepsilon}\bs A\bs D_{\varepsilon}\bs D_{a_{(l_{2})}}\bs r)^2|\mathcal{J}]-\tr(\bs D_{a_{(l_{1})}}\bs D_{\varepsilon}\bs A\bs D_{\varepsilon}\bs D_{a_{(l_{2})}})^2\\
    & \leq Cn^{-2}\tr(\bs D_{a_{(l_{1})}}\bs D_{\varepsilon}\bs A\bs D_{\varepsilon}\bs D_{a_{(l_{2})}}\bs D_{a_{(l_{2})}}\bs D_{\varepsilon}\bs A'\bs D_{\varepsilon}\bs D_{a_{(l_{1})}})\\
    &\leq Cn^{-2}\tr(\bs D_{\varepsilon}\bs A\bs D_{\varepsilon}^2\bs A'\bs D_{\varepsilon}).
\end{align*}
Using the expressions for $\bs A$ as in \eqref{eq:unbiasedvariancepartsV} and \eqref{eq:unbiasedvarianceparts}, we see that $(b.III)\rightarrow_{a.s.}0$.

We continue with $(c.I)$--$(c.III)$.
\begin{align*}
    (c.I)&=n^{-2}\E[(\bar{\bs z}_{(l_{1})}'\bs D_{r}\bs D_{Pr}\bs V\bar{\bs z}_{(l_{2})}-\E[\bar{\bs z}_{(l_{1})}'\bs D_{r}\bs D_{Pr}\bs V\bar{\bs z}_{(l_{2})}|\mathcal{J}])^2|\mathcal{J}]\\
    &=n^{-2}\E[(\bs r'\bs P\bs D_{r}\bs D_{\bar{z}_{(l_{1})}}\bs V\bar{\bs z}_{(l_{2})}-\E[\bar{\bs z}_{(l_{1})}'\bs D_{P}\bs V\bar{\bs z}_{(l_{2})}|\mathcal{J}])^2|\mathcal{J}]\\
    &=n^{-2}\tr(\bs D_{P}\bs D_{\bar{z}_{(l_{1})}}\bs D_{V\bar{z}_{(l_{2})}}^2\bs D_{\bar{z}_{(l_{1})}})- 2n^{-2}\tr(\bs D_{\bar{z}_{(l_{1})}}\bs V\bar{\bs z}_{(l_{2})}\bar{\bs z}_{(l_{2})}'\bs V\bs D_{\bar{z}_{(l_{1})}}\bs D_{P})\\
    &\quad +n^{-2}\bs\iota'(\bs P\odot\bs P\odot(\bs D_{\bar{z}_{(l_{1})}}\bs V\bar{\bs z}_{(l_{2})}\bar{\bs z}_{(l_{2})}'\bs V\bs D_{\bar{z}_{(l_{1})}}))\bs\iota\\
    &=n^{-2}\bar{\bs z}_{(l_{2})}'\bs V\bs D_{\bar{z}_{(l_{1})}}\bs D_{P}\bs D_{\bar{z}_{(l_{1})}}\bs V\bar{\bs z}_{(l_{2})}-2n^{-2}\bar{\bs z}_{(l_{2})}'\bs V\bs D_{\bar{z}_{(l_{1})}}\bs D_{P}\bs D_{\bar{z}_{(l_{1})}}\bs V\bar{\bs z}_{(l_{2})}\\
    &\quad+ n^{-2}\bar{\bs z}_{(l_{2})}'\bs V\bs D_{\bar{z}_{(l_{1})}}(\bs P\odot\bs P)\bs D_{\bar{z}_{(l_{1})}}\bs V\bar{\bs z}_{(l_{2})}\\
    &\leq \left(\frac{1}{n^2}\sum_{i=1}^{n}\bar{z}_{l_{1},i}^4\frac{1}{n^2}\sum_{i=1}^{n}(\bar{\bs z}_{(l_{2})}'\bs V\bs e_{i})^4\right)^{1/2}\rightarrow_{a.s.}0,
\end{align*}
with the convergence implied by Assumption \ref{ass:eigbound}.

$(c.II)$ follows by analogous arguments. For $(c.III)$, we have 
\begin{align*}
    (c.III)&=n^{-2}\E[(\bs a_{(l_{1})}'\bs D_{\varepsilon}\bs D_{Pr}\bs V\bs D_{r}\bs D_{\varepsilon}\bs a_{(l_{2})}-\E[\bs a_{(l_{1})}'\bs D_{\varepsilon}\bs D_{Pr}\bs V\bs D_{r}\bs D_{\varepsilon}\bs a_{(l_{2})}|\mathcal{J}])^2|\mathcal{J}]\\
    &=n^{-2}\E[(\bs r'\bs P\bs D_{a_{(l_{1})}}\bs P\bs D_{a_{(l_{2})}}\bs r-\tr(\bs P\bs D_{a_{(l_{1})}}\bs P\bs D_{a_{(l_{2})}}))^2|\mathcal{J}]\\
    &\leq n^{-2}\tr(\bs P\bs D_{a_{(l_{1})}}\bs P\bs D_{a_{(l_{2})}}^2\bs P\bs D_{a_{(l_{1})}}\bs P)\rightarrow_{a.s.}0.
\end{align*}

Proceeding with $(d.I)$--$(d.III)$, we have
\begin{align*}
    (d.I)&=n^{-2}\E[(\bar{\bs z}_{(l_{1})}'\bs D_{r}\bs D_{Pr}\bs V\bs D_{Pr}\bs D_{r}\bar{\bs z}_{(l_{2})}-\E[\bar{\bs z}_{(l_{1})}'\bs D_{r}\bs D_{Pr}\bs V\bs D_{Pr}\bs D_{r}\bar{\bs z}_{(l_{2})}|\mathcal{J}])^2|\mathcal{J}].
\end{align*}
Notice that
\begin{align*}
     &n^{-1} \bar{\bs z}_{(l_{1})}'\bs D_{r}\bs D_{Pr}\bs V\bs D_{Pr}\bs D_{r}\bar{\bs z}_{(l_{2})}\\
     & = n^{-1}\bs\iota\bs D_{r}\bs P\bs D_{r}\bs D_{\bar{z}_{(l_{1})}}\bs V\bs D_{\bar{z}_{(l_{2})}}\bs D_{r}\bs P\bs D_r\bs\iota\\
       &=n^{-1}\bar{\bs z}_{(l_{1})}'\bs D_{P}\bs V\bs D_{P}\bar{\bs z}_{(l_{2})}+n^{-1}\bs\iota'\bs D_{P}\bs D_{\bar{z}_{(l_{1})}}\bs V\bs D_{\bar{z}_{(l_{2})}}\bs D_{r}\dot{\bs P}\bs D_{r}\bs\iota\\
      &\quad+ n^{-1}\bs\iota'\bs D_{r}\dot{\bs P}\bs D_{r}\bs D_{\bar{z}_{(l_{1})}}\bs V\bs D_{\bar{z}_{(l_{2})}}\bs D_{P}\bs\iota +n^{-1}\bs\iota'\bs D_{r}\dot{\bs P}\bs D_{r}\bs D_{\bar{z}_{(l_{1})}}\dot{\bs V}\bs D_{\bar{z}_{(l_{2})}}\bs D_{r}\dot{\bs P}\bs D_{r}\bs\iota.
\end{align*}
The second and third term after the final equality sign have expectation equal to zero. The difference of these terms from their expectation converges almost surely to zero by the same arguments as used in showing convergence of parts (a)--(c). The final term has expectation $\bar{\bs z}_{(l_{1})}'(\dot{\bs V}\odot\dot{\bs P}\odot\dot{\bs P})\bar{\bs z}_{(l_{2})}$. Subtracting this expectation, and defining $\bs r_{-ij}$ as the vector $\bs r$ with the $i^{\text{th}}$ and $j^{\text{th}}$ element set to zero, the final term can be written as
\begin{align*}
    \tr(\dot{\bs P}\bs D_{r}\bs D_{\bar{z}_{(l_{1})}}\dot{\bs V}\bs D_{\bar{z}_{(l_{2})}}\bs D_{r}\dot{\bs P}) + n^{-1}\sum_{i_{1}=1}^{n}\sum_{i_{2}\neq i_{1}}r_{i_{1}}r_{i_{2}}\bs r_{-i_{1}i_{2}}'\bs D_{Pe_{i_{1}}}\bs D_{\bar{z}_{(l_{1})}}\dot{\bs V}\bs D_{\bar{z}_{(l_{2})}}\bs D_{Pe_{i_{2}}}\bs r_{-i_{1}i_{2}}.
\end{align*}
Squaring and taking the expectation, we get the bound 
\begin{align*}
    &\frac{2}{n^2}\E[\tr(\dot{\bs P}\bs D_{r}\bs D_{\bar{z}_{(l_{1})}}\dot{\bs V}\bs D_{\bar{z}_{(l_{2})}}\bs D_{r}\dot{\bs P})^2|\mathcal{J}]\\
    &\qquad + \frac{4}{n^2}\sum_{i_{1}=1}^{n}\sum_{i_{2}=1}^{n}\E[(\bs r_{-i_{1}i_{2}}'\bs D_{Pe_{i_{1}}}\bs D_{\bar{z}_{(l_{1})}}\dot{\bs V}\bs D_{\bar{z}_{(l_{2})}}\bs D_{Pe_{i_{2}}}\bs r_{-i_{1}i_{2}})^2|\mathcal{J}]\\
    &\leq \frac{2}{n^2}\E[(\bs r'\bs D_{\bar{z}_{(l_{1})}}(\dot{\bs V}\odot\dot{\bs P}^2)\bs D_{\bar{z}_{(l_{2})}}\bs r)^2]\\
    &\quad+ \frac{4}{n^2}\sum_{i_{1}=1}^{n}\sum_{i_{2}=1}^{n}\tr(\bs D_{Pe_{i_{1}}}\bs D_{\bar{z}_{(l_{1})}}\dot{\bs V}\bs D_{\bar{z}_{(l_{2})}}\bs D_{Pe_{i_{2}}}\bs D_{Pe_{i_{2}}}\bs D_{\bar{z}_{(l_{2})}}\dot{\bs V}\bs D_{\bar{z}_{(l_{1})}}\bs D_{Pe_{i_{1}}})\\
     &\leq C\left(\frac{1}{n^2}\sum_{i=1}^n\bar{z}_{(l_{1}),i}^4\frac{1}{n^2}\sum_{i=1}^n\bar{z}_{(l_{2}),i}^4\right)^{1/2}\rightarrow_{a.s.}0.
\end{align*}
$(d.II)$ follows from analogous arguments. Finally,
\begin{align*}
    (d.III) & = n^{-2}\E[(\bs a_{(l_{1})}'\bs D_{\varepsilon}\bs D_{Pr}\bs V\bs D_{Pr}\bs D_{\varepsilon}\bs a_{(l_{2})}-\tr(\bs P\bs D_{a_{(l_{1})}}\bs P\bs D_{a_{(l_{2})}}))^2|\mathcal{J}]\\
    & = n^{-2}\E[(\bs r'\bs P\bs D_{ a_{(l_{1})}}\bs P\bs D_{a_{(l_{2})}}\bs P\bs r)^2|\mathcal{J}]-\tr(\bs P\bs D_{a_{(l_{1})}}\bs P\bs D_{a_{(l_{2})}}))^2|\mathcal{J}]^2\\
    &\leq n^{-2}\tr(\bs P\bs D_{a_{(l_{1})}}\bs P\bs D_{a_{(l_{2})}}^2\bs P\bs D_{a_{(l_{1})}}\bs P)\rightarrow_{a.s.}0.
\end{align*}

Lastly, we consider the estimator of the covariance between the AR and the score statistic. From \eqref{eq:covariance AR score} and \eqref{eq:omegacov} we can bound the variance of $[\hat{\bs \Sigma}_{n,r}(\bs\beta_{0})]_{l+1,1}$ as 
\begin{align}
       &\E[(\hat{\bs \Sigma}_{n,r}(\bs\beta_{0})]_{l+1,1})^2|\mathcal{J}] \nonumber\\
       &\leq \frac{4}{ nk}\E[(\tr(\bs\Psi^{(l)}\odot\bs P)-(\bar{\bs z}_{(l)}+\bs D_{r}\bs D_{\varepsilon}\bs a_{(l)}+\bs u_{(l)})'(\bs D_V-\bs D_{r}(\bs V\odot\bs P))\bs D_P\bs\varepsilon)^2|\mathcal{J}]\nonumber\\
        &\leq\frac{C}{ nk}(\E[(\tr(\bs\Psi^{(l)}\odot\bs P)-\bs a_{(l)}'\bs D_{r}\bs D_{\varepsilon}(\bs D_V-\bs D_{r}(\bs V\odot\bs P))\bs D_P\bs\varepsilon)^2|\mathcal{J}]\nonumber\\
        &\quad+\E[(\bar{\bs z}_{(l)}'(\bs D_V-\bs D_{r}(\bs V\odot\bs P))\bs D_P\bs\varepsilon)^2|\mathcal{J}]+\E[(\bs u_{(l)}'(\bs D_V-\bs D_{r}(\bs V\odot\bs P))\bs D_P\bs\varepsilon)^2|\mathcal{J}])\nonumber\\
        &=\frac{C}{ nk}(\E[(\bar{\bs z}_{(l)}'(\bs D_V-\bs D_{r}(\bs V\odot\bs P))\bs D_P\bs\varepsilon)^2|\mathcal{J}]+\E[(\bs u_{(l)}'(\bs D_V-\bs D_{r}(\bs V\odot\bs P))\bs D_P\bs\varepsilon)^2|\mathcal{J}]).\label{eq:consistency covariance}
\end{align}
The first term becomes, by using the law of iterated expectations, Assumption \ref{ass:main} and Lemma \ref{thm:Rademacher} in Supplementary Appendix \ref{sapp:rademacher} which gives the expectation over products of Rademacher random variables,
\begin{align*}
    &\frac{C}{nk}\E[(\bar{\bs z}_{(l)}'(\bs D_V-(\bs V\odot\bs P))\bs D_P\bs\varepsilon)^2|\mathcal{J}]\\
    &\leq\frac{C}{nk}(\E[(\bar{\bs z}_{(l)}'\bs D_V\bs D_{r}\bs D_P\bs\varepsilon)^2|\mathcal{J}]+\E[(\bar{\bs z}_{(l)}'\bs D_{r}(\bs V\odot\bs P)\bs D_P\bs\varepsilon)^2|\mathcal{J}])\\
    &=\frac{C}{nk}(\E[\bar{\bs z}_{(l)}'\bs D_V\bs D_P\bs D_{\varepsilon}\bs r\bs r'\bs D_V\bs D_P\bs D_{\varepsilon}\bar{\bs z}_{(l)}|\mathcal{J}]\\
    &\quad+\E[\bs r'\bs D_{\bar{z}_{(l)}}(\bs V\odot\bs P)\bs D_P\bs\varepsilon\bs\varepsilon'\bs D_P(\bs V\odot\bs P)\bs D_{\bar{z}_{(l)}}\bs r|\mathcal{J}])\\
    &=\frac{C}{nk}(\tr(\bar{\bs z}_{(l)}'\bs D_V\bs D_P^3\bar{\bs z}_{(l)})+\bs\varepsilon'\bs D_P(\bs V\odot\bs P)\bs D_{\bar{z}_{(l)}}^2(\bs V\odot\bs P)\bs D_P\bs\varepsilon)\\
    &=\frac{C}{nk}(\tr(\bar{\bs z}_{(l)}'\bs D_V\bs D_P^3\bar{\bs z}_{(l)})+\bs\iota'\bs D_{\varepsilon}\bs D_P(\bs V\odot\bs P)\bs D_{\bar{z}_{(l)}}^2(\bs V\odot\bs P)\bs D_P\bs D_{\varepsilon}\bs\iota)\\
     &\leq\frac{C}{nk}(\tr(\bar{\bs z}_{(l)}'\bs D_V\bs D_P^3\bar{\bs z}_{(l)})+\bar{\bs z}_{(l)}'\bs D_{V}\bs D_{P}\bar{\bs z}_{(l)})
    \rightarrow_{a.s.}0,
\end{align*}
by Assumption \ref{ass:eigbound}. 
The final inequality uses that $\bs e_{i}'(\bs V\odot \bs V)\bs D_{P}\bs D_{\varepsilon}^2\bs\iota = \sum_{j}V_{ij}^2P_{jj}\varepsilon_{j}^2\leq \sum_{j}V_{ij}^2\varepsilon_{j}^2=V_{ii}$ and hence,
\begin{align*}
    \bs\iota'\bs D_{\varepsilon}\bs D_P(\bs V\odot\bs P)\bs D_{\bar{z}_{(l)}}^2(\bs V\odot\bs P)\bs D_P\bs D_{\varepsilon}\bs\iota&\leq \bar{\bs z}_{(l)}'\bs D_{\varepsilon}\bs D_{V}^2\bs D_{\varepsilon}\bar{\bs z}_{(l)}=\bar{\bs z}_{(l)}'\bs D_{V}\bs D_{P}\bar{\bs z}_{(l)}.
\end{align*}
We conclude that, under $H_0\colon\bs\beta=\bs\beta_0$, $\hat{\bs\Sigma}_n$ is consistent for $\bs\Sigma_n$.
\end{proof}

\subsection{Expectations over Rademacher random variables}\label{sapp:rademacher}
For the proof of Lemma \ref{lem:AR score} we require the following result.
\begin{lemma}\label{thm:Rademacher} Consider an $n\times 1$ vector $\bs r$ with independent Rademacher entries. Let $\bs A_1,\ldots,\bs A_4$ denote generic $n\times n$ matrices and $\bs v$ an $n\times 1$ vector. Then,
\begin{enumerate}
    \item\label{it:rAr} $\E[\bs r'\bs A_1\bs r] = \tr(\bs A_1)$.
    \item\label{it:rArrAr} \citet{ullah2004finite}:\newline
    $\E[\bs r'\bs A_1\bs r\bs r'\bs A_2\bs r] = -2\tr(\bs D_{A_1}\bs A_2) + \tr(\bs A_1)\tr(\bs A_2)+\tr(\bs A_1\bs A_2)+\tr(\bs A_1'\bs A_2)$.
     \item $\E[\bs v'\bs r\bs r'\bs A_2\bs D_{r}\bs A_1\bs r] = \bs v'\bs A_2\bs D_{A_1}\bs \iota+ \bs \iota'(\bs A_2\odot \bs A_1')\bs v + \bs \iota'\bs D_{A_2}\bs A_1\bs v - 2\bs \iota'\bs D_{A_2}\bs D_{A_1}\bs v$.
    \item\label{it:rArrArrArrAr}
    $\begin{aligned}[t]
            &\E[\bs r'\bs A_1\bs D_{r}\bs A_2\bs r\bs r'\bs A_3\bs D_{r}\bs A_4\bs r]
             = \tr(\bs D_{A_4A_1}\bs D_{ A_2A_3})\\ 
            &\quad +\bs\iota'\bs D_{A_3}\bs A_4\bs A_1\bs D_{A_2}\bs\iota-2\tr\left(\bs D_{A_3}\bs A_4\bs A_1\bs D_{A_2}\right)+\bs\iota'\bs A_2\odot\bs A_3\odot(\bs A_1'\bs A_4')\bs\iota\\ 
            &\quad +\bs\iota'\bs D_{A_1}\bs A_2\bs A_3\bs D_{A_4}\bs\iota-2\tr\left(\bs D_{A_1}\bs A_2\bs A_3\bs D_{A_4}\right)+\bs\iota'\bs A_4\odot\bs A_1\odot(\bs A_3'\bs A_2')\bs\iota\\ 
            &\quad +\bs\iota'(\bs A_1\odot \bs A_3)(\bs A_2\odot \bs A_4)\bs\iota+\tr(\bs A_1'\bs A_2'\odot \bs A_3'\bs A_4')\\ 
            &\quad -2\tr((\bs A_1\odot\bs A_3)(\bs A_2\odot\bs A_4))+\bs\iota'\bs D_{A_4A_3}\bs A_1\bs D_{A_2}\bs\iota -2\tr(\bs D_{A_4A_3}\bs A_1\bs D_{A_2})\\
            &\quad +\bs\iota'\bs D_{A_2A_1}\bs A_3\bs D_{A_4}\bs\iota -2\tr(\bs D_{A_2A_1}\bs A_3\bs D_{A_4})\\
            &\quad -2\bs\iota'\bs D_{A_1}\bs D_{A_2}\bs A_3\bs D_{A_4}\bs\iota-2\bs\iota'\bs D_{A_3}\bs D_{A_4}\bs A_1\bs D_{A_2}\bs\iota + 16\tr(\bs D_{A_3}\bs D_{A_4}\bs A_1\bs D_{A_2})\\ 
            &\quad -2\bs\iota'\bs A_1\bs D_{A_2}\odot\bs A_4\odot \bs A_3'\bs\iota-2\bs\iota'\bs A_3\bs D_{A_4}\odot\bs A_2\odot \bs A_1'\bs\iota\\
            & \quad+  \bs\iota'\bs D_{A_4}\bs A_3'\bs A_1\bs D_{A_2}\bs\iota -\tr(\bs D_{A_4}\bs A_3'\bs A_1\bs D_{A_2})-\tr(\bs D_{A_2}\bs A_1'\bs A_3\bs D_{A_4})\\
            &\quad +   \bs\iota'((\bs A_3\odot\bs A_1')\bs A_4)\odot \bs A_2\bs\iota -2\bs\iota'((\bs A_1\odot\bs A_3'\odot\bs I)\bs A_2)\odot \bs A_4\bs\iota\\ 
            &\quad+  \bs\iota'(\bs A_1\odot(\bs A_3(\bs A_2'\odot\bs A_4)))\bs\iota  -2\tr((\bs A_1\odot(\bs A_3(\bs A_2'\odot\bs A_4))))\\ 
            &\quad-2\tr((\bs A_3\odot(\bs A_1(\bs A_4'\odot\bs A_2))))\\
            &\quad + \bs\iota'(\bs A_1\odot \bs A_2')\bs A_4'\bs D_{A_3}\bs\iota -2\bs\iota'(\bs A_1\odot\bs A_2'\odot\bs I)\bs A_4'\bs D_{A_3}\bs\iota \\ 
            &\quad + \bs\iota'(\bs A_3\odot \bs A_4')\bs A_2'\bs D_{A_1}\bs\iota -2\bs\iota'(\bs A_3\odot\bs A_4'\odot\bs I)\bs A_2'\bs D_{A_1}\bs\iota \\ 
            &\quad + \bs\iota'\bs D_{A_1}\bs A_2\bs A_4'\bs D_{A_3}\bs\iota. 
        \end{aligned}$
      \item\label{it:rADrBDrCr} 
    $\begin{aligned}[t]
            &\E[\bs r'\bs A_{1}\bs D_{r}\bs A_{2}\bs A_{3}\bs D_{r}\bs A_{4}\bs r] = \tr(\bs D_{A_{4}A_{1}}\bs D_{A_{2}\bs A_{3}})+\bs\iota'\bs D_{A_{1}}\bs A_{2}\bs A_{3}\bs D_{A_{4}}\bs\iota\\
            &\quad - 2\tr(\bs D_{A_{1}}\bs A_{2}\bs A_{3}\bs D_{A_{4}})+\bs\iota'(\bs A_{4}\odot\bs A_{1}\odot(\bs A_{3}'\bs A_{2}'))\bs\iota.
        \end{aligned}$
        \setcounter{expectRademachers}{\value{enumi}}
\end{enumerate}
Suppose now that $\bs A_1$ and $\bs A_2$ are symmetric matrices with all diagonal elements equal to zero. Then,
\begin{enumerate}
    \setcounter{enumi}{\value{expectRademachers}}
    \item\label{it:rAr2rAr2} \citet{bao2010expectation}:\newline
        $\begin{aligned}[t]
    		&\E[(\bs r'\bs A_1\bs r)^2(\bs r'\bs A_2\bs r)^2]=4\tr(\bs A_1^2)\tr(\bs A_2^2)+8\tr^2(\bs A_1\bs A_2)\\
    		&\quad+32\tr(\bs A_1\bs A_2\bs A_1\bs A_2)+16\tr(\bs A_1\bs A_2\bs A_2\bs A_1)-32\bs\iota'(\bs I\odot\bs A_1^2)(\bs I\odot\bs A_2^2)\bs\iota\\
    		&\quad-64\bs\iota'(\bs I\odot\bs A_1\bs A_2)(\bs I\odot\bs A_1\bs A_2)\bs\iota+32\bs\iota'(\bs A_1\odot\bs A_1\odot\bs A_2\odot\bs A_2)\bs\iota.
    	\end{aligned}$
    \item\label{it:rArrArrAr}
    \citet{ullah2004finite}:\newline
    $\begin{aligned}[t]
    	    \E[\bs r'\bs A_1\bs r\bs r'\bs A_2\bs r\bs r'\bs A_3\bs r]
    		&=\tr(\bs A_3[20(\bs A_1\odot \bs A_2)-3\bs I\odot(2\bs A_1\bs A_2+2\bs A_2\bs A_1)\\
    		&\quad+4\bs A_1\bs A_2+4\bs A_2\bs A_1+2\tr(\bs A_1\bs A_2)\bs I]).
    	\end{aligned}$
\end{enumerate}
\end{lemma}

\begin{proof}
\begin{enumerate}\itemsep0em 
    \item $\E[\bs r'\bs A_1\bs r] = \tr(\bs A_1\E[\bs r\bs r']) = \tr(\bs A_1)$.
    \item See \citet{ullah2004finite}, Appendix A5.
    \item Denote $\bs\Delta = \bs r\bs r'-\bs I$. We split the expectation into two parts,
\begin{align*}
    \E[\bs v'\bs r\bs r'\bs A_1\bs D_{r}\bs A_2\bs r]& = \underbrace{\E[\bs v'\bs A_1\bs D_{r}\bs A_2\bs r]}_{(I)} + \underbrace{\E[\bs v'\bs\Delta \bs A_1\bs D_{r}\bs A_2\bs r]}_{(II)}.
\end{align*}
For the first part, using independence of the Rademacher random variables,
\begin{align*}
    (I)&=\E\left[\sum_{i_{1},i_{2},i_{3}} v_{i_{1}}a_{1,i_{1}i_{2}}a_{2,i_{2}i_{3}}r_{i_{2}}r_{i_{3}}\right]
    = \sum_{i_{1},i_{2}} v_{i_{1}}a_{1,i_{1}i_{2}}a_{2,i_{2}i_{2}}
    = \bs v'\bs A_1\bs D_{A_2}\bs\iota.
\end{align*}
For $(II)$ we write $
    (II) = \E\left[\sum_{i_{1},i_{2},i_{3},i_{4}}v_{i_{1}}\delta_{i_{1}i_{2}}a_{1,i_{2}i_{3}}a_{2,i_{3}i_{4}}r_{i_{3}}r_{i_{4}}\right].$
There are two cases where the expectation is nonzero. In case $(II.a)$, $i_{1}=i_{3},i_{2}=i_{4},i_{1}\neq i_{2}$, and
\begin{align*}
    (II.a) & = \sum_{i_{1}\neq i_{2}}v_{i_{1}}a_{1,i_{2}i_{1}}a_{2,i_{1}i_{2}}
    =\bs\iota \bs A_1\odot \bs A_2'\bs v -\bs \iota'\bs D_{A_1}\bs D_{A_2}\bs v.
\end{align*}
In case $(II.b)$, $i_{1}=i_{4},i_{2}=i_{3},i_{1}\neq i_{2}$, such that
\begin{align*}
    (II.b)  = \sum_{i_{1}\neq i_{2}}v_{i_{1}}a_{1,i_{2}i_{2}}a_{2,i_{2}i_{1}}
    & = \bs\iota'\bs D_{A_1}\bs A_2\bs v-\bs\iota'\bs D_{A_1}\bs D_{A_2}\bs v.
\end{align*}
    \item\label{it:rADrArrADrAr proof}
    We decompose the expectation as
    \begin{align}\label{eq:parts}
    &\E[\bs r'\bs A_{1}\bs D_{r} \bs A_{2}\bs r\bs r'\bs A_{3}\bs D_{r}\bs A_{4}\bs r]
     = \underbrace{\E[\tr(\bs A_{1}\bs D_{r}\bs A_{2}\bs A_{3}\bs D_{r}\bs A_{4})]}_{(I)}\nonumber\\
     &\quad+ \underbrace{\E[\tr(\bs A_1\bs D_{r}\bs A_2(\bs r\bs r'-\bs I)\bs A_3\bs D_{r}\bs A_4(\bs r\bs r'-\bs I))]}_{(III)}\nonumber\\
    &\quad +\underbrace{\E[\tr(\bs A_{1}\bs D_{r}\bs A_2(\bs r\bs r'-\bs I)\bs A_3\bs D_{r}\bs A_4)]}_{(II)}+\underbrace{\E[\tr(\bs A_{1}\bs D_{r}\bs A_2\bs A_3\bs D_{r}\bs A_4(\bs r\bs r'-\bs I))]}_{(II')}.
    \end{align}
    Starting with $(I)$, we have that
    \begin{align*}
    (I)
    &=\E\left[\sum_{i_{1},i_{2},i_{3}}\bs e_{i_{1}}'\bs A_1\bs e_{i_{2}}\bs e_{i_{2}}'\bs D_{r}\bs e_{i_{2}}\bs e_{i_{2}}'\bs A_2\bs A_3\bs e_{i_{3}}\bs e_{i_{3}}'\bs D_{r}\bs e_{i_{3}}\bs e_{i_{3}}'\bs A_4\bs e_{i_{1}}\right]\\
    & = \sum_{i_{1},i_{2}}\bs e_{i_{1}}'\bs A_1\bs e_{i_{2}}\bs e_{i_{2}}'\bs A_2\bs A_3\bs e_{i_{2}}\bs e_{i_{2}}'\bs A_4\bs e_{i_{1}}
     = \tr(\bs D_{A_{4} A_{1}}\bs D_{A_{2}A_{3}}).
    \end{align*}
    For $(II)$, define $\delta_{i_{1}i_{2}} = [\bs r\bs r'-\bs I]_{i_{1}i_{2}}$ and note that $\delta_{i_{1}i_{1}}=0$, and $\E[\delta_{i_{1}i_{2}}]=\E[r_{i_{1}}r_{i_{2}}]=0$ if $i_{1}\neq i_{2}$ and $\E[\delta_{i_{1}i_{2}}^2]=\E[r_{i_{1}}^2r_{i_{2}}^2]=1$.
    \begin{align*}
    (II) & = \E\left[\sum_{i_{1}}\bs e_{i_{1}}'\bs A_1\bs D_{r}\bs A_2(\bs r\bs r'-\bs I)\bs A_3\bs D_{r}\bs A_4\bs e_{i_{1}}\right]\\
    & = \E\left[\sum_{i_{1},i_{2},i_{3},i_{4},i_{5}}\bs e_{i_{1}}'\bs A_1\bs e_{i_{2}}\bs e_{i_{2}}'\bs D_{r}\bs e_{i_{2}}\bs e_{i_{2}}'\bs A_2\bs e_{i_{3}}\delta_{i_{3}i_{4}}\bs e_{i_{4}}'\bs A_3\bs e_{i_{5}}\bs e_{i_{5}}'\bs D_{r}\bs e_{i_{5}}\bs e_{i_{5}}'\bs A_4\bs e_{i_{1}}\right].
    \end{align*}
    There are two cases when the expectation is nonzero: (a) $i_{2}=i_{3},i_{4}=i_{5},i_{2}\neq i_{5}$ and (b) $i_{2}=i_{4},i_{3}=i_{5},i_{2}\neq i_{5}$. Starting with case (a),
    \begin{align*}
    (II.a)&=\sum_{i_{1}}\sum_{i_{2}\neq i_{5}}\bs e_{i_{1}}'\bs A_1\bs e_{i_{2}}\bs e_{i_{2}}'\bs A_2\bs e_{i_{2}}\bs e_{i_{5}}'\bs A_3\bs e_{i_{5}}\bs e_{i_{5}}'\bs A_4\bs e_{i_{1}}\\
    &=\bs\iota'\bs D_{A_3}\bs A_4\bs A_1\bs D_{A_2}\bs\iota \underbrace{- \tr\left(\bs D_{A_3}\bs A_4\bs A_1\bs D_{A_2}\right)}_{(II.a.2)}.
    \end{align*}
    For case (b), we have
    \begin{align*}
    (II.b)&=\sum_{i_{1}}\sum_{i_{2}\neq i_{5}}\bs e_{i_{2}}'\bs A_1\bs e_{i_{2}}\bs e_{i_{2}}'\bs A_2\bs e_{i_{5}}\bs e_{i_{5}}'\bs A_3\bs e_{i_{5}}\bs e_{i_{5}}'\bs A_4\bs e_{i_{1}}
      = \bs\iota'\bs A_2\odot\bs A_3\odot(\bs A_1'\bs A_4')\bs\iota + (II.a.2).
    \end{align*}
    By rotation invariance, the expressions for $(II')$ can be obtained by changing $\bs A_{2}\rightarrow \bs A_{4}$, $\bs A_{3}\rightarrow \bs A_{1}$, $\bs A_{4}\rightarrow \bs A_{2}$, $\bs A_{1}\rightarrow \bs A_{3}$.
    
    The most difficult term to deal with in \eqref{eq:parts} is 
    \begin{equation}
    (III)= \E\left[\sum_{i_{1},\ldots,i_{6}}a_{1,i_{1}i_{2}}a_{2,i_{2}i_{3}}a_{3,i_{4}i_{5}}a_{4,i_{5}i_{6}}r_{i_{2}}r_{i_{5}}\delta_{i_{3}i_{4}}\delta_{i_{6}i_{1}}\right].
    \end{equation}
    There are now 10 cases to consider, which we label $(III.a)$--$(III.j)$. All of them satisfy $i_{3}\neq i_{4}$, $i_{6}\neq i_{1}$
    \begin{center}
    \begin{tabular}{llllllllll}
    	\hline
    a.&	$i_{2}=i_{5}$& $i_{3}=i_{6}$& $i_{4}=i_{1}$& $i_{3}\neq i_{1}$&f.&	          & 	  & $i_{3}=i_{1}$ & $i_{5}=i_{6}$ \\
    b.&	     & $i_{3}=i_{1}$& $i_{4}=i_{6}$& $i\neq i_{4}$&g.&	          & $i_{2}=i_{6}$ & $i_{3}=i_{5}$ & $i_{4}=i_{1}$ \\
    c.&	$i_{2}\neq i_{5}$ & $i_{2}=i_{3}$ & $i_{4}=i_{6}$ & $i_{5}=i_{1}$ &h.&	          &       & $i_{3}=i$ & $i_{5}=i_{4}$ \\
    d.&	          &       & $i_{4}=i_{1}$ & $i_{5}=i_{6}$ &i.&	          & $i_{2}=i_{1}$ & $i_{3}=i_{6}$ & $i_{4}=i_{5}$ \\
    e.&	          & $i_{2}=i_{4}$ & $i_{3}=i_{6}$ & $i_{5}=i_{1}$ &j.&	          &       & $i_{3}=i_{5}$ & $i_{4}=i_{6}$\\
    \hline
    \end{tabular}
    \end{center}
    We work out $(III.a)$--$(III.c)$ explicitly. The remaining cases follow by analogous calculations.
    \begin{align*}
    (III.a) & = \sum_{i_{1},i_{2},i_{3}\neq i_{1}}a_{1,i_{1}i_{2}}a_{2,i_{2}i_{3}}a_{3,i_{1}i_{2}}a_{4,i_{2}i_{3}}\\
    & = \sum_{i_{1},i_{2},i_{3}\neq i_{1}}\bs e_{i_{1}}'(\bs A_1\odot\bs A_3)\bs e_{i_{2}}\bs e_{i_{2}}'
    (\bs A_2\odot\bs A_4)\bs e_{i_{3}}\\
    &=\bs\iota'(\bs A_1\odot \bs A_3)(\bs A_2\odot \bs A_4)\bs\iota \underbrace{- \tr((\bs A_1\odot\bs A_3)(\bs A_2\odot\bs A_4))}_{(III.a.2)}.\\
    (III.b) & = \sum_{i_{1},i_{2},i_{4}\neq i_{1}}a_{1,i_{1}i_{2}}a_{2,i_{2}i_{1}}a_{3,i_{4}i_{2}}a_{4,i_{2}i_{4}} = \tr(\bs A_1'\bs A_2'\odot \bs A_3'\bs A_4')+(III.a.2).\\
    (III.c) & = \sum_{i_{1},i_{2}\neq \{i_{1},i_{4}\},i_{4}\neq \{i_{1},i_{2}\}}a_{1,i_{1}i_{2}}a_{2,i_{2}i_{2}}a_{4,i_{1}i_{4}}a_{3,i_{4}i_{1}}\\
    &=\sum_{i_{1},i_{2}\neq i_{1}}\bs e_{i_{1}}\bs D_{A_4A_3}\bs A_1\bs D_{A_2}\bs e_{i_{2}} -\bs e_{i_{1}}'\bs D_{A_3}\bs D_{A_4}\bs A_1\bs D_{A_2}\bs e_{i_{2}}\\
    &\quad -\bs e_{i_{1}}'(\bs A_1\bs D_{A_2})\odot \bs A_4\odot \bs A_1'\bs e_{i_{2}}\\
    & = \bs \iota'\bs D_{A_4A_3}\bs A_1\bs D_{A_2}\bs\iota - \tr(\bs D_{A_4A_3}\bs A_1\bs D_{A_2}) - \bs\iota'\bs D_{A_3}\bs D_{A_4}\bs A_1\bs D_{A_2}\bs\iota \\
    & \quad+2\tr(\bs D_{A_3}\bs D_{A_4}\bs A_1\bs D_{A_2}) - \bs\iota'(\bs A_1\bs D_{A_2})\odot \bs A_4\odot \bs A_3'\bs\iota.
\end{align*}
    
    There are many repeated elements in the expressions for $(III.d)$--$(III.j)$. We introduce the following notation
    
    \begin{equation*}
        \begin{array}{rlrl}
    	(c.1) & = \bs\iota'\bs D_{A_4A_3}\bs A_1\bs D_{A_2}\bs\iota,&
    	(c.2) & = -\tr(\bs D_{A_4A_3}\bs A_1\bs D_{A_2}),\\
    	(c.3) & = -\bs\iota'\bs D_{A_3}\bs D_{A_4}\bs A_1\bs D_{A_2}\bs\iota,&
    	(c.4) & = \tr(\bs D_{A_3}\bs D_{A_4}\bs A_1\bs D_{A_2}),\\
    	(c.5) & = -\bs\iota'(\bs A_1\bs D_{A_2})\odot \bs A_4\odot \bs A_3'\bs\iota,&
    	(d.1) & =  \bs\iota'\bs D_{A_4}\bs A_3'\bs A_1(\bs I\odot\bs A_2)\bs\iota, \\
    	(d.2) & = -\tr(\bs D_{A_4}\bs A_3'\bs A_1\bs D_{A_2}),&
    	(e.1) & = \bs\iota'((\bs A_3\odot\bs A_1')\bs A_4)\odot \bs A_2\bs\iota, \\
    	(e.3) & = -\bs\iota'((\bs A_1 \odot \bs A_3'\odot \bs I)\bs A_2)\odot \bs A_4\bs\iota,&
    	(g.1) & = \bs\iota'(\bs A_1\odot(\bs A_3(\bs A_2'\odot\bs A_4)))\bs\iota,\\
    	(g.2) & = -\tr((\bs A_1\odot (\bs A_3(\bs A_2'\odot \bs A_4)))),&
    	(h.1) & = \bs\iota'(\bs A_1\odot \bs A_2')\bs A_4'\bs D_{A_3}\bs\iota,\\
    	(h.5) & = -\bs\iota'((\bs A_1\odot \bs A_2')\odot \bs I)\bs 'A_4(\bs A_3'\odot \bs I)\bs\iota,&
    	(i.1) & = \bs\iota'\bs D_{A_1}\bs A_2\bs A_4'\bs D_{A_3}\bs\iota.
    	\end{array}
    \end{equation*}
    
    Furthermore, let any of these with a asterisk denote the same term but with $\bs A_{2}\rightarrow\bs A_{4}$, $\bs A_{3}\rightarrow \bs A_{1}$, $\bs A_{4}\rightarrow \bs A_{2}$, $\bs A_{1}\rightarrow \bs A_{3}$. Then
    \begin{equation}\begin{split}
    (III.c) & = (c.1)+(c.2)+ (c.3) + (c.4)+ (c.5)+  (c.4),\\
    (III.d) & = (d.1) +(d.2)+ (c.3)^* +(c.4)+ (c.3)+(c.4),\\
    (III.e) & = (e.1) +(c.5)^*+ (e.3)+(c.4)+ (c.5)^* +(c.4),\\
    (III.f) & = (c.1)^* + (c.5)^*+ (c.3)^* +(c.4)+ (c.2)^* +(c.4),\\
    (III.g) & = (g.1) +(g.2)^*+ (g.2) +(c.4)+ (d.2)^*  +(c.4),\\
    (III.h) & = (h.1) + (g.2)^*+ (c.2)^*  +(c.4)+ (h.5) + (c.4),\\
    (III.i) & = (i.1) +(e.3)+ (h.5)^* +(c.4)+ (h.5) +(c.4),\\
    (III.j) & = (h.1)^* + (g.2)+  (h.5)^* + (c.4)+ (c.2) + (c.4).
    \end{split}\end{equation}
    Putting everything together, we obtain the desired result.
    \item Can be obtained from Item \ref{it:rADrArrADrAr proof} by only considering the terms $(I)$ and $(II)'$ in the proof.
    \item See \citet{bao2010expectation}, Theorem 2.
    \item See \citet{ullah2004finite}, Appendix A5.
\end{enumerate}
\end{proof}

\subsection{Proof of Lemma \texorpdfstring{\ref{lem:AR score}}{4}}\label{sapp:proof AR score}
\begin{proof} The proof of Lemma \ref{lem:AR score} is similar to the proof of Lemma A2 in \citet{chao2012asymptotic}.
\subsubsection{Rewriting the statistic}\label{sapp:rewrite}
First we rewrite the $\AR$ statistic. In Section \ref{sec:group} we showed that $\frac{1}{\sqrt{k}}(\AR(\bs\beta_{0})-k)\overset{(d)}{=}\frac{1}{\sqrt{k}}(\AR_r(\bs\beta_{0})-k)$. Then defining
\begin{align*}
    w_{1n,\AR} &= \frac{2}{\sqrt{k}}P_{12},\quad
    y_{in,\AR} = \frac{2}{\sqrt{k}}\left[\sum_{j<i}P_{ij}r_{j}\right]\cdot r_{i},
\end{align*}
we have $\frac{1}{\sqrt{k}}(\AR_r(\bs\beta_{0})-k)=w_{1n,\AR}+\sum^n_{i=3}y_{in,\AR}$.

Next, we consider the score. We rewrite the first order conditions as
\begin{align*}
	\frac{\partial Q(\bs\beta)}{\partial \beta_{l}}\bigg|_{\bs\beta=\bs\beta_{0}}& = -\frac{1}{n}\bs x_{(l)}'(\bs I-\bs D_{P\iota})\bs V\bs\varepsilon\\
	&=-\frac{1}{n}\bigg[\bar{\bs x}_{(l)}'(\bs I-\bs D_{P\iota})\bs V\bs\varepsilon+ \bs\varepsilon'\bs D_{a_{(l)}}(\bs I-\bs D_{P\iota})\bs V\bs\varepsilon\bigg]\\
	&\overset{(d)}{=}-\frac{1}{n}\bigg[\bar{\bs x}_{(l)}'\bs V\bs D_{\varepsilon}\bs r-\bs r'\bs P\bs D_{r}\bs D_{\bar{x}_{(l)}}\bs V\bs D_{\varepsilon}\bs r+\bs r'\bs D_{a_{(l)}}\bs P\bs r- \bs r'\bs P\bs D_{a_{(l)}}\bs P\bs r\bigg]\\
	&=-\frac{1}{n}\bigg[\bar{\bs x}_{(l)}'\bs V\bs D_{\varepsilon}\bs r + \bs r'\bs \Psi^{(l)}\bs r - \bs r'\bs P\bs D_{r}\bs D_{\bar{x}_{(l)}}\bs V\bs D_{\varepsilon}\bs r\bigg],
\end{align*}
where
\begin{equation}\label{eq:defPsihandPhih}
    \bs\Psi^{(l)}\equiv \bs M\bs D_{a_{(l)}}\bs P,\quad \bs\Phi^{(l)}\equiv \bs D_{\bar{x}_{(l)}}\bs V\bs D_{\varepsilon}.
\end{equation}
We rewrite the final term as
\begin{align*}
	\bs r'\bs P\bs D_{r}\bs \Phi^{(l)}\bs r &= \tr(\bs P\bs D_{r}\bs\Phi^{(l)}) + \tr(\bs P\bs D_{r}\bs\Phi^{(l)} \bs \Delta)\qquad \bs\Delta\equiv \bs r\bs r'-\bs I_n\\
	&=\tr(\bs\Phi^{(l)}\bs D_{r}) + \tr(\bs P\bs D_{r}\bs \Phi^{(l)}\bs \Delta)\\
	&=\bar{\bs x}_{(l)}'\bs D_{V}\bs D_{\varepsilon}\bs r + \sum_{\substack{i_{1},i_{2},i_{3}\\ i_{1}\neq i_{3}}}P_{i_{1}i_{2}}\Phi_{i_{2}i_{3}}^{(l)}r_{i_{1}}r_{i_{2}}r_{i_{3}}\\
	&=\bar{\bs x}_{(l)}'\bs D_{V}\bs D_{\varepsilon}\bs r + \sum_{\substack{i_{1}, i_{2}\neq i_{1},\\ i_{3}\neq i_{1}}}P_{i_{1}i_{2}}\Phi_{i_{2}i_{3}}^{(l)}r_{i_{1}}r_{i_{2}}r_{i_{3}} +\sum_{i_{1}\neq i_{3}}P_{i_{1}i_{1}}\Phi_{i_{1}i_{3}}^{(l)}r_{i_{3}}\\
	&=\bar{\bs x}_{(l)}'\bs D_{V}\bs D_{\varepsilon}\bs r + \sum_{\substack{i_{1}, i_{2}\neq i_{1}, \\i_{3}\neq \{i_{1},i_{2}\}}}P_{i_{1}i_{2}}\Phi_{i_{2}i_{3}}^{(l)}r_{i_{1}}r_{i_{2}}r_{i_{3}} +\sum_{i_{1}\neq i_{3}}P_{i_{1}i_{1}}\Phi_{i_{1}i_{3}}^{(l)}r_{i_{3}}+\sum_{i_{1}\neq i_{2}}P_{i_{1}i_{2}}\Phi_{i_{2}i_{2}}^{(l)}r_{i_{1}}\\
		&=\bar{\bs x}_{(l)}'\bs D_{V}\bs D_{\varepsilon}\bs r + \sum_{\substack{i_{1}, i_{2}\neq i_{1},\\ i_{3}\neq \{i_{1},i_{2}\}}}P_{i_{1}i_{2}}\Phi_{i_{1}i_{3}}^{(l)}r_{i_{1}}r_{i_{2}}r_{i_{3}} +\sum_{i_{1}\neq i_{3}}P_{i_{1}i_{1}}\Phi_{i_{1}i_{3}}^{(l)}r_{i_{3}}+\sum_{i_{1}\neq i_{2}}P_{i_{1}i_{1}}\Phi_{i_{1}i_{2}}^{(l)}r_{i_{1}}\\
	&=\bar{\bs x}_{(l)}'\bs D_{V}\bs D_{\varepsilon}\bs r +2\sum_{i_{2}\neq i_{1}}\Phi_{i_{2}i_{1}}^{(l)}P_{i_{2}i_{2}} r_{i_{1}}  + \sum_{\substack{i_{1}, i_{2}\neq i_{1},\\ i_{3}\neq \{i_{1},i_{2}\}}}P_{i_{1}i_{2}}\Phi_{i_{2}i_{3}}^{(l)}r_{i_{1}}r_{i_{2}}r_{i_{3}}.
\end{align*}
Notice that $\bs\Phi^{(l)}\bs P=\bs \Phi^{(l)}$ and hence $\bs\Phi^{(l)}(\bs\Psi^{(l)})'=\bs \Phi^{(l)}\bs D_{a_{(l)}}\bs M$. Furthermore, $\tr(\bs \Psi^{(l)})=0$. We conclude that,
\begin{align*}
		\sqrt{n}\frac{\partial Q(\bs\beta)}{\partial \beta_{l}}&\overset{(d)}{=}\frac{1}{\sqrt{n}}\bigg(-\sum_{i_{2}\neq i_{1}}\Phi_{i_{2}i_{1}}^{(l)}(1-2P_{i_{2}i_{2}})r_{i_{1}} - \sum_{i_{2}\neq i_{1}}\Psi_{i_{2}i_{1}}^{(l)}r_{i_{2}}r_{i_{1}}+\sum_{\substack{i_{1}, i_{2}\neq i_{1}, \\i_{3}\neq \{i_{1},i_{2}\}}}P_{i_{1}i_{2}}\Phi_{i_{2}i_{3}}^{(l)}r_{i_{1}}r_{i_{2}}r_{i_{3}}\bigg)\\
		&=w_{1n,S}^{(l)}+  \sum_{i_{1}=3}^{n}y_{i_{1}n,S}^{(l)},
\end{align*}
where we defined
\begin{align*}
		w_{1n,S}^{(l)}& = -\frac{1}{\sqrt{n}}\sum_{i_{1}\neq 1}\Phi_{i_{1}1}^{(l)}(1-2P_{i_{1}i_{1}})r_{1} -\frac{1}{\sqrt{n}}\sum_{i_{1}\neq 2}\Phi_{i_{1}2}^{(l)}(1-2P_{i_{1}i_{1}})r_{2}  - \frac{1}{\sqrt{n}}\Psi_{[21]}^{(l)}r_{2}r_{1}, \\
		y_{i_{1}n,S}^{(l)}&=\bigg[-\frac{1}{\sqrt{n}}\sum_{i_{2}\neq i_{1}}\Phi_{i_{2}i_{1}}^{(l)}(1-2P_{i_{2}i_{2}})- \frac{1}{\sqrt{n}}\sum_{i_{2}<i_{1}}\Psi_{[i_{1}i_{2}]}^{(l)}r_{i_{2}} +\frac{1}{\sqrt{n}}\sum_{i_{3}<i_{2}<i_{1}}A_{[i_{1}i_{2}i_{3}]}^{(l)}r_{i_{2}}r_{i_{3}}\bigg]\cdot r_{i_{1}},\\
		\Psi_{[i_{1}i_{2}]}^{(l)} &= \Psi_{i_{1}i_{2}}^{(l)}+\Psi_{i_{2}i_{1}}^{(l)},\\
		A_{[i_{1}i_{2}i_{3}]}^{(l)}& = A_{i_{1}i_{2}i_{3}}^{(l)}+A_{i_{1}i_{3}i_{2}}^{(l)}+A_{i_{2}i_{1}i_{3}}^{(l)}+A_{i_{2}i_{3}i_{1}}^{(l)}+A_{i_{3}i_{1}i_{2}}^{(l)}+A_{i_{3}i_{2}i_{1}}^{(l)},\quad A_{i_{1}i_{2}i_{3}}^{(l)} = P_{i_{1}i_{2}}\Phi_{i_{2}i_{3}}^{(l)}
\end{align*}

We have now shown the following distributional equivalence,
\begin{align*}
		\bs Y_n=\begin{pmatrix}
			\frac{1}{\sqrt{k}}(\AR(\bs\beta)-k)\\
			\sqrt{n}\cdot \bs S(\bs\beta)'
			\end{pmatrix} \overset{(d)}{=}\bs Y_{nr}= \begin{pmatrix}
				w_{1n,\AR} \\
				\bs w_{1n,S}
				\end{pmatrix} + \sum_{i=3}^{n}\begin{pmatrix}
				y_{in,\AR}\\
				\bs y_{in,S}
				\end{pmatrix}.
\end{align*}

\subsubsection{Conditional distribution of \texorpdfstring{$\bs t'\bs\Sigma_n^{-1/2}\bs Y_{nr}$}{tSigmaY}}\label{sapp:conditional distribution}
To use the Cramér-Wold theorem in Supplementary Appendix \ref{sapp:Cramer-Wold} we need to show that for any $\bs t\in\mathbb{R}^{p+1}$ $\bs t'\bs\Sigma_{n}^{-1/2}\bs Y_{nr}\rightarrow_d\bs t'\bs Z$. When $\bs t=\bs 0$ the condition is trivially satisfied. Therefore, let $\bs t\in\mathbb{R}^{p+1}\setminus\bs 0$ and write $\bs t=C\bs\alpha(\bs\alpha'\bs\alpha)^{-1/2}$ for $\bs\alpha\in\mathbb{R}^{p+1}\setminus\bs 0$. Consider $(\bs\alpha'\bs\alpha)^{-1/2}\bs\alpha'\bs\Sigma_n^{-1/2}\bs Y_{nr}$ and define $\Xi_{n}=\var(\bs\alpha'\bs\Sigma_{n}^{-1/2}\bs Y_{nr}|\mathcal{J})$. Then $(\bs \alpha'\bs\alpha)^{-1/2}\bs\alpha '\bs \Sigma^{-1/2}_{n}\bs Y_{nr} = w_{1n}+ \sum_{i=3}^{n}y_{in}$, where we define
\begin{align}\label{eq:w y}
		w_{1n}&=\Xi_{n}^{-1/2}\left[c_{1n}w_{1n,\AR} + \bs c_{2n}'\bs w_{1n,S}\right],\nonumber\\
		y_{i_{1} n}& = \Xi_{n}^{-1/2}\left[-\frac{1}{\sqrt{n}}\sum_{i_{2}\neq i_{1}}\bs c_{2n}'\bs\phi_{i_{2}i_{1}}(1-2P_{i_{1}i_{1}})\right. - \frac{1}{\sqrt{n}}\sum_{i_{2}<i_{1}}(\bs c_{2n}'\bs\psi_{[i_{1}i_{2}]}- 2c_{1n}\gamma_nP_{i_{1}i_{2}})r_{i_{2}}\nonumber\\
		&\quad  + \left.\frac{1}{\sqrt{n}}\sum_{i_{3}<i_{2}<i_{1}}\bs c_{2n}'\bs a_{[i_{1}i_{2}i_{3}]}r_{i_{3}}r_{i_{2}}\right]\cdot r_{i_{1}},
\end{align}
where $\bs c_{n} = (c_{1n},\bs c_{2n}')'=\bs\Sigma^{-1/2}_{n}\bs\alpha$ for $0<\bs \alpha'\bs \alpha\leq C$, $\bs \phi_{i_{2}i_{1}} = (\Phi_{i_{2}i_{1}}^{(1)},\ldots,\Phi_{i_{2}i_{1}}^{(p)})'$, $\bs \psi_{[i_{2}i_{1}]}= (\Psi_{[i_{2}i_{1}]}^{(1)},\ldots,\Psi_{[i_{2}i_{1}]}^{(p)})'$, $\bs a_{[i_{1}i_{2}i_{3}]} = (A_{[i_{1}i_{2}i_{3}]}^{(1)},\ldots,A_{[i_{1}i_{2}i_{3}]}^{(p)})$ and $\gamma_n=\sqrt{n/k}$. Notice that $\bs c_{n}'\bs c_{n}\leq C$, which implies $c_{1n}^2\leq C$ and $\bs c_{2n}'\bs c_{2n}\leq C$.

For later purposes, it will be useful to write the bracketed term in $y_{in}$ in matrix notation. Define $\bs S_{i-1}$ as the $n\times n$ matrix with in the left-upper $i-1\times i-1$ block the identity matrix and zeroes elsewhere. Define
\begin{equation}\label{eq:defPsiandPhi}
    \bs\Psi=\bs M\bs D_{\sum_{l=1}^{p}c_{2n,l}a_{(l)}}\bs P,\quad \bs\Phi=\bs D_{\sum_{l=1}^{p}c_{2n,l}\bar{x}_{(l)}}\bs V\bs D_{\varepsilon},
\end{equation}
then we can write
\begin{align}\label{eq:A short}
		y_{in} &= \Xi_{n}^{-1/2}\bigg\{-\frac{1}{\sqrt{n}}\bs c_{2n}'\bar{\bs X}'(\bs I_{n}-2\bs D_{P})\dot{\bs V}\bs D_{\varepsilon}\bs e_{i}\nonumber\\
		&\quad - \frac{1}{\sqrt{n}}\bs r'\bs S_{i-1}\left[\left(\bs \Psi+\bs\Psi'-2\bs D_{\Psi}\right) - 2c_{1n}\gamma_n\dot{\bs P}\right]\bs e_{i}  + \frac{1}{\sqrt{n}}\bs r'\bs A_{-i}\bs r\bigg\}\cdot r_{i}\nonumber,\\
		\bs A_{-i}&=\bs S_{i-1}\bs A_{i}\bs S_{i-1}=\bs S_{i-1}[\dot{\bs P}\bs D_{\Phi e_{i}} + \bs D_{\Phi e_{i}}\dot{\bs P}+ \bs P\bs e_{i}\bs e_{i}'\bs\Phi-\bs D_{P e_{i}}\bs D_{e_{i}'\Phi}]\bs S_{i-1}.
\end{align}

We will now show that $(\bs\alpha'\bs\alpha)^{-1/2}\bs\alpha'\bs\Sigma^{-1/2}\bs Y_n$ converges to a standard normally distributed random variable. As in \citet{chao2012asymptotic} we first show that $w_{1n}=o_p(1)$ such that we can focus on $\sum_{i=3}^ny_{in}$. Next, we check conditions of the martingale difference array CLT.

\paragraph{\texorpdfstring{$w_{1n}=o_p(1)$}{w1n=op1} unconditionally}
Consider $w_{1n}$ as defined in \eqref{eq:w y}. Following \citet{chao2012asymptotic}, we show that $w_{1n}=o_{p}(1)$ by showing that $\E[\|w_{1n}\|^4|\mathcal{J}]\rightarrow_{a.s.}0$. To bound the terms from $\bs w_{1n,S}$ we need the following three bounds. First, using Assumption \ref{ass:eigbound}, we have that
\begin{align}\label{eq:pi z v bounded}
		\frac{1}{n^2}\sum_{i=1}^{n}\|\bar{\bs Z}'\bs V\bs e_{i}\varepsilon_{i}\|^4
		&\leq \frac{1}{n^2}\max_{i=1,\ldots,n}\|\bar{\bs Z}'\bs V\bs e_{i}
		\varepsilon_{i}\|^2\sum_{i=1}^n\|\bar{\bs Z}'\bs V\bs e_{i}
		\varepsilon_{i}\|^2\nonumber\\
		&\leq\frac{o_{a.s.}(1)}{n}\sum_{l=1}^{p}\sum_{i=1}^n(\bs e_{l}'\bar{\bs Z}'\bs V\bs e_{i}
		\varepsilon_{i})^2\nonumber\\
		&= \frac{o_{a.s.}(1)}{n}\sum_{l=1}^{p}\bs e_{l}'\bar{\bs Z}'\bs V
		\bar{\bs Z}\bs e_{l}\nonumber\\
		&\leq \frac{o_{a.s.}(1)}{n}\lambdamax(\bs V
		)\sum_{l=1}^{p}\bs e_{l}'\bar{\bs Z}'\bar{\bs Z}\bs e_{l}\rightarrow_{a.s.}0,
\end{align}
by Assumption \ref{ass:eigbound} and where $o_{a.s.}(1)$ is a term converging to zero $a.s.$

Second, under Assumption \ref{ass:eigbound} and by the finite fourth moment of the elements of $\bs U$ following from Assumption \ref{ass:model}, we have
\begin{align}\label{eq:sum phi bounded}
		\E\bigg[\frac{1}{n^2}\sum_{i_{1},i_{2}}\Vert\phi_{i_{1}i_{2}}\Vert^4\bigg|\mathcal{J}\bigg]
		&=\E\bigg[\frac{1}{n^2}\sum_{i_{1},i_{2}}\Vert\bar{\bs X}'\bs e_{i_{1}} V_{i_{1}i_{2}} \varepsilon_{i_{2}}\Vert^4\bigg|\mathcal{J}\bigg]\nonumber\\
	&=\E\bigg[\frac{1}{n^2}\sum_{i_{1},i_{2}}\Vert\bar{\bs X}'\bs e_{i_{1}}\Vert^4(\bs e_{i_{1}}'\bs V\bs D_{\varepsilon}\bs e_{i_{2}}\bs e_{i_{2}}'\bs D_{\varepsilon}\bs V\bs e_{i_{1}})^2\bigg|\mathcal{J}\bigg]\nonumber\\
		&\leq\E[\frac{1}{n^2}\sum_{i_{1},i_{2}}\Vert\bar{\bs X}'\bs e_{i_{1}}\Vert^4(\bs e_{i_{1}}'\bs V\bs D_{\varepsilon}^2\bs V\bs e_{i_{1}}(\bs e_{i_{1}}'\bs V\bs D_{\varepsilon}\bs e_{i_{2}}\bs e_{i_{2}}'\bs D_{\varepsilon}\bs V\bs e_{i_{1}})|\mathcal{J}]\nonumber\\
		&= \E\bigg[\frac{1}{n^2}\sum_{i_{1}}\Vert\bar{\bs X}'\bs e_{i_{1}}\Vert^4V_{i_{1}i_{1}}^2\bigg|\mathcal{J}\bigg]\nonumber\\
		&\leq\E\bigg[\frac{C}{n^2}\sum_{i_{1}}\Vert\bar{\bs X}'\bs e_{i_{1}}\Vert^4\bigg|\mathcal{J}\bigg]\nonumber\\
		&\leq\frac{C}{n^2}\sum_{i_{1}}\bigg(\Vert\bar{\bs Z}'\bs e_{i_{1}}\Vert^4+\E[\Vert\bs U'\bs e_{i_{1}}\Vert^4|\mathcal{J}]\bigg)\rightarrow_{a.s.}0.
\end{align}

Third, as the rows of $\bs U$ are independent and by Theorem 2 in \citet{whittle1960bounds},
\begin{align}\label{eq:UVDe4}
        \frac{1}{n^2}\sum_{i=1}^n\E[\Vert\bs U'\bs V\bs D_{\varepsilon}\bs e_i\Vert^4|\mathcal{J}]   &\leq\frac{C}{n^2}\sum_{l=1}^{p}\sum_{i=1}^n\E[(\bs u_{(l)}'\bs V\bs D_{\varepsilon}\bs e_i)^4|\mathcal{J}]\nonumber\\
        &\leq \frac{C}{n^2}\sum_{l=1}^{p}\max_{i}\E[u_{(l),i}^{4}|\mathcal{J}]\sum_{i=1}^{n}(\bs e_{i}'\bs D_{\varepsilon}\bs V^2\bs D_{\varepsilon}\bs e_{i})^2\nonumber\\
        &\leq \frac{C}{n^2}\sum_{l=1}^{p}\max_{i}\E[u_{(l),i}^{4}|\mathcal{J}]\tr(\bs V^2)\rightarrow_{a.s.}0.
\end{align}

Now we have that, since $\bs c_{2n}'\bs c_{2n}\leq C$ and using the definition of $\bs w_{1n,S}$, 
\begin{align*}
        &\E\left[\Vert\bs c_{2n}'\bs w_{1n,S}\Vert^4\bigg|\mathcal{J}\right]\leq C\cdot \E\left[\Vert\bs w_{1n,S}\Vert^4\bigg|\mathcal{J}\right]\\
	    &= C\E\left[\left. \Vert\frac{-1}{\sqrt{n}}\sum_{i\neq 1}\bs\phi_{i1}(1-2P_{ii})r_{1} -\frac{1}{\sqrt{n}}\sum_{i\neq 2}\bs\phi_{i2}(1-2P_{ii})r_{2}  - \frac{1}{\sqrt{n}}\bs\psi_{[21]}r_{2}r_{1}\Vert^4 \right|\mathcal{J} \right]\\
	    &\leq\frac{C}{n^2}\E\left[ \Vert\sum_{i\neq 1}\bs\phi_{i1}(1-2P_{ii})\Vert^4+\Vert\sum_{i\neq 2}\bs\phi_{i2}(1-2P_{ii})\Vert^4+ \Vert\bs\psi_{[21]}\Vert^4\bigg|\mathcal{J}\right]\\
	    &\leq\frac{C}{n^2}\E\bigg[\Vert\bar{\bs Z}'\bs V\bs e_{1}
	    \varepsilon_{1}\Vert^4 + \Vert\bs U'\bs V\bs e_{1}
	    \varepsilon_{1}\Vert^4+\Vert\bs\phi_{11}(1-2P_{11})\Vert^4\\
	    &\quad+\Vert\bar{\bs Z}'\bs V\bs e_{2}
	    \varepsilon_{2}\Vert^4 + \Vert\bs U'\bs V\bs e_{2}
	    \varepsilon_{2}\Vert^4+\Vert\bs\phi_{22}(1-2P_{22})\Vert^4+p\max_{l=1,\ldots,p}(\Psi_{[21]}^{(l)})^4\bigg|\mathcal{J}\bigg]\\
	    &\leq\frac{C}{n^2}\E\left[\sum_{i}\Vert\bar{\bs Z}'\bs V\bs e_{i}
	    \varepsilon_{i}\Vert^4 + \Vert\bs U'\bs V\bs e_{i}	    \varepsilon_{i}\Vert^4+\Vert\bs\phi_{11}\Vert^4+\Vert\bs\phi_{22}\Vert^4+C\bigg|\mathcal{J}\right]\rightarrow_{a.s.}0,
\end{align*}
where for the final line we use \eqref{eq:pi z v bounded}, \eqref{eq:sum phi bounded}, \eqref{eq:UVDe4}, Assumption \ref{ass:eigbound} and that 
\begin{equation}\label{eq:Psi bound}
   |\Psi_{i_{1}i_{2}}^{(l)}|= |\bs e_{i_{1}}'\bs M\bs D_{a_{(l)}}\bs P\bs e_{i_{2}}|\leq [\bs e_{i_{1}}\bs M\bs e_{i_{1}}\bs e_{i_{2}}'\bs P\bs D_{a_{(l)}^2}\bs P\bs e_{i_{2}}]^{1/2}\leq \max_{i=1,\ldots,n}|a_{(l),i}|\leq C \quad a.s.n.
\end{equation}
with the second inequality by $P_{ii}<1$ $a.s.n.$ by Assumption~\ref{ass:eigbound}.

For the part of $w_{1n}$ due to the AR statistic, we have
\begin{align}
		\E[\|c_{1n}w_{1n,\AR}\|^4|\mathcal{J}]&=\frac{16\cdot c_{1n}^4}{k^2}P_{12}^4
		\leq \frac{C}{k^2}\bigg(\sum_{i=1}^nP_{1i}^2\bigg)^2\leq\frac{C}{k^2}P_{11}^2
		\rightarrow_{a.s.}0.
\end{align}

As in the proof of Lemma A2 in \citet{chao2012asymptotic}, the above results imply that $w_{1n}=c_{1n}w_{1n,\AR} + \bs c_{2n}'\bs w_{1n,S}\rightarrow_{p} 0$ unconditionally, and hence $(\bs\alpha'\bs\alpha)^{-1/2}\bs\alpha'\bs\Sigma_{n}^{-1/2}\bs Y_{n} = \sum_{i=3}^{n}y_{in} + o_{p}(1).$

\paragraph{Martingale difference array}
Define the $\sigma$-fields $\mathcal{F}_{i,n} = \sigma(r_{1},\ldots,r_{i})$ such that $\mathcal{F}_{i-1,n}\subset \mathcal{F}_{i,n}$. We have $\E[y_{in}|\mathcal{J}, \mathcal{F}_{i-1,n}] = 0$, due to the $r_i$ that multiplies all the terms. Hence, conditional on $\mathcal{J}$, $\{y_{in},\mathcal{F}_{i,n},1\leq i\leq n,n\geq 3\}$ is a martingale difference array.

\paragraph{Variance bounded away from zero}\label{subsubsec:varianceboundedawayfromzero}
For our statistic to be well defined we require the existence of $\bs\Sigma^{-1}_n$ almost surely. We start by considering a quadratic form of $\bs\Omega$ defined in \eqref{eq:var} in Lemma \ref{thm:expandvar}. Let $\bs v$ be any $p$ dimensional vector satisfying $\bs v'\bs v=1$. Then,
$\bs v'\bs\Omega\bs v\geq \sum_{i}c_{i}\bs v'\bs\Sigma^{U}_{i}\bs v,$ 
where from \eqref{eq:var} we have $c_{i} = V_{ii}-3V_{ii}P_{ii}+4V_{ii}P_{ii}^2-2\sum_{j=1}^{n}V_{ij}^2\varepsilon_{j}^2P_{ji}^2$. Then since $\bs\Sigma^{U}_{i}$ is positive definite and using that for $i\neq j$, $P_{ij}^2\leq P_{jj}(1-P_{jj})$ 
\begin{align}\label{eq:omega bigger zero}
        c_{i}&=
        V_{ii}(1-3P_{ii}+4P_{ii}^2)-2\left(V_{ii}P_{ii}^3+\sum_{j\neq i}V_{ij}^2\varepsilon_{j}^2P_{ji}^2\right)\nonumber\\
        &\geq V_{ii}(1-3P_{ii}+4P_{ii}^2)-2\left(V_{ii}P_{ii}^3+P_{ii}(1-P_{ii})\sum_{j\neq i}V_{ij}^2\varepsilon_{j}^2\right)\nonumber\\
        &= V_{ii}(1-3P_{ii}+4P_{ii}^2)-2\left(V_{ii}P_{ii}^3+P_{ii}(1-P_{ii})(V_{ii}-V_{ii}P_{ii})\right)\nonumber\\
        &= V_{ii}(1-5P_{ii}+8P_{ii}^2-4P_{ii}^3)\nonumber\\
        &=V_{ii}(1-P_{ii})(1-2P_{ii})^2.
\end{align}
Consider the case where the inequality holds with equality, such that for every $ij$, either  $P_{ij}^2 = P_{jj}(1-P_{jj})$ or $V_{ji}^2\varepsilon_{i}^2=0$. The last condition cannot hold for every $i\neq j$, because $\sum_{i\neq j}V_{ji}^2\varepsilon_{i}^2=\sum_{i} V_{ji}^2\varepsilon_{i}^2-V_{jj}P_{jj}=\bs e_{j}'\bs V\bs D_{\varepsilon}^2\bs V\bs e_{j}-V_{jj}P_{jj}=V_{jj}(1-P_{jj})$. Since $P_{jj}<1$ and $V_{jj}>0$ by Assumption~\ref{ass:eigbound}, we conclude that there must be at least one $i$ such that $V_{ji}^2\varepsilon_{i}^2\neq 0$. For such an $i$, equality therefore only obtains if $P_{ij}^2 = P_{jj}(1-P_{jj})$. Assume moreover that $P_{jj}=1/2$, so that $P_{ij}=\pm \frac{1}{2}$. Using that $\bs P$ is a projection matrix, this implies that $P_{ii}=1/2$, $P_{i_{1}j}^2=0$ for $i_{1}\neq i$ and $P_{i_{2}i}=0$ for $i_{2}\neq j$. However, this means that two columns of $\bs D_{\varepsilon}\bs Z$ coincide up to their sign. This case is excluded by Assumption~\ref{ass:eigbound} that states that $\lambda_{\min}(\bs Z'\bs D_{\varepsilon}^2\bs Z/n)\geq C$ $a.s.n.$ We then have
\begin{align*}
        \frac{1}{n}\sum_{i}c_{i}\bs v'\bs\Sigma^U_i\bs v&\geq\frac{\lambdamin(\bs\Sigma^{U}_{i})}{n}\sum_{i}c_{i}\geq\frac{C}{n}\tr(\bs V)\geq \frac{C\tr(\bs Z\bs Z')}{n\lambdamax(\bs Z'\bs D_{\varepsilon}^2\bs Z)}
        \\&\geq\frac{Ck\lambdamin(\bs Z'\bs Z)}{n\lambdamax(\bs Z'\bs D_{\varepsilon}^2\bs Z)}>0 \quad a.s.n.,
\end{align*}
by Assumption \ref{ass:eigbound} and because $k/n>0$.

Now let $\bs b=[\bs\Sigma]_{2:p+1,1}$ the covariance between the AR statistic and the score. Then $\det(\bs\Sigma_n)=\det(\bs\Omega)\det(\bs\Omega-\bs b\bs b'\sigma^{-2}_n)$ by Schur complements. The $(l_{1},l_{2})^{\text{th}}$ element in $\bs b\bs b'\sigma^{-2}_n$ is the covariance of the AR statistic with $l_{1}^{\text{th}}$ and $l_{2}^{\text{th}}$ element of the score divided by the variance of the AR statistic. Hence this is equal to the correlation of the AR statistic with the $l_{1}^{\text{th}}$ and $l_{2}^{\text{th}}$ element of the score statistic times the standard deviations of the $l_{1}^{\text{th}}$ and $l_{2}^{\text{th}}$ element of the score statistic. Let $\bs\rho$ be the vector of correlations between the AR statistic and the score. That is, $\rho_{l}=\corr(1/\sqrt{k}(\text{AR}(\bm \beta_0)-k),1/\sqrt{n}S_{l}(\bs\beta_{0})|\mathcal{J})$. Then,
\begin{align*}
        \det(\bs\Sigma_n)&=\det(\bs\Omega)\det(\bs\Omega-\bs b\bs b'\sigma^{-2}_n)
        =\det(\bs\Omega)\det(\bs I+\bs D_{\rho})\det(\bs\Omega)\det(\bs I-\bs D_{\rho})>0,
\end{align*}
if $\rho_{l}\neq\pm 1$ for all $l$. We now prove that this is indeed the case. Consider first the variance of the score. We have with $[\bs D_{c}]_{ii}=c_{i}$ from \eqref{eq:omega bigger zero}
\begin{align*}
        \Omega_{ll}&\geq \frac{1}{n}\left[\tr(\bs\Sigma^{U}_{(l,l)}\bs D_{c}) +\bs a_{(l)}'(\bs D_{P}-\bs P\odot\bs P)(\bs I- 2(\bs D_{P}-(\bs P\odot \bs P)))\bs a_{(l)}\right].
\end{align*}
Define $\bs\Delta = \bs D_{P}-\bs P\odot\bs P$. Then, for the squared correlation coefficient, we get
\begin{equation*}
    \rho_{l}^2 = 2\frac{(n^{-1}\bs a_{(l)}'\bs\Delta\bs D_{P}\bs\iota)^2}{n^{-1}\tr(\bs\Delta)[n^{-1}\tr(\bs\Sigma^{U}_{(l,l)}\bs D_{c})+n^{-1}\bs a_{(l)}'\bs\Delta(\bs I-2\bs\Delta)\bs a_{(l)}]}.
\end{equation*}
As this is a correlation coefficient, we have $|\rho_{l}^2|\leq 1$ and this holds even if $\tr(\bs\Sigma^{U}_{(l,l)}\bs D_{c})$ is arbitrarily small. However, as $n^{-1}\tr(\bs\Sigma^{U}_{(l,l)}\bs D_{c})\geq C>0$ $a.s.n$., we have $|\rho_{l}^2|\leq C<1$ $a.s.n.$ We conclude that $\bs\Sigma^{-1}_n$ exists $a.s.n.$

\paragraph{Lyapunov condition}
In this section we show that the martingale difference array $\{y_{in},\mathcal{F}_{i,n},1\leq i\leq n, n\geq 3\}$ satisfies the following Lyapunov condition 
\begin{align}\label{eq:lyapunov}
	\sum_{i=3}^{n}\E[y_{in}^4|\mathcal{J}]
	&\leq C \Xi_{n}^{-2}\sum_{i=3}^{n}\underbrace{\E\bigg[\bigg(-\frac{1}{\sqrt{n}}\sum_{j\neq i}\bs c_{2n}'\bs\phi_{ji}(1-2P_{ii})r_i\bigg)^4\bigg|\mathcal{J}\bigg]}_{\text{linear}} \nonumber\\
	&\quad  +\underbrace{\E\bigg[\bigg(\frac{1}{\sqrt{n}}\sum_{j<i}(\bs c_{2n}'\bs\psi_{[ij]}- 2c_{1n}\gamma_nP_{ij})r_{j}r_i\bigg)^4\bigg|\mathcal{J}\bigg]}_{\text{quadratic}}\nonumber \\
	&\quad+ \underbrace{\E\bigg[\bigg(\frac{1}{\sqrt{n}}\sum_{i_{1}<j<i}\bs c_{2n}'\bs a_{[iji_{1}]}r_{i_{1}}r_{j} r_{i}\bigg)^4\bigg|\mathcal{J}\bigg]}_{\text{cubic}}\rightarrow_{a.s.}0.
\end{align}

Since the variance $\bs\Sigma_{n}$ was shown in Appendix \ref{subsubsec:varianceboundedawayfromzero} to be bounded away from zero, $\Xi^{-2}_n$ is finite, as $
	\Xi_n=\var(\bs\alpha'\bs\Sigma_n^{-\frac{1}{2}}\bs Y_{nr}|\mathcal{J})=(\bs\alpha'\bs\alpha)\var( w_{1n}+\sum^{n}_{i=3} y_{in}|\mathcal{J})=(\bs\alpha'\bs\alpha)(1+o_{a.s.}(1))>0$. We now subsequently consider the linear, quadratic and cubic terms in \eqref{eq:lyapunov}.

\subparagraph{Linear term} For the term linear in $\bs r$, we have that
\begin{align*}
	&\E\bigg[\frac{1}{n^2}\sum_{i=3}^{n}\bigg(\sum_{j\neq i}\bs c_{2n}'\bs \phi_{ji}(1-2P_{ii})\bigg)^4\bigg|\mathcal{J}\bigg]\\
	&\leq \E\bigg[\frac{C}{n^2}\sum_{i=3}^{n}(1-2P_{ii})^4 \|\sum_{j=1}^{n}\bs \phi_{ji}-\bs\phi_{ii}\|^4\bigg|\mathcal{J}\bigg]\\
	&\leq \E\bigg[\frac{C}{n^2}\sum_{i=3}^{n} \left(\|\bar{\bs Z}'\bs V\bs D_{\varepsilon}\bs e_{i}\|^4 \right.+\|\bs U'\bs V\bs D_{\varepsilon}\bs e_{i}\|^4  \left.+\Vert\bs\phi_{ii}\Vert^4\right)\bigg|\mathcal{J}\bigg]\rightarrow_{a.s.}0,\\
\end{align*}
since $(1-2P_{ii})^2<1$ and by Assumption \ref{ass:eigbound}, \eqref{eq:pi z v bounded}, \eqref{eq:sum phi bounded} and \eqref{eq:UVDe4}.

\subparagraph{Quadratic term} For the term quadratic in $\bs r$ in \eqref{eq:lyapunov}, we first notice that
\begin{equation*}
	\frac{1}{n^2}\sum_{i=3}^{n}\E\bigg[\|\sum_{j<i}\gamma_nP_{ij}r_ir_{j}\|^4\bigg|\mathcal{J}\bigg] \leq \frac{\gamma_n^2}{n^2}\sum_{i=3}^{n}\bigg(\sum_{j<i}P_{ij}^4 + 3\sum_{\substack{(j,i_{1})<i\\j\neq i_{1}}}P_{ij}^2P_{ii_{1}}^2\bigg)\leq C\frac{k}{nk}\rightarrow 0.
\end{equation*}
Similarly,
\begin{align}\label{eq:Psi4}
	&\frac{1}{n^2}\sum_{i=3}^{n}\E\bigg[\|\sum_{j<i}\bs c_{2n}'\bs\psi_{[ij]} r_{i}r_{j}\|^4\bigg|\mathcal{J}\bigg]\nonumber\\
	&\leq \frac{C}{n^2}\sum_{i_{1}=3}^{n}\sum_{i_{2}<i_{1}}\sum_{i_{3}<i_{1}}\sum_{i_{4}<i_{1}}\sum_{i_{5}<i_{1}}|\bs c_{2n}'\bs \psi_{[i_{1}i_{2}]}||\bs c_{2n}'\bs\psi_{[i_{1}i_{3}]}||\bs c_{2n}'\bs\psi_{[i_{1}i_{4}]}||\bs c_{2n}'\bs\psi_{[i_{1}i_{5}]}|\E[r_{i_{2}}r_{i_{3}}r_{i_{4}}r_{i_{5}}|\mathcal{J}]\nonumber\\
	&\leq\frac{C}{n^2}\sum_{i_{1}=3}^{n}\bigg(\sum_{i_{2}<i_{1}}\big(\bs c_{2n}'\bs \psi_{[i_{1}i_{2}]}\big)^4 + 3\sum_{\substack{(i_{2},i_{3})<i\\i_{2}\neq i_{3}}} \big(\bs c_{2n}'\bs\psi_{[i_{1}i_{2}]}\big)^2 \big(\bs c_{2n}'\bs \psi_{[i_{1}i_{3}]}\big)^2\bigg)\nonumber\\
	&\leq \frac{C}{n^2}\sum_{i_{1}=3}^{n}\bigg(\sum_{i_{2}<i_{1}}\sum_{l=1}^{p}\big(\Psi_{[i_{1}i_{2}]}^{(l)}\big)^4 + 3\sum_{\substack{(i_{2},i_{3})<i_{1}\\i_{2}\neq i_{3}}} \sum_{l=1}^{p}\big(\Psi_{[i_{1}i_{2}]}^{(l)}\big)^2\sum_{l=1}^{p}\big(\Psi_{[i_{1}i_{3}]}^{(l)}\big)^2\bigg).
\end{align}

To bound this expression, note that by \eqref{eq:Psi bound} $\bs e_{i}'\bs\Psi^{(l)} \bs e_{i} \leq C$ $a.s.n.$
Also, for any vector $\bs v$, $(\bs\Psi^{(l)} \bs v)^2= \bs v'\bs P\bs D_{a_{(l)}}\bs M\bs D_{a_{(l)}}\bs P\bs v\leq \max_{i=1,\ldots,n} a_{(l),i}^2\cdot \bs v'\bs P\bs v.$
This implies that $
	\sum_{i=1}^{n}(\Psi_{ij}^{(l)})^2\leq \max_{i=1,\ldots,n} a_{(l),i}^2 \cdot P_{jj}\leq C\quad a.s.n.$
Then,
\begin{align*}
		\frac{1}{n^2}\sum_{i_{1},i_{2}}(\Psi^{(l)}_{i_{1}i_{2}})^4&\leq\frac{1}{n^2}\sum_{i_{1},i_{2}}\left(\sum_{i_{3}}(\Psi^{(l)}_{i_{3}i_{2}})^2\right)(\Psi^{(l)}_{i_{1}i_{2}})^2\\
		&\leq\frac{1}{n^2}\max_{i=1,\ldots,n}a_{(l),i}\sum_{i_{2}}P_{i_{2}i_{2}}\sum{i_{1}}(\Psi_{i_{1}i_{2}}^{(l)})^2\\
		&\leq\frac{1}{n^2}\max_{i=1,\ldots,n}a_{(l),i}^2\sum_{i_{2}}P_{i_{2}i_{2}}^2\leq \frac{Ck}{n^2}\rightarrow_{a.s.}0.
\end{align*}
Using this result, we have for the first term on the final line of \eqref{eq:Psi4}
\begin{align*}
		\frac{1}{n^2}\sum_{i_{1}=3}^{n}\sum_{i_{2}<i_{1}}\left(\Psi_{[i_{1}i_{2}]}^{(l)}\right)^4&\leq\frac{1}{n^2} \sum_{i_{1},i_{2}}\left(\Psi_{i_{1}i_{2}}^{(l)}+\Psi^{(l)}_{i_{2}i_{1}}\right)^4\leq\frac{C}{n^2} \sum_{i_{1},i_{2}}\big[(\Psi_{i_{1}i_{2}}^{(l)})^4+(\Psi_{i_{2}i_{1}}^{(l)})^4\big]\leq \frac{Ck}{n^2}\rightarrow_{a.s.}0.
\end{align*}
For the second term on the final line of \eqref{eq:Psi4}, we have
\begin{align*}
		\frac{1}{n^2}\sum_{i_{1}=3}^{n}\sum_{\substack{(i_{2},i_{3})<i_{1}\\ i_{2}\neq i_{3}}} \big(\Psi_{[i_{1}i_{2}]}^{(l)}\big)^2\big(\Psi_{[i_{1}i_{3}]}^{(l)}\big)^2&\leq \frac{C}{n^2}\sum_{i_{1},i_{2},i_{3}}\big((\Psi_{i_{1}i_{2}}^{(l)})^2+(\Psi_{i_{2}i_{1}}^{(l)})^2\big)\big((\Psi_{i_{1}i_{3}}^{(l)})^2+(\Psi_{i_{3}i_{1}}^{(l)})^2\big)\\
		&\leq \frac{C}{n^2}\sum_{i_{1},i_{2},i_{3}}\bigg[(\Psi_{i_{1}i_{2}}^{(l)})^2(\Psi_{i_{1}i_{3}}^{(l)})^2+(\Psi_{i_{1}i_{2}}^{(l)})^2(\Psi_{i_{3}i_{1}}^{(l)})^2\\
		&\quad\qquad\qquad +(\Psi_{i_{2}i_{1}}^{(l)})^2(\Psi_{i_{1}i_{3}}^{(l)})^2+(\Psi_{i_{2}i_{1}}^{(l)})^2(\Psi_{i_{2}i_{1}}^{(l)})^2\bigg].
\end{align*}
We now show almost sure convergence to zero of the sums over the four terms within the brackets. First, 
\begin{align*}
		\frac{1}{n^2}\sum_{i_{1},i_{2},i_{3}}(\Psi^{(l)}_{i_{1}i_{2}})^2(\Psi^{(l)}_{i_{1}i_{3}})^2&\leq\frac{1}{n^2}\sum_{i_{1}}\left(\sum_{i_{2}}(\Psi^{(l)}_{i_{1}i_{2}})^2\right)^2\nonumber\\
		&\leq\frac{1}{n^2}\sum_{i_{1}}\left(\bs e_{i_{1}}'\bs\Psi^{(l)}\bs\Psi^{(l)\prime}\bs e_{i_{1}}\right)^2\nonumber\\
		&\leq\frac{1}{n^2}\tr(\bs M\bs D_{a_{(l)}}\bs P\bs D_{a_{(l)}}\bs M\bs D_{a_{(l)}}\bs P\bs D_{a_{(l)}}\bs M)\\
		&\leq \frac{Ck}{n^2}\rightarrow_{a.s.}0.
\end{align*}
For the second, and likewise for the third term,
\begin{align*}
		\frac{1}{n^2}\sum_{i_{1},i_{2},i_{3}}(\Psi^{(l)}_{i_{1}i_{2}})^2(\Psi^{(l)}_{i_{3}i_{1}})^2&\leq\frac{1}{n^2}\sum_{i_{1}}\bs e_{i_{1}}' \bs\Psi^{(l)\prime}\bs\Psi^{(l)}\bs e_{i_{1}}\bs e_{i_{1}}' \bs\Psi^{(l)}\bs\Psi^{(l)\prime}\bs e_{i_{1}}\\
		&\leq\frac{1}{n^2}\sum_{i_{1}}|\bs e_{i_{1}}' \bs M\bs D_{a_{(l)}}\bs P\bs D_{a_{(l)}}\bs M\bs e_{i_{1}}||\bs e_{i_{1}}' \bs P\bs D_{a_{(l)}}\bs M\bs D_{a_{(l)}}\bs P\bs e_{i_{1}}|\\
		&\leq \frac{1}{n^2}\sum_{i_{1}}|\max_{i_{2}=1,\ldots,n}a_{(l),i_{2}}^2|^2P_{i_{1}i_{1}}\leq \frac{Ck}{n^2}\rightarrow_{a.s.}0.
\end{align*}
For the fourth and final term we have,
\begin{align*}
		\frac{1}{n^2}\sum_{i_{1},i_{2},i_{3}}(\Psi^{(l)}_{i_{2}i_{1}})^2(\Psi^{(l)}_{i_{3}i_{1}})^2&\leq\frac{1}{n^2}\sum_{i_{1},i_{2}}(\Psi^{(l)}_{i_{2}i_{1}})^2\sum_{i_{3}}(\Psi^{(l)}_{i_{3}i_{1}})^2\\
        &\leq\frac{1}{n^2}\sum_{i_{1}}\max_{i_{2}=1,\ldots,n}a_{(l),i_{2}}^4 P_{i_{1}i_{1}}^2\leq \frac{Ck}{n^2}\rightarrow_{a.s.}0.
\end{align*}
Consequently, the quadratic term in \eqref{eq:lyapunov} converges to zero almost surely.

\subparagraph{Cubic term}
From \eqref{eq:A short}, the cubic term can be written as,
\begin{align}\label{eq:Lyapunov cubic}
		&\sum^n_{i=3}\E\bigg[\bigg(\frac{1}{\sqrt{n}}\bs r'\bs A_{-i}\bs r r_{i}\bigg)^4\bigg|\mathcal{J}\bigg]=\sum^n_{i=3}\frac{C}{n^2}\E[\E[(\bs r'\bs A_{-i}\bs r)^4|\mathcal{J},\bs U]|\mathcal{J}].
\end{align}
As $\bs A_{-i}$ is symmetric with zeroes on its diagonal, we have by Lemma \ref{thm:Rademacher} in Appendix \ref{sapp:rademacher}
\begin{align}\label{eq:cubic92trAi2}
	&\sum^n_{i=3}\frac{C}{n^2}\E[\E[(\bs r'\bs A_{-i}\bs r)^4|\mathcal{J},\bs U]|\mathcal{J}]\nonumber\\
	&\leq\sum^n_{i=3}\frac{C}{n^2}\E[
 \tr(\bs A_{-i}^2)^2+
 \tr(\bs A_{-i}^4)
	+
 \bs\iota'(\bs A_{-i}\odot \bs A_{-i}\odot \bs A_{-i}\odot \bs A_{-i})\bs\iota|\mathcal{J}]\nonumber\\
	&\leq \sum^n_{i=3}\frac{C}{n^2}\E[
 \tr(\bs A_{-i}^2)^2|\mathcal{J}].
\end{align}
The second inequality follows since $\bs A_{-i}^2$ is p.s.d., hence $\tr(\bs A_{-i}^4)\leq \tr(\bs A_{-i}^2)^2$, and
\begin{align*}
		\bs\iota'(\bs A_{-i}\odot\bs A_{-i}\odot\bs A_{-i}\odot\bs A_{-i})\bs\iota = \sum_{i=1}^{n}\sum_{j=1}^{n}(\bs e_{i}'\bs A_{-i}\bs e_{j})^4\leq  \sum_{i=1}^{n}\sum_{j=1}^{n}\bs e_{i}'\bs A_{-i}^2 \bs e_{i} (\bs e_{i}'\bs A_{-i}\bs e_{j})^2\\
		=\sum_{i=1}^{n}(\bs e_{i}'\bs A_{-i}^2 \bs e_{i})^2\leq \left(\sum_{i=1}^{n}(\bs e_{i}'\bs A_{-i}^2 \bs e_{i})\right)^2= \tr(\bs A_{-i}^2)^2.
\end{align*}
The final line of \eqref{eq:cubic92trAi2} can be further bounded as
\begin{align}\label{eq:cubicEtrAi2}
        \sum^n_{i=3}\frac{C}{n^2}\E[\tr(\bs A_{-i}^2)|\mathcal{J}]
        &=\sum^n_{i=3}\frac{C}{n^2}\E[\tr(\bs S_{i-1}\bs A_{i}\bs S_{i-1}\bs A_{i}\bs S_{i-1})|\mathcal{J}]\nonumber\\
        &\leq \sum^n_{i=3}\frac{C}{n^2}\E[\tr([\dot{\bs P}\bs D_{\Phi e_{i}}]^2 + [\bs D_{\Phi e_{i}}\dot{\bs P}]^2 +[\bs P\bs e_{i}\bs e_{i}'\bs\Phi]^2-[\bs D_{P e_{i}}\bs D_{e_{i}'\Phi}]^2)\big|\mathcal{J}].
\end{align}
To bound these four terms we use the following result
\begin{align}\label{res:sum phi 4}
        \frac{C}{n^2}\sum_{i=1}^n\E[\left(\bs e_i'\bs\Phi^{\prime}\bs\Phi\bs e_i \right)^2|\mathcal{J}]   
        &\leq\frac{C}{n^2}\sum_{i=1}^n\E[\bs e_i'\bs\Phi^{\prime}\bs\Phi\bs\Phi^{\prime}\bs\Phi\bs e_i|\mathcal{J}]\nonumber\\ 
        &\leq\frac{C}{n^2}\sum_{l=1}^p\sum_{i=1}^n\E[\bar{x}_{(l),i}^4|\mathcal{J}]\nonumber\\
        &\leq\frac{C}{n^2}\sum_{i=1}^n\Vert\bar{\bs Z}'\bs e_i\Vert^4+\E[\Vert\bs U'\bs e_i\Vert^4|\mathcal{J}] \rightarrow_{a.s.}0,
\end{align}
by Assumption \ref{ass:eigbound} and the finite fourth moment of the elements of $\bs U$.

For the first and second term of \eqref{eq:cubicEtrAi2}, we have by \eqref{res:sum phi 4}
\begin{align*}
	    \frac{1}{n^2}\sum_{i=1}^n\E[\tr^2(\dot{\bs P}\bs D_{\Phi e_i}\dot{\bs P}\bs D_{\Phi e_i})|\mathcal{J}]
		&\leq\frac{C}{n^2}\sum_{i=1}^n\E\bigg[\bigg(\sum^n_{j=1}\bs e_j'\bs P\bs D^2_{\Phi e_i}\bs P\bs e_j \bigg)^2\bigg|\mathcal{J}\bigg]\\
		&\leq\frac{C}{n^2}\sum_{i=1}^n\E\bigg[\bigg(\sum^n_{j=1}\bs e_j'\bs D^2_{\Phi e_i}\bs e_j \bigg)^2\bigg|\mathcal{J}\bigg]\rightarrow_{a.s.}0.
\end{align*}
For the third term of \eqref{eq:cubicEtrAi2}, also by \eqref{res:sum phi 4},
\begin{align*}
	    \frac{1}{n^2}\sum_{i=1}^n\E[\tr^2(\bs P\bs e_i\bs e_i'\bs\Phi\bs P\bs e_i\bs e_i'\bs\Phi)|\mathcal{J}]
		&=\frac{1}{n^2}\sum_{i=1}^n\E[\tr^2(\bs P\bs e_i\bs e_i'\bs\Phi\bs\Phi^{\prime}\bs e_i\bs e_i'\bs P)|\mathcal{J}]\\
		&\leq\frac{1}{n^2}\sum_{i=1}^n\E[\left(P_{ii} \bs e_i'\bs\Phi\bs\Phi^{\prime}\bs e_i\right)^2|\mathcal{J}]\\
		&\leq\frac{1}{n^2}\sum_{i=1}^n\E[\left(\bs e_i'\bs\Phi\bs\Phi^{\prime}\bs e_i\right)^2|\mathcal{J}]\rightarrow_{a.s.}0.
\end{align*}
And finally, for the fourth term of \eqref{eq:cubicEtrAi2}, \eqref{res:sum phi 4} implies that
\begin{align*}
	    \frac{1}{n^2}\sum_{i=1}^n\E[\tr^2(\bs D_{Pe_i}\bs D_{e_i'\Phi}\bs D_{Pe_i}\bs D_{e_i'\Phi}|\mathcal{J}]
	    &=\frac{1}{n^2}\sum_{i=1}^n\E[\left[\sum_{j=1}^nP_{ji}^2(\Phi_{ji})^2\right]^2|\mathcal{J}]\\
	    &\leq\frac{1}{n^2}\sum_{i=1}^n\E[\left[\sum_{j=1}^nP_{ii}^2(\Phi_{ji})^2\right]^2|\mathcal{J}]\\
	    &\leq\frac{1}{n^2}\sum_{i=1}^n\E[\left[\bs e_i'\bs \Phi^{\prime}\bs \Phi\bs e_i\right]^2|\mathcal{J}]\rightarrow_{a.s.}0.
\end{align*}
Hence the cubic term converges to zero almost surely. Therefore, the Lyapunov condition is satisfied.

\paragraph{Converging conditional variance}
This part of the proofs shows the following convergence result: for any $\epsilon>0$,
\begin{equation*}
    \Pr\bigg(\bigg|\sum_{i=3}^{n}\E[y_{in}^2|\bs r_{-i},\mathcal{J}]-s_{n}^2\bigg|\geq \epsilon\bigg|\mathcal{J}\bigg)\rightarrow_{a.s.}0.
\end{equation*}
We start by noting that,
\begin{align*}
		s_{n}^2 &= \E\bigg[\bigg(\sum_{i=3}^{n}y_{in}\bigg)^2\bigg|\mathcal{J}\bigg]=\E[((\bs\alpha'\bs\alpha)^{-1/2}\bs\alpha'\bs\Sigma_{n}^{-1/2}\bs Y_{n}+o_{a.s.}(1))^2|\mathcal{J}]=1+o_{a.s.}(1),
\end{align*}
where the vanishing part is due to $w_{1n}$. We can conclude that $s_{n}^2$ is bounded and bounded away from zero in probability. Now define $\bs r_{-i}=r_1,\ldots,r_{i-1}$ and write $y_{in}$ in \eqref{eq:A short} as $y_{in}=\Xi_{n}^{-1/2}(y_{in}^{(1)}+y_{in}^{(2)}+y_{in}^{(1)})$ with
\begin{align*}
		y_{in}^{(1)} &=\frac{-1}{\sqrt{n}}\bs c_{2n}'\bar{\bs X}'(\bs I_{n}-2\bs D_{P})\dot{\bs V}\bs D_{\varepsilon}\bs e_{i}r_{i}, \\
		y_{in}^{(2)}&= \frac{-1}{\sqrt{n}}\bs r'\bs S_{i-1}\left[\left(\bs \Psi+\bs\Psi'-2\bs D_{\Psi}\right) - 2c_{1n}\gamma_n\dot{\bs P}\right]\bs e_{i}r_{i}, \\
	    y_{in}^{(3)}&=\frac{1}{\sqrt{n}}\bs r'\bs A_{-i}\bs rr_{i}.
\end{align*}
Using a conditional version of Chebyshev's inequality, we have
\begin{align}\label{eq:plim conditional variance}
	&\Pr\bigg(\bigg|\sum_{i=3}^{n}\E[y_{in}^2|\bs r_{-i},\mathcal{J}]-s_{n}^2(\mathcal{J})\bigg|\geq \epsilon\bigg|\mathcal{J}\bigg)\nonumber\\
	&=\Pr\bigg(\bigg|\sum_{i=3}^{n}\E[y_{in}^2|\bs r_{-i},\mathcal{J}]-\sum^n_{i=3}\E[y_{in}^2|\mathcal{J}]\bigg|\geq \epsilon\bigg|\mathcal{J}\bigg)\nonumber\\
	&\leq \frac{1}{\epsilon^2}\E\bigg[\bigg(\sum_{i=3}^{n}\E[y_{in}^2|\bs r_{-i},\mathcal{J}]-\sum^n_{i=3}\E[y_{in}^2|\mathcal{J}]\bigg)^2\bigg|\mathcal{J}\bigg]\nonumber\\
	&\leq\frac{C}{\epsilon^2}\E\bigg[\bigg(\sum_{i=3}^{n}\E[(y_{in}^{(1)}+y_{in}^{(2)}+y_{in}^{(3)})^2|\bs r_{-i},\mathcal{J}]-\sum^n_{i=3}\E[(y_{in}^{(1)}+y_{in}^{(2)}+y_{in}^{(3)})^2|\mathcal{J}]\bigg)^2\bigg|\mathcal{J}\bigg]\nonumber\\
	&\leq \frac{C}{\epsilon^2}\left\{\E\left[\left.\left(\sum_{i=3}^{n}\E[(y_{in}^{(1)})^2|\bs r_{-i},\mathcal{J}]-\sum^n_{i=3}\E[(y_{in}^{(1)})^2|\mathcal{J}]\right)^2\right|\mathcal{J}\right]\right.\nonumber\\
	&\quad+ \E\left[\left.\left(\sum_{i=3}^{n}\E[(y_{in}^{(2)})^2|\bs r_{-i},\mathcal{J}]-\sum^n_{i=3}\E[(y_{in}^{(2)})^2|\mathcal{J}]\right)^2\right|\mathcal{J}\right]\nonumber\\
	&\quad+ \E\left[\left.\left(\sum_{i=3}^{n}\E[(y_{in}^{(3)})^2|\bs r_{-i},\mathcal{J}]-\sum^n_{i=3}\E[(y_{in}^{(3)})^2|\mathcal{J}]\right)^2\right|\mathcal{J}\right]\nonumber\\
	&\quad+ \E\left[\left.\left(\sum_{i=3}^{n}\E[(y_{in}^{(1)})(y_{in}^{(2)})|\bs r_{-i},\mathcal{J}]-\sum^n_{i=3}\E[(y_{in}^{(1)})(y_{in}^{(2)})|\mathcal{J}]\right)^2\right|\mathcal{J}\right]\nonumber\\
	&\quad+ \E\left[\left.\left(\sum_{i=3}^{n}\E[(y_{in}^{(1)})(y_{in}^{(3)})|\bs r_{-i},\mathcal{J}]-\sum^n_{i=3}\E[(y_{in}^{(1)})(y_{in}^{(3)})|\mathcal{J}]\right)^2\right|\mathcal{J}\right]\nonumber\\
	&\quad+ \left.\E\left[\left.\left(\sum_{i=3}^{n}\E[(y_{in}^{(2)})(y_{in}^{(3)})|\bs r_{-i},\mathcal{J}]-\sum^n_{i=3}\E[(y_{in}^{(2)})(y_{in}^{(3)})|\mathcal{J}]\right)^2\right|\mathcal{J}\right]\right\}.
\end{align}

Each of these terms converges to zero almost surely. Below we show how this follows for the cross product between the quadratic and the cubic how. For the other terms we do this in a separate document that is available upon request.

Define $\hat{\bs\Phi}=\bs D_{\sum_{h=1}^pc_{2n,h}\bar{Z}_{(l)}}\bs V\bs D_\varepsilon$ and $\hat{\bs A}_{-i}$ as $\bs A_{-i}$ but with $\hat{\bs\Phi}$ instead of $\bs\Phi$ and similarly for other variables that contain $\bs\Phi$. The product between the quadratic and cubic term in \eqref{eq:plim conditional variance} then is
\begin{align}\label{eq:quadratic cubic psi and p}
        &\E\bigg[\bigg(\sum_{i=3}^{n}\E[(y_{in}^{(2)})(y_{in}^{(3)})|\bs r_{-i},\mathcal{J}]-\sum^n_{i=3}\E((y_{in}^{(2)})(y_{in}^{(3)})|\mathcal{J})\bigg)^2\bigg|\mathcal{J}\bigg]\nonumber\\
        &\leq\frac{C}{n^2} \E\bigg[\bigg(\sum_{i=3}^{n}\bs r'\bs S_{i-1}[(\bs \Psi+\bs\Psi'-2\bs D_{\Psi}) - 2c_{1n}\gamma_n\dot{\bs P}]\bs e_{i}\bs r'\bs S_{i-1}\E(\bs A_{i}|\mathcal{J})\bs S_{i-1}\bs r\bigg)^2\bigg|\mathcal{J}\bigg]\nonumber\\
        &=\frac{C}{n^2} \E\bigg[\bigg(\sum_{i=3}^{n}\bs r'\bs S_{i-1}\left[\left(\bs \Psi+\bs\Psi'-2\bs D_{\Psi}\right) - 2c_{1n}\gamma_n\dot{\bs P}\right]\bs e_{i}\bs r'\bs S_{i-1}\hat{\bs A}_{i}\bs S_{i-1}\bs r\bigg)^2\bigg|\mathcal{J}\bigg]\nonumber\\
        &\leq\frac{C}{n^2}\bigg( \E\bigg[\bigg(\sum_{i=3}^{n}\bs r'\bs S_{i-1}\left(\bs \Psi+\bs\Psi'-2\bs D_{\Psi}\right)\bs e_{i}\bs r'\bs S_{i-1}\hat{\bs A}_{i}\bs S_{i-1}\bs r\bigg)^2\bigg|\mathcal{J}\bigg]\nonumber\\
        &\quad+ \E\bigg[\bigg(\sum_{i=3}^{n}\bs r'\bs S_{i-1}2c_{1n}\gamma_n\dot{\bs P}\bs e_{i}\bs r'\bs S_{i-1}\hat{\bs A}_{i}\bs S_{i-1}\bs r\bigg)^2\bigg|\mathcal{J}\bigg]\bigg),
\end{align}
due to the odd number of Rademacher random variables in $y_{in}^{(2)}y_{in}^{(3)}$. We only bound the second term. The first term follows by similar arguments and an eigenvalue bound on $\bs\Psi+\bs\Psi'-2\bs D_{\Psi}$. By completing the square we obtain
\begin{align}\label{eq:quadratic cubic}
        &\frac{C}{nk} \E[\sum_{i,j=3}^{n}\bs r'\bs S_{i-1}\dot{\bs P}\bs e_{i}\bs e_{j}'\dot{\bs P}\bs S_{j-1}\bs r\bs r'\hat{\bs A}_{-i}\bs r\bs r'\hat{\bs A}_{-j}\bs r|\mathcal{J}]\nonumber\\
        &=\frac{C}{nk}\sum_{i,j=3}^{n}\tr(\bs S_{i-1}\dot{\bs P}\bs e_{i}\bs e_{j}'\dot{\bs P}\bs S_{j-1}[20(\hat{\bs A}_{-i}\odot\hat{\bs A}_{-j})-6\bs I\odot(\hat{\bs A}_{-i}\hat{\bs A}_{-j}+\hat{\bs A}_{-j}\hat{\bs A}_{-i})\nonumber\\
        &\quad+4\hat{\bs A}_{-i}\hat{\bs A}_{-j}+4\hat{\bs A}_{-j}\hat{\bs A}_{-i}+\tr(\hat{\bs A}_{-i}\hat{\bs A}_{-j})\bs I])\nonumber\\
        &=\frac{C}{nk}\sum_{i,j=3}^{n}\bs e_{j}'\dot{\bs P}\bs S_{j-1}[20(\hat{\bs A}_{-i}\odot\hat{\bs A}_{-j})-6\bs I\odot(\hat{\bs A}_{-i}\hat{\bs A}_{-j}+\hat{\bs A}_{-j}\hat{\bs A}_{-i})+4\hat{\bs A}_{-i}\hat{\bs A}_{-j}\nonumber\\
        &\quad+4\hat{\bs A}_{-j}\hat{\bs A}_{-i}+\tr(\hat{\bs A}_{-i}\hat{\bs A}_{-j})\bs I]\bs S_{i-1}\dot{\bs P}\bs e_{i},
\end{align}
by Lemma \ref{thm:Rademacher} in Appendix \ref{sapp:rademacher}. For sake of space, we focus on the first term, which can be written as  
\begin{align}\label{eq:quadratic cubic 1}
        &\frac{C}{nk}\sum^{n}_{i,j=3}\bs e_j'\dot{\bs P}\bs S_{j-1}(\hat{\bs A}_{-i}\odot\hat{\bs A}_{-j})\bs S_{i-1}\dot{\bs P}\bs e_i\nonumber\\
        &=\frac{C}{nk}\sum^{n}_{i,j=3}\bs e_j'\dot{\bs P}\bs S_{j-1}\bs S_{i-1}(\hat{\bs A}_{i}\odot\hat{\bs A}_{j})\bs S_{i-1}\bs S_{j-1}\dot{\bs P}\bs e_i\nonumber\\
        &=\frac{C}{nk}[2\sum^{n}_{i=3}\sum_{j=i+1}^n\bs e_j'\dot{\bs P}\bs S_{i-1}(\hat{\bs A}_{i}\odot
        \hat{\bs A}_{j})\bs S_{i-1}\dot{\bs P}\bs e_i+\sum^{n}_{i=3}\bs e_i\dot{\bs P}\bs S_{i-1}(\hat{\bs A}_{i}\odot
        \hat{\bs A}_{i})\bs S_{i-1}\dot{\bs P}\bs e_i]\nonumber\\
        &=\frac{C}{nk}[2\sum^{n}_{i=3}\sum_{j=i+1}^n\bs e_j'\dot{\bs P}\bs S_{i-1}(\sum^{4}_{r=1}\hat{\bs A}_{i}^{(r)}\odot\sum^{4}_{s=1}\hat{\bs A}_{j}^{(s)})\bs S_{i-1}\dot{\bs P}\bs e_i\nonumber\\
        &\quad+\sum^{n}_{i=3}\bs e_i'\dot{\bs P}\bs S_{i-1}(\sum^{4}_{r=1}\hat{\bs A}_{i}^{(r)}\odot\sum^{4}_{s=1}\hat{\bs A}_{i}^{(s)})\bs S_{i-1}\dot{\bs P}\bs e_i],
\end{align}
where $\hat{\bs A}_i^{(r)}$ for $r=1,\ldots,4$ are the four terms between the $\bs S_{i-1}$ in $\bs A_{-i}$ from \eqref{eq:A short}, but with $\bs\Phi$ substituted by $\hat{\bs\Phi}$. Again, we only consider the first term, which consist of 16 cross products for the different $r$ and $s$. Let $\sum_{j=i+1}^n\bs e_j'\bs e_j=\bs I_n-\bs S_{i-1}-\bs e_i'\bs e_i=\tilde{\bs I}_{in}$. Then for $r=1,s=1$
\begin{align*}
        &\frac{C}{nk}\sum_{i=3}^n\sum_{j=i+1}^n\bs e_j'\dot{\bs P}\bs S_{i-1}(\hat{\bs A}_{i}^{(1)}\odot\hat{\bs A}_{j}^{(1)})\bs S_{i-1}\dot{\bs P}\bs e_i\\
        &=\frac{C}{nk}\sum_{i=3}^n\sum_{j=i+1}^n\bs e_j'\dot{\bs P}\bs S_{i-1}(\bs D_{\hat{\Phi}e_i}\dot{\bs P}\odot\bs D_{\hat{\Phi}e_j}\dot{\bs P})\bs S_{i-1}\dot{\bs P}\bs e_i\\
        &=\frac{C}{nk}\sum_{i=3}^n\sum_{j<i}\bs e_j'\dot{\bs P}\tilde{\bs I}_{in}\hat{\bs\Phi}'\bs e_j\bs e_j'\hat{\bs\Phi}\bs e_i\bs e_j'(\dot{\bs P}\odot\dot{\bs P})\bs S_{i-1}\dot{\bs P}\bs e_i\\
        &=\frac{C}{nk}\sum_{i=3}^n\bs e_i'\hat{\bs\Phi}'\bs D_{\dot{P}\tilde{I}_{in}\hat{\Phi}'}\bs S_{i-1}(\dot{\bs P}\odot\dot{\bs P})\bs S_{i-1}\dot{\bs P}\bs e_i\\
        &\leq\frac{C}{nk}\sum_{i=3}^n[\bs e_i'\hat{\bs\Phi}'\bs D_{\dot{P}\tilde{I}_{in}\hat{\Phi}'}^2\hat{\bs\Phi}\bs e_i\bs e_i'\dot{\bs P}\bs S_{i-1}(\dot{\bs P}\odot\dot{\bs P})\bs S_{i-1}(\dot{\bs P}\odot\dot{\bs P})\bs S_{i-1}\dot{\bs P}\bs e_i]^{\frac{1}{2}}\\
        &\leq\frac{C}{nk}\sum_{i=3}^n[\bs e_i'\bs P(\sum_{l=1}^p\bs D_{\bar{z}_{(l)}})\bs D_{\dot{V}\tilde{I}_{in}P(\sum_{l=1}^p D_{\bar{z}_{(l)}})}^2(\sum_{l=1}^p\bs D_{\bar{z}_{(l)}})\bs P\bs e_i\bs e_i'\dot{\bs P}^2\bs e_i]^{\frac{1}{2}},
\end{align*}
by the Cauchy--Schwarz inequality twice and because $\lambdamax((\dot{\bs P}\odot\dot{\bs P})(\dot{\bs P}\odot\dot{\bs P}))\leq C$. Since for a $n\times n$ matrix $\bs A$ we have $\lambdamax(\bs D_{A}^2)=\max_{j=1,\ldots,n}(\bs e_j'\bs A\bs e_j)^2= \max_{j=1,\ldots,n}\bs e_j'\bs A'\bs e_j\bs e_j'\bs A\bs e_j\leq\allowbreak\max_{j=1,\ldots,n}\allowbreak\bs e_j'\bs A'\bs A\bs e_j\leq \lambdamax(\bs A'\bs A)$ and $\lambdamax(\dot{\bs V})\leq C$ we have $\lambdamax(\bs D_{\dot{V}\tilde{I}_{in}P(\sum_{l=1}^p D_{\bar{z}_{(l)}})}^2)\leq C\lambdamax([\sum_{l=1}^p \bs D_{\bar{z}_{(l)}}]^2)=C \max_{i=1,\ldots,n}\Vert \bar{\bs z}_{i}\Vert^2$. Therefore the equation above becomes
\begin{align*}
        &\frac{C}{nk}\sum_{i=3}^n[\bs e_i'\bs P(\sum_{l=1}^p\bs D_{\bar{z}_{(l)}})\bs D_{\dot{V}\tilde{I}_{in}P(\sum_{h=1}^p D_{\bar{z}_{(l)}})}^2(\sum_{l=1}^p\bs D_{\bar{z}_{(l)}})\bs P\bs e_i\bs e_i'\dot{\bs P}^2\bs e_i]^{\frac{1}{2}}\\
        &\leq\frac{C}{nk}\sum_{i=3}^n[(\max_{j=1,\ldots,n}\Vert \bar{\bs z}_{j}\Vert^2)^2\bs e_i'\bs P\bs e_i\bs e_i'\dot{\bs P}^2\bs e_i]^{\frac{1}{2}}\leq\frac{C}{nk}\sum_{i=3}^n\max_{j=1,\ldots,n}\Vert \bar{\bs z}_{j}\Vert^2\bs e_i'\bs P\bs e_i\rightarrow_{a.s.}0,
\end{align*}
by Assumption \ref{ass:eigbound}. The other combinations of $r$ and $s$ can be shown to converge to zero using similar arguments. Continuing like this we can show that \eqref{eq:plim conditional variance} converges to zero almost surely.

\subsubsection{Unconditional distribution of \texorpdfstring{$\bs t'\bs\Sigma_n^{-1/2}\bs Y_n$}{tSigmaY} by Lebesgue's dominated convergence Theorem}\label{sapp:unconditional}
To obtain the unconditional distribution, note that for some $\epsilon>0$, say $\epsilon=1$, we have
\begin{align*}
		&\sup_n\E([|\Pr((\bs\alpha'\bs\alpha)^{-1/2}\bs\alpha'\bs\Sigma_{n}^{-1/2}\bs Y_{nr}<y|\mathcal{J})|^{1+\epsilon}]
		\\
		&= \sup_n\E[(\Pr((\bs\alpha'\bs\alpha)^{-1/2}\bs\alpha'\bs\Sigma_{n}^{-1/2}\bs Y_{nr}<y|\mathcal{J}))^{2}]\leq\sup_n\E[1^2]
		\leq\infty.
\end{align*}
Therefore, $\Pr((\bs\alpha'\bs\alpha)^{-1/2}\bs\alpha'\bs\Sigma_{n}^{-1/2}\bs Y_{nr}<y|\mathcal{J})$ is uniformly integrable \citep[p.\ 338]{billingsley1995probability} and we can apply a version of Lebesgue's dominated convergence theorem \citep[Theorem 25.12]{billingsley1995probability}
\begin{align*}
		&\Pr((\bs\alpha'\bs\alpha)^{-1/2}\bs\alpha'\bs\Sigma_{n}^{-1/2}\bs Y_{n}<y)=\E[\Pr((\bs\alpha'\bs\alpha)^{-1/2}\bs\alpha'\bs\Sigma_{n}^{-1/2}\bs Y_{n}<y|\mathcal{J})]\\&=\E[\Pr((\bs\alpha'\bs\alpha)^{-1/2}\bs\alpha'\bs\Sigma_{n}^{-1/2}\bs Y_{nr}<y|\mathcal{J})]\rightarrow_{a.s.}\E[\Phi(y)]=\Phi(y).
\end{align*}

\subsubsection{Distribution of \texorpdfstring{$\bs Y_n$}{Y} by the Cramér--Wold theorem}\label{sapp:Cramer-Wold}
We have shown that for any $\bs\alpha$ we have $(\bs\alpha'\bs\alpha)^{-1/2}\bs\alpha'\bs\Sigma_{n}^{-1/2}\bs Y_{n}\rightarrow_d(\bs\alpha'\bs\alpha)^{-1/2}\bs\alpha'\bs Z$, with $\bs Z\sim N(\bs 0,\bs I_{p+1})$. Then also $C(\bs\alpha'\bs\alpha)^{-1/2}\bs\alpha'\bs\Sigma_{n}^{-1/2}\bs Y_{n}\rightarrow_dC(\bs\alpha'\bs\alpha)^{-1/2}\bs\alpha'\bs Z$ and by the Cramér--Wold theorem \citep[T29.4]{billingsley1995probability}  $\bs\Sigma_{n}^{-1/2}\bs Y_n\rightarrow_d\bs Z$.

\end{proof}

\section{Control variables in linear IV}\label{app:controls}
In Example \ref{ex:linIV} we wrote the linear IV model without exogenous control variables. In this section we show that if the number of control variables is not too large and invariance of the moment conditions stems from symmetry in the second stage errors, as in Section \ref{sec:score}, then including the exogenous control variables does change the main results.

Consider the linear IV model from \eqref{eq:model1} with exogenous control variables $\bs F$
\begin{align*}
        \bs y=\bs X\bs\beta+\bs F\bs\Gamma_1+\bs\varepsilon\nonumber\\
        \bs X=\bs Z\bs\Pi+\bs F\bs\Gamma_2+\bs\eta.
\end{align*}
Assume that the control variables are such that the diagonal elements of $\bs P_{F}$ converge to zero fast enough. In particular, $\max_{i=1,\dots,n}P_{F,ii}=O(n^{-(1/2+\delta)})$ $a.s.n.$ for some $\delta>1/6$.

We can partial out the exogenous control variables by pre-multiplying both equations with $\bs M_{F}$. Take $\tilde{\bs Z}=\bs M_{F}\bs Z$ and write the CU objective function as 
\begin{align*}
    F(\bs\beta)&=\frac{1}{n}\bs\varepsilon'\bs M_{F}\bs Z(\bs Z'\bs M_{F}\bs D_{M_{F}\varepsilon}^2\bs M_{F}\bs Z)^{-1}\bs Z'\bs M_{F}\bs\varepsilon=\frac{1}{n}\bs\varepsilon'\tilde{\bs Z}(\tilde{\bs Z}'\bs D_{M_{F}\varepsilon}^2\tilde{\bs Z})^{-1}\tilde{\bs Z}'\bs\varepsilon.
\end{align*}
We will now show that $\bs\varepsilon'\tilde{\bs Z}(\tilde{\bs Z}'\bs D_{M_{F}\varepsilon}^2\tilde{\bs Z})^{-1}\tilde{\bs Z})^{-1}\tilde{\bs Z}'\bs\varepsilon\rightarrow_{a.s.}\bs\varepsilon'\tilde{\bs Z}(\tilde{\bs Z}'\bs D_{\varepsilon}^2\tilde{\bs Z})^{-1}\tilde{\bs Z})^{-1}\tilde{\bs Z}'\bs\varepsilon$, such that partialling out the exogenous controls does not obstruct the derivation of our main results, once we take $\tilde{\bs Z}$ as the new instruments. 

Note that $\tilde{\bs Z}'\bs D_{M_{F}\varepsilon}^2\tilde{\bs Z}=\tilde{\bs Z}'\bs D_{\varepsilon}^2\tilde{\bs Z}+\tilde{\bs Z}'\bs D\tilde{\bs Z}$, with $D_{ii}=-2\varepsilon_{i}\bs e_{i}'\bs P_{F}\bs\varepsilon+(\bs e_{i}'\bs P_{F}\bs\varepsilon)^2$. We can write
\begin{align*}
    |D_{ii}|
    &\leq 2(\varepsilon_{i}^2\bs\varepsilon'\bs P_{F}\bs e_{i}\bs e_{i}'\bs P_{F}\bs\varepsilon)^{1/2}+\bs\varepsilon'\bs P_{F}\bs e_{i}\bs e_{i}'\bs P_{F}\bs\varepsilon\leq2(\varepsilon_{i}^2P_{F,ii}\bs\varepsilon'\bs P_{F}\bs\varepsilon)^{1/2}+P_{F,ii}\bs\varepsilon'\bs P_{F}\bs\varepsilon,
\end{align*}
such that if $P_{F,ii}\bs\varepsilon'\bs P_{F}\bs\varepsilon\rightarrow_{a.s.}0$ we have $D_{ii}\rightarrow_{a.s.}0$.
\begin{align}\label{eq:D}
    P_{F,ii}\bs\varepsilon'\bs P_{F}\bs\varepsilon=P_{F,ii}\sum_{j=1}^nP_{F,jj}(\varepsilon_{j}^2-\sigma_{j}^2)+P_{F,ii}\sum_{j=1}^nP_{F,jj}\sigma_{j}^2+2P_{F,ii}\sum_{j=1}^{n-1}\sum_{i_{1}<j}\varepsilon_{j}P_{F,ji_{1}}\varepsilon_{i_{1}}.
\end{align}
The first term converges to zero almost surely by the assumption on the diagonal of $\bs P_{F}$ and the finite fourth moments of $\varepsilon_i$. The second term converges to zero almost surely by the same assumption on the diagonal of $\bs P_{F}$ and since we assume $\sigma^2_{i}<C$ for all $i$. The final term is shown to converge to zero almost surely in the next subsection.

Then we can write
\begin{align*}
    \tilde{\bs Z}'\bs D_{M_{F}\varepsilon}^2\tilde{\bs Z}&=(\tilde{\bs Z}'\bs D_{\varepsilon}^2\tilde{\bs Z})^{1/2}[\bs I+(\tilde{\bs Z}'\bs D_{\varepsilon}^2\tilde{\bs Z})^{-1/2}\tilde{\bs Z}\bs D\tilde{\bs Z}(\tilde{\bs Z}'\bs D_{\varepsilon}^2\tilde{\bs Z})^{-1/2}](\tilde{\bs Z}'\bs D_{\varepsilon}^2\tilde{\bs Z})^{1/2}\\
    &=(\tilde{\bs Z}'\bs D_{\varepsilon}^2\tilde{\bs Z})^{1/2}\bs U[\bs I+\bs\Lambda]\bs U'(\tilde{\bs Z}'\bs D_{\varepsilon}^2\tilde{\bs Z})^{1/2},
\end{align*}
where $\bs U\bs\Lambda\bs U'=(\tilde{\bs Z}'\bs D_{\varepsilon}^2\tilde{\bs Z}/n)^{-1/2}\tilde{\bs Z}\bs D\tilde{\bs Z}/n(\tilde{\bs Z}'\bs D_{\varepsilon}^2\tilde{\bs Z}/n)^{-1/2}$ the eigenvalue decomposition. Then since $D_{ii}\rightarrow_{a.s.}0$ we have $\Lambda_{ii}\rightarrow_{a.s.}0$ by $1/C\leq\lambdamin(\tilde{\bs Z}'\bs D_{\varepsilon}^2\tilde{\bs Z}/n)$ and $\lambdamax(\tilde{\bs Z}'\tilde{\bs Z}/n)\leq C$.

Next, write
\begin{align*}
    (\tilde{\bs Z}'\bs D_{M_F\varepsilon}\tilde{\bs Z})^{-1}&=(\tilde{\bs Z}'\bs D_{\varepsilon}^2\tilde{\bs Z})^{-1/2}\bs U[\bs I+\bs\Lambda]^{-1}\bs U'(\tilde{\bs Z}'\bs D_{\varepsilon}^2\tilde{\bs Z})^{-1/2}\\
    &=(\tilde{\bs Z}'\bs D_{\varepsilon}^2\tilde{\bs Z})^{-1}-(\tilde{\bs Z}'\bs D_{\varepsilon}^2\tilde{\bs Z})^{-1/2}\bs U\bs\Lambda[\bs I+\bs\Lambda]^{-1}\bs U'(\tilde{\bs Z}'\bs D_{\varepsilon}^2\tilde{\bs Z})^{-1/2}.
\end{align*}
Now take $\bs v$ and $\bs w$ such that their quadratic forms with respect to $(\tilde{\bs Z}'\bs D_{\varepsilon}^2\tilde{\bs Z})^{-1}$ are $O(1)$. Then the bilinear form of $\bs v$ and $\bs w$ with the second term in the equation above can be shown to converge to zero almost surely, by applying the Cauchy-Schwarz inequality and using that $\lambdamax(\bs\Lambda)\rightarrow_{a.s.}0$ and $\lambdamax([\bs I+\bs\Lambda]^{-1})\rightarrow_{a.s.}1$.

\subsection{Convergence of the martingale difference array}
We now show that the final term in \eqref{eq:D} converges completely to zero, and hence almost surely to zero, using Theorem 2 in \citet{ghosal1998complete}.

For notational convenience, let the following be conditional on $\bs Z$ and $\bs F$. Also, write $\bs F_{n}$ to emphasize the dependence of $\bs F$ on $n$. To match the notation in \citet{ghosal1998complete} define
\begin{align*}
    X_{n,k}&=P_{F_n,ii}\sum_{j<k}\varepsilon_{j}P_{F_n,jk}\varepsilon_{k} \text{ for $k<n$ and $0$ otherwise};\\
    \mathcal{F}_{n,k}&=\{\varepsilon_{i}:i\leq k\}.
\end{align*}

We first need that $\{(X_{n,k},\mathcal{F}_{n,k}),k\geq 1\}$ are sequences of square-integrable martingale differences. $\E(X_{n,k}|\mathcal{F}_{n,k-1})=0$ by independence and mean zero of $\varepsilon_i$, hence it is a martingale difference sequence. Square-integrability follows because
\begin{align*}
    \E(X_{n,k}^2)&=\E([P_{F_n,ii}\sum_{j<k}\varepsilon_{j}P_{F_n,jk}\varepsilon_{k}]^2)=\E(P_{F_n,ii}^2\sum_{j<k}\varepsilon_{j}^2P_{F_n,jk}^2\varepsilon_{k}^2)\leq C P_{F_n,ii}^2\sum_{j<k}P_{F_n,jk}^2\\
    &\leq C P_{F_n,ii}^2P_{F_n,kk}<\infty.
\end{align*}

Next, we require constants $\{M_n\}$ such that $\sum_{k=1}^\infty\E(X_{n,k}^2|\mathcal{F}_{n,k-1})\leq M_{n}$ $a.s.$, where $\mathcal{F}_{n,0}$ is the trivial $\sigma$-field. Let $\bs S_{k}$ be the matrix with ones on the first $k$ diagonal elements and zeros elsewhere. We have
\begin{align*}
    \sum_{k=1}^\infty\E(X_{n,k}^2|\mathcal{F}_{n,k-1})
    &=\sum_{k=1}^n\E(P_{F_n,ii}^2\varepsilon_{k}^2[\sum_{j<k}P_{F_n,jk}\varepsilon_{j}]^2|\mathcal{F}_{n,k-1})\\
    &=\sum_{k=1}^n P_{F_n,ii}^2\sigma_{k}^2\bs\varepsilon'\bs S_{k-1}\bs P_{F_n}\bs e_{k}\bs e_{k}'\bs P_{F_n}\bs S_{k-1}\bs\varepsilon\\
    &\leq\sum_{k=1}^n P_{F_n,ii}^2\sigma_{k}^2P_{F_n,kk}\bs\varepsilon'\bs\varepsilon\\
    &\leq C\sum_{k=1}^n nP_{F_n,ii}^2P_{F_n,kk}\bs\varepsilon'\bs\varepsilon/n\leq Cn^2\max_{i=1,\dots,n}P_{F_n,ii}^3,
\end{align*}
since by a strong law of large numbers $\bs\varepsilon'\bs\varepsilon/n\leq C$ $a.s.n.$ Then, by the assumption that $\max_{i=1,\dots,n}P_{F_n,ii}=O(n^{-(1/2+\delta)})$ $a.s.n.$ with $\delta>1/6$ we can find $\{M_{n}\}=\{Cn^{1/2-3\delta}\}$ for which the condition holds. Moreover, for such $\{M_n\}$, $c_n=1$ and $\lambda>1/(1/2-3\delta)$ we have that $\sum_{n=1}^\infty c_nM_{n}^\lambda<\infty$, as needed.

Finally, we require that for all $\epsilon>0$ that $\sum_{n=1}^\infty c_n\sum_{k=1}^\infty\Pr(|X_{n,k}|>\epsilon)<\infty$. Let $\epsilon>0$ be given and take $c_n=1$. Then by Markov's inequality
\begin{align*}
    \sum_{n=1}^\infty c_n\sum_{k=1}^\infty\Pr(|X_{n,k}|>\epsilon)&\leq C\sum_{n=1}^\infty\sum_{k=1}^\infty\E([P_{F_n,ii}\sum_{j<k}\varepsilon_{j}P_{F_n,jk}\varepsilon_{k}]^2)\\
    &\leq C\sum_{n=1}^\infty\sum_{k=1}^\infty\E(P_{F_n,ii}^2\sum_{j<k}\varepsilon_{j}^2P_{F_n,jk}^2\varepsilon_{k}^2)\\
    &\leq C\sum_{n=1}^\infty\sum_{k=1}^\infty P_{F_n,ii}^2\sum_{j<k}P_{F_n,jk}^2\\
    &\leq C\sum_{n=1}^\infty\sum_{k=1}^\infty P_{F_n,ii}^2P_{F_n,kk}^2<\infty,
\end{align*}
because both infinite sums converge by the condition on $P_{F_n,ii}$.

By Theorem 2 in \citet{ghosal1998complete}, we have complete convergence of $\sum_{k=1}^\infty X_{nk}$ to zero, which implies almost sure convergence.

\end{appendices}
\end{document}